\documentclass[11pt,a4paper]{article}

\usepackage[T1]{fontenc}
\usepackage[utf8]{inputenc}
\usepackage[
  top=2.5cm,
  bottom=2.5cm,
  left=2.5cm,
  right=2.5cm
]{geometry}

\usepackage{amsmath,amssymb,amsfonts,mathrsfs,amsthm,mathtools}
\usepackage{graphicx,xcolor,booktabs,array,longtable,float,ragged2e}
\usepackage[colorlinks=true]{hyperref}

\definecolor{ABJMlink}{HTML}{176C68}
\definecolor{ABJMcite}{HTML}{8A4F6B}

\hypersetup{
  colorlinks=true,
  linkcolor=ABJMlink,
  urlcolor=ABJMlink,
  citecolor=ABJMcite
}
\usepackage{tikz}
\usepackage{adjustbox}
\usepackage{caption}
\usetikzlibrary{arrows.meta,positioning,calc,fit,backgrounds,shapes.geometric}
\usepackage{orcidlink}
\usepackage{cite}
\usepackage{enumitem}
\usepackage{dsfont}
\usepackage{upgreek}

\theoremstyle{plain}
\newtheorem{theorem}{Theorem}[section]
\newtheorem{proposition}[theorem]{Proposition}
\newtheorem{lemma}[theorem]{Lemma}

\theoremstyle{definition}

\numberwithin{equation}{section}

\DeclareMathOperator{\Tr}{Tr}
\DeclareMathOperator{\re}{Re}
\DeclareMathOperator{\im}{Im}
\DeclareMathOperator{\Li}{Li}

\newcommand{\ii}{\mathrm{i}}
\newcommand{\e}{\mathrm{e}}
\newcommand{\rd}{\mathrm{d}}
\newcommand{\cW}{\mathcal{W}}
\newcommand{\cE}{\mathcal{E}}
\newcommand{\cL}{\mathcal{L}}
\newcommand{\cN}{\mathcal{N}}

\newcommand{\cR}{\mathcal{R}}
\newcommand{\cU}{\mathcal{U}}
\newcommand{\Wt}{\widetilde{\mathcal{W}}}
\newcommand{\Nh}{\widehat{N}}
\newcommand{\lamh}{\widehat{\lambda}}
\newcommand{\CW}{C}
\newcommand{\Cint}{C_{\rm int}}
\newcommand{\Cltwo}{\mathrm{Cl}_2}
\newcommand{\Hgen}{\mathcal{H}}
\newcommand{\sthree}{S^3}
\newcommand{\fzero}{\widehat{f}_0}
\newcommand{\U}{\mathrm{U}}
\newcommand{\bea}{\begin{equation}\begin{aligned}}
\newcommand{\eea}{\end{aligned}\end{equation}}
\newcommand{\be}{\begin{equation}}
\newcommand{\ee}{\end{equation}}
\newcommand{\wt}{\widetilde}

\newcommand{\cC}{\mathcal{C}}

\newcommand{\cI}{\mathcal{I}}
\newcommand{\cK}{\mathcal{K}}

\newcommand{\Eul}{\mathrm{E}}

\title{%
  \vspace{-1.5cm}
  \bfseries
  Constant Maps and Exponentially Small Sectors\\
  in ABJM Bethe Observables
}

\author{%
  Seyed Morteza Hosseini\,\orcidlink{0000-0001-8205-400X}\\[0.5em]
  \small Centre for Theoretical Physics, Department of Physics and Astronomy,\\
  \small Queen Mary University of London,
  London E1 4NS, United Kingdom\\[0.5em]
  \small School of Mathematics and Maxwell Institute for Mathematical Sciences,\\
  \small University of Edinburgh, Edinburgh EH9 3FD, UK\\[0.5em]
}

\date{\today}

\begin{document}
\maketitle

\thispagestyle{empty}
 
\begin{abstract}
\noindent The topologically twisted index of ABJM theory captures the microscopic entropy of magnetically charged supersymmetric $\mathrm{AdS}_4$ black holes, while the effective twisted superpotential determines the Bethe vacua entering its exact evaluation. At large $N$, the critical value of the twisted superpotential supplies the field-theory input to the gravitational blocks entering the bulk entropy function. The companion Letter~\cite{Hosseini:2026dkj} reconstructed the finite-rank structure of the on-shell twisted superpotential for a selected vacuum at the universal twist, including its shifted-rank perturbative form, constant map, all-genus type-IIA expansion, and exponentially small sectors. Using high-precision Bethe-vacuum data, physics-informed symbolic regression, and integer-relation detection, we develop the corresponding finite-rank structure of the index, which also depends on the Hessian, effective dilaton, and remaining one-loop factors and therefore contains information beyond the critical value of the twisted superpotential.

The two observables share the same shifted rank. Within a declared finite class, we reconstruct the index constant map from the ABJM $S^3$ constant maps at levels $k$ and $k/2$, the constant map of the twisted superpotential, and elementary contributions. We prove that this expression equals a convergent one-kernel integral for every real $k>0$. The integral resums the known large-$k$ expansion, determines in closed form a constant previously obtained numerically, and yields the type-IIA coefficients at arbitrary genus. We also prove that the finite Clausen expression for the twisted-superpotential constant map found in~\cite{Hosseini:2026dkj} is equivalent to an alternating Euler sum and an absolutely convergent integral, obtaining closed evaluations of two families of rational-argument Euler sums.

At $k=1,2,4$, the exponentially suppressed sectors exhibit a striking arithmetic organization that is not apparent in the defining Bethe-vacua formulas. The reconstructed coefficients of the twisted superpotential obey a divisor-sum formula whose conjectural all-order continuation is an Eichler-integral series associated with a weight-four Eisenstein form. The index coefficients reduce exactly over the available orders to two rational generators. For one generator, closed modular-product formulas reproduce every reconstructed coefficient exactly and furnish a concrete conjectural all-order completion of this sector. Conditional on this continuation, modular transformations determine the nearest logarithmic singularity, the radius of convergence, and the leading coefficient growth. The emergent arithmetic structure sets a stringent finite-$N$ benchmark for a quantum completion of the gravitational-block framework and, in the present setting, for a bulk computation of the quantum entropy of magnetically charged supersymmetric $\mathrm{AdS}_4$ black holes.
\end{abstract}

\clearpage
\setcounter{page}{1}

\noindent\rule{\textwidth}{0.6pt}
\vspace{-0.5cm}

\begingroup
\hypersetup{
  linktoc=page,
  linkcolor=ABJMlink
}
{\small\tableofcontents}
\endgroup

\vspace{0.2cm}
\noindent\rule{\textwidth}{0.6pt}

\section{Introduction}
\label{sec:intro}

Recovering the full large-$N$ content of an exact finite-$N$ observable is a
nontrivial asymptotic problem. Supersymmetric localization brings this
problem into particularly sharp focus by reducing a broad class of protected
observables to exact finite-dimensional representations, including matrix
integrals and Bethe-vacuum sums (see~\cite{Pestun:2016zxk} for a review).
Determining the $N$-dependent perturbative expansion to all orders can
already require substantial analytic or numerical work. Even complete
knowledge of this $N$-dependent perturbative sector does not fix the
rank-independent contribution, let alone sectors that are exponentially
small at large $N$. In a holographic theory, resolving this finite-$N$
structure furnishes boundary data that probe the holographic dictionary
beyond the classical large-$N$ limit. Extracting this finite-rank structure
is the problem addressed in this paper.

Three-dimensional $\cN=2$ Chern--Simons-matter
theories provide a particularly useful setting for
this kind of problem, because several of their partition functions are
organized by the same finite set of massive Bethe vacua. The effective
twisted superpotential $\Wt(u,\Delta)$, as a function of the gauge variables
$u$ and flavor chemical potentials $\Delta$, determines the Bethe equations
and controls the Cardy limit of each holomorphic block
\cite{Pasquetti:2011fj,Beem:2012mb}. The topologically twisted index on
$S^1\times\Sigma_{\mathfrak g}$, with $\Sigma_{\mathfrak g}$ a closed Riemann
surface of genus $\mathfrak g$, is evaluated on the same critical points,
while its Bethe-vacua formula contains additional data from the Hessian of
$\Wt$, the effective dilaton, and the remaining one-loop factors
\cite{Nekrasov:2009uh,Benini:2015noa,Closset:2018ghr}. The on-shell twisted
superpotential and the index therefore probe the same Bethe saddle from
complementary directions. The former records its critical value, while the
latter retains additional fluctuation data around that critical point.

For ABJM theory and other holographic M2-brane theories, this distinction has
a direct gravitational counterpart. At large $N$, the on-shell twisted
superpotential is identified, up to conventions, with the
supergravity prepotential entering a gravitational block, and gluing the blocks produces the entropy
function of rotating, supersymmetric $\mathrm{AdS}_4$ black
holes~\cite{Hosseini:2019iad}.
The topologically twisted index accounts at
leading order for the entropy of magnetically charged supersymmetric $\mathrm{AdS}_4$ black
holes~\cite{Benini:2015eyy,Benini:2016rke}. Both statements hold at leading
order in $N$, and their finite-rank completions are what supply the boundary
data against which quantum corrections to the entropy functions, and the
related Cardy limits of the dual field theory, can eventually be compared
\cite{Choi:2019dfu,Hosseini:2022vho,Bobev:2024mqw}.

We study these questions in $\U(N)_k\times\U(N)_{-k}$ ABJM theory, which
describes the low-energy dynamics of $N$ M2-branes probing
$\mathbb C^4/\mathbb Z_k$ \cite{Aharony:2008ug}, at the universal twist.
Previous high-precision analyses of the Bethe equations reconstructed the
$N$-dependent part of the fixed-$k$ perturbative large-$N$ expansion to all
orders, with the expansion organized in terms of a shifted rank
\cite{Bobev:2022jte,Bobev:2022eus}. Two sectors remained outside this
reconstruction. The first is the rank-independent term at order $N^0$. At the
universal twist, its large-$k$ asymptotic expansion was extracted numerically
through several orders. All displayed coefficients except the $k^0$
coefficient were recognized in closed form, while the latter remained known
only as a high-precision numerical constant \cite{Bobev:2022eus}. The second
is a tower of corrections exponentially small in $\sqrt N$. Numerical fits
identified the leading exponential scale and found it to coincide with that
of the $S^3$ partition function, while its coefficient and the remaining
nonperturbative structure were left undetermined
\cite{Bobev:2022eus}.

On the gravity side, four-derivative supergravity reproduces the first
subleading $N^{1/2}$ correction at the universal twist
\cite{Bobev:2021oku}. Equivariant localization organizes supersymmetric
on-shell actions in terms of fixed-point data
\cite{BenettiGenolini:2023kxp}, and a recent quantum-M-theory construction,
assuming the all-orders validity of the quantum-corrected localization
formula, captures the all-orders $N$-dependent perturbative large-$N$
structure while leaving the overall prefactor and nonperturbative corrections
undetermined \cite{BenettiGenolini:2026cyc}. The bulk and boundary analyses
therefore leave open the same two classes of finite-$N$ data. Determining
these missing finite-rank and nonperturbative contributions is the central
problem of this paper.

The twisted-superpotential side was developed in the companion
Letter~\cite{Hosseini:2026dkj}. There we reconstructed the shifted-rank parent
and constant map, derived the type-IIA coefficients at arbitrary genus, and
found at $k=1,2,4$ a divisor-sum structure governing the exponentially small
sectors, whose proposed all-order continuation has an Eisenstein--Eichler
form. The topologically twisted index
is built from the same Bethe saddle, but its Bethe-vacua formula carries more
than the on-shell value of $\Wt$. The Hessian of $\Wt$, the effective
dilaton, and the remaining one-loop factors all contribute, and none of them
is fixed by the twisted-superpotential result alone. The index is a
genuinely new finite-rank problem for this reason, with its own constant map
and its own exponential sectors to reconstruct, and not merely a corollary of
the companion analysis. The present paper supplies the root-of-unity proof of
the twisted-superpotential constant-map identities deferred there and
develops this larger index problem, asking how much of the
twisted-superpotential structure persists and what new information the index
carries.

The ABJM partition function on the round $S^3$ provides a closely related
exact problem for which the finite-rank structure is exceptionally well developed.
Supersymmetric localization reduces the partition function to a matrix model
\cite{Kapustin:2009kz}, whose 't~Hooft expansion is governed by the topological
string on local $\mathbb{P}^1\times\mathbb{P}^1$
\cite{Marino:2009jd,Drukker:2010nc}. The perturbative all-genus expansion
resums to an Airy function \cite{Fuji:2011km}. In the fixed-$k$ M-theory
regime, the Fermi-gas reformulation recovers this Airy structure and recasts
the grand partition function as a spectral determinant, thereby providing
direct access to exact finite-$N$ data \cite{Marino:2011eh}. In this
formulation, the modified grand potential contains the constant-map
contribution $A(k)$, together with worldsheet- and membrane-instanton
corrections \cite{Hatsuda:2012dt,Hatsuda:2013oxa}. The chemical potential
appearing in this description is the grand-canonical variable conjugate to
the rank. Mixed worldsheet--membrane bound-state effects can moreover be
absorbed into an effective chemical potential, so they need not be treated
as an independent sector \cite{Hatsuda:2013gj}. Modular structures are also
present on the sphere side. The topological string on local
$\mathbb{P}^1\times\mathbb{P}^1$ admits an almost-modular description, while
at the maximally supersymmetric levels the exact grand partition function
can be written in terms of Jacobi theta functions associated with the
spectral curve \cite{Aganagic:2006wq,Codesido:2014oua}; more generally, the
spectral-theory formulation leads to quantum theta functions with controlled
modular properties \cite{Grassi:2016nnt}.

In the present analysis, the sphere result enters directly through the
functions $A(k)$ and $A(k/2)$, which are prescribed features in the
reconstruction of the index constant term. The modular structures found below
have a different provenance. They emerge from the reconstructed Bethe-vacuum
coefficients and are not imported from the sphere spectral problem.
Section~\ref{sec:discussion} compares the nonperturbative actions,
coefficient structures, and special-level modular behavior in the two
settings.

We denote by $\Delta_a$ the chemical potentials conjugate to the flavor
charges and by $\mathfrak n_a$ the corresponding background magnetic fluxes
through $\Sigma_{\mathfrak g}$, with $a=1,\ldots,4$. Our analysis uses
high-precision data for a selected Bethe vacuum together with its full
$k$-fold $\mathbb Z_k$ orbit. The Bethe equations are invariant under
simultaneous multiplication of all exponentiated Bethe variables by a $k$-th
root of unity, so that the selected vacuum generates an orbit of $k$ distinct
solutions; the
derivation is given in Section~\ref{sec:setup}. We specialize to
$S^1\times S^2$ at the universal values $\Delta_a=\pi/2$ and
$\mathfrak n_a=1/2$. We denote the corresponding on-shell twisted
superpotential by $\Wt(N,k)$ and find that its real part vanishes to the full
working precision throughout the computed grid. For the index, we consider
$F(N,k)\equiv-\re\log\left[k Z_{\rm data}(N,k)\right]$,
the real free-energy contribution associated with this orbit, where
$Z_{\rm data}(N,k)$ denotes the recorded contribution of a single Bethe
vacuum and the explicit factor of $k$ accounts for the $k$ equal
contributions in its $\mathbb Z_k$ orbit. We combine these
high-precision Bethe-vacuum data with physics-informed symbolic regression
and integer-relation detection to reconstruct the finite-rank structure and
identify the terms not fixed by the perturbative large-$N$ expansion.

\subsection*{Results and their scope}

Both observables are organized by the same shifted rank, instanton action, and
physical nome,
\be
\Nh=N-\frac{k}{24}+\frac{2}{3k}\,,
\qquad
t=2\pi\sqrt{\frac{2\Nh}{k}}\,,
\qquad
q=\e^{-t}\,.
\label{eq:intro-variables}
\ee
Their fixed-$k$ decompositions are
\bea
\im\Wt(N,k)
&=\frac{\pi^2}{3}\sqrt{\frac{k}{2}}\,\Nh^{3/2}
  +\CW(k)+\cW_{\rm np}(N,k)\,,
\\
F(N,k)
&=\frac{\pi\sqrt{2k}}{3}
  \left(\Nh^{3/2}-\frac{3}{k}\Nh^{1/2}\right)
  +\frac12\log\Nh-\fzero(k)+F_{\rm np}(N,k)\,.
\label{eq:intro-parents}
\eea
The remainders are exponentially small in $t$ over the available rank ranges.
We use ``constant map'' for the rank-independent terms $\CW(k)$ and
$\fzero(k)$ after subtraction of the displayed rank-dependent parents.

The companion analysis reconstructed $\CW(k)$ in parity-dependent finite
Clausen form. For every real $k>0$, define
\be
\Cint(k)
\equiv
-\frac{k^2\zeta(3)}{16\pi}
+\frac{2k}{\pi}\int_0^\infty
\log(2\cosh x)\,\log\bigl(1-\e^{-kx}\bigr)\,\rd x\,.
\label{eq:intro-CW-integral}
\ee
The integral converges absolutely, and Theorem~\ref{thm:C} proves that it
equals $\CW(k)$ for every positive integer $k$.
The proof reduces the integral to alternating
Euler sums at rational arguments and evaluates them by roots of unity. It
therefore yields, independently of the Bethe problem, closed finite formulas
for two families of rational-argument Euler sums.

The index constant map is reconstructed in a declared six-dimensional class
built from the ABJM $S^3$ constant maps at levels $k$ and $k/2$, the
twisted-superpotential constant $\CW(k)$, and elementary terms.
Theorem~\ref{thm:f0} proves that the resulting expression has the one-kernel
representation
\bea
\fzero(k)
&=
-\frac{3\zeta(3)}{8\pi^2}k^2
+\frac76\log k+f_0
+\frac{k}{\pi^2}\int_0^\infty
\cK(x)\log\bigl(1-\e^{-kx}\bigr)\,\rd x\,,
\\
\cK(x)
&=6x\tanh x+4\log\frac{\sinh x}{x}
  +4\left(\frac{x}{\tanh x}-1\right),
\qquad
f_0=-8\zeta'(-1)-\frac{23}{6}\log2-\frac23\log\pi\,.
\label{eq:intro-f0-integral}
\eea
This identity holds for every real $k>0$. The integral resums the previously known
large-$k$ series for $\fzero(k)$ and fixes its numerically fitted constant in
closed form.

Both theorems establish identities between analytic representations of
functions of $k$, independently of the Bethe-vacuum data. The role of the
data is instead to identify these analytic functions with the physical
plateaux extracted from the finite-rank observables. For $\CW(k)$, this
identification is verified over the full range $1\leq k\leq50$, although
only the seven lowest levels are used to determine the closed form. For
$\fzero(k)$, the six-dimensional class is fixed using six levels and then
tested out of sample against nine additional levels. Together
with the two shifted-rank parents, the constant-map integrals then give
closed type-IIA coefficients at arbitrary genus for both observables.
The corresponding genus expansions are asymptotic.

The exponentially small sector of the twisted superpotential is recalled from
the companion analysis because it supplies the comparison point for the index.
At $k=1,2,4$, every resolved coefficient agrees exactly with
\bea
\cW_{\rm np}(N,k)
&=
\frac{k^2}{2\pi}\sum_{m\geq1}c_m^{(k)}
\left(t+\frac1m\right)q^m\,,
\\
c_m^{(k)}
&=
\frac{(-1)^{m+1}}{m^2}
\begin{cases}
5\bigl[\sigma_3(m)-4\sigma_3(m/2)\bigr],&k=1,2,\\[3pt]
\sigma_3(m)-16\sigma_3(m/4),&k=4,
\end{cases}
\label{eq:intro-W-continuation}
\eea
with $\sigma_3(x)=0$ for $x\notin\mathbb Z_{>0}$. The coefficient law is
reconstructed through $51$, $67$, and $93$ sectors at $k=1,2,4$, and
conjectured to continue to all $m$ in~\eqref{eq:intro-W-continuation}. Assuming
this all-order continuation, the resulting $q$-series has unit radius of
convergence and can be expressed as the image, under a first-order
differential operator in $t$, of an Eichler integral associated with a
weight-four Eisenstein form. Importantly, the divisor-sum coefficients are
read directly off the sector peel; the Eisenstein--Eichler structure is
identified only afterward and is not used as input to the reconstruction.

This reconstruction has an immediate physical corollary. Evaluated on the
same Bethe saddle, $\im\Wt(N,k)$ is exactly the on-shell quantity that
controls the leading term of the superconformal index in its Cardy, or
high-temperature, limit~\cite{Choi:2019dfu,Bobev:2022wem,Bobev:2024mqw}.
The finite-rank structure
reconstructed above therefore determines the strict Cardy coefficient
associated with this saddle. In particular, its constant map
and exponentially small sectors are inherited, term by term, by the
leading Cardy exponent, without requiring any additional finite-rank input.
Section~\ref{sec:Cardy} makes this relation precise and, under a choice of integration
contour and dominant saddle whose validity is discussed there, uses it to
extract the large-charge growth of the charge-refined index coefficients.

The index exhibits a richer nonperturbative structure than the twisted
superpotential. At $k=1,2,4$, we resolve $17$, $15$, and $15$ exponential
sectors, respectively, and find that each sector is quadratic in the
instanton action $t$, in contrast with the linear dependence on $t$ of the
twisted-superpotential sectors. The expression
\be
F_{\rm np}(N,k)
=
-\sum_{m\geq1}
\left[
R_m^{(k)}
+\frac{k^2}{\pi^2}
\left(
A_m^{(k)}t^2+B_m^{(k)}\left(t+\frac1m\right)
\right)
\right]q^m\,.
\label{eq:intro-index-continuation}
\ee
therefore defines a candidate all-order continuation of the sector
structure observed at finite depth. Only the coefficients through
$m=17,15,15$ at $k=1,2,4$, respectively, are fixed by the reconstruction;
the continuation beyond these orders is conjectural. An exact rational
relation among the three coefficient sequences reduces the problem to two
generating functions,
$\cR_k(q)=\sum_{m\geq1}R_m^{(k)}q^{m}$ and
$\cU_k(q)=\sum_{m\geq1}m\,A_m^{(k)}q^{m}$.
This reduction allows the existence of an all-order closed form to be tested
separately for the two generators.

For $\cR_k(q)$, the reconstructed coefficients exhibit arithmetic and sign
patterns, finite-order growth, and candidate zero locations that motivate two
data-informed modular-product ans\"atze in the shifted nome $Q=-q$, given in
Eq.~\eqref{eq:Rfamilies}. Their integer parameters are fixed exactly from
low-order coefficients, after which all remaining reconstructed coefficients
are reproduced with no additional parameter adjustment. Assuming that these
products give the exact all-order continuation, their modular transformation
properties determine the nearest singularity, the corresponding radius of
convergence, and the leading large-order coefficient growth
(Proposition~\ref{prop:radius}). No analogous representation for
$\cU_k(q)$ is found within the modular class examined in
Section~\ref{sec:U}. The available data therefore provide an all-order
candidate for only one of the two generators, and the complete
nonperturbative finite-rank structure of the index remains open.

The modular organization therefore emerges a posteriori from the
reconstruction. Neither the Bethe equations nor the defining formulas for
the twisted superpotential and the topologically twisted index display an evident 
Eisenstein, Eichler, or modular-product structure. Instead, the
high-precision coefficients first reveal divisor sums, alternating sign
patterns, level-dependent arithmetic, and candidate singularity scales;
modular objects enter only after these features are identified and assembled
into candidate all-order structures. For the twisted superpotential, the
exponentially small sector reduces to a single divisor-sum sequence, whereas
the index retains one generator admitting a related modular-product
description together with a second generator for which no comparable closed
form is presently known. This comparison distinguishes the arithmetic
structure shared by the two observables from the additional finite-$N$
information carried by the index.

Throughout, a candidate expression is fixed from a restricted subset of the
available ranks, levels, or instanton orders, and then checked against data
not used in that determination. The examples above follow this pattern. The
reconstruction of $\cR_k(q)$ has one additional step, because the
modular-product families themselves are motivated by the resolved
coefficients. Once these families are specified, a few low-order coefficients
determine their integer parameters, and the remaining coefficients test the
resulting closed-form candidates for $\cR_k(q)$ in
Eq.~\eqref{eq:Rmodular}.
Section~\ref{sec:exactness} fixes the vocabulary used throughout for these
distinctions --- proved identities, exact reconstructions, numerical
statements, empirical patterns, conjectural continuations, and open questions
--- while Table~\ref{tab:status} summarizes the status and limitations of the
principal claims. Section~\ref{sec:method} then gives the complete
reconstruction and validation protocol.

\subsection*{Organization of the paper}

Section~\ref{sec:setup} fixes conventions, specifies the selected Bethe orbit,
introduces the natural variables and modular notation, and describes the
high-precision data. Sections~\ref{sec:W}--\ref{sec:Cardy} develop the
twisted-superpotential side, recalling the finite-rank structure established
in~\cite{Hosseini:2026dkj}, supplying the constant-map proof deferred there,
and deriving the associated Cardy application. Section~\ref{sec:Z} analyzes
the finite-rank structure of the topologically twisted index, including its
constant map, type-IIA genus expansion, exponentially suppressed sectors, the
modular-product candidate for $\cR_k(q)$, and the unresolved generator
$\cU_k(q)$. Section~\ref{sec:method} presents the analysis framework
underlying these reconstructions. Section~\ref{sec:validation}
provides quantitative tests of the reconstructed
formulas. Section~\ref{sec:discussion} interprets the arithmetic and
nonperturbative structures in relation to earlier ABJM results and develops
the open problems that follow from them. The appendices
contain the root-of-unity proof and closed Euler-sum evaluations
(Appendix~\ref{app:proof}), the constant-map identities
(Appendix~\ref{app:Smaps}), and the complete reconstructed sector
coefficients (Appendix~\ref{app:tables}).

\paragraph{Note added.}
While this work was being finalized, we became aware of the independent
work~\cite{Hong:2026zul}, which determines the rank-independent constants
of the ABJM twisted superpotential and topologically twisted index at the
superconformal point. After translating conventions, the corresponding
constants are related to ours by
\begin{equation*}
\CW(k)=2\pi\,\widehat g_0(k,\Delta_{\rm sc})\,,
\qquad
\fzero(k)=\widehat f_0(k,\Delta_{\rm sc},\mathfrak n_{\rm sc})+\log k\,.
\end{equation*}
Here $\widehat g_0$ and $\widehat f_0$ denote the rank-independent terms in
the conventions of Ref.~\cite{Hong:2026zul}, while $\Delta_{\rm sc}$ and
$\mathfrak n_{\rm sc}$ denote the superconformal chemical potentials and
magnetic fluxes. The second relation accounts for the $k$-fold $\mathbb Z_k$
orbit included in our normalization~\eqref{eq:orbit}. The parity-dependent finite Clausen
representations and rational-argument Euler-sum identities of
Theorems~\ref{thm:C} and~\ref{thm:Euler}, the finite-rank exponential sectors,
and their modular and Eichler organization are not addressed in
Ref.~\cite{Hong:2026zul}.

\section{Setup, conventions and data}
\label{sec:setup}

\subsection{The theory}

The ABJM model~\cite{Aharony:2008ug} is a three-dimensional
$\U(N)_k\times\U(N)_{-k}$ Chern--Simons--matter theory. Its matter sector
consists of two chiral multiplets $A_{1,2}$ transforming in
$(\mathbf N,\overline{\mathbf N})$ and two chiral multiplets $B_{1,2}$
transforming in the conjugate representation
$(\overline{\mathbf N},\mathbf N)$. The quartic superpotential is
$W=\varepsilon^{ab}\varepsilon^{cd}\Tr(A_aB_cA_bB_d)$, where
$\varepsilon^{12}=1$ and $a,b,c,d\in\{1,2\}$ are summed.
At generic integer level $k$, the theory has $\cN=6$ supersymmetry and
continuous global symmetry $\mathrm{SU}(4)_R\times\U(1)_T$.
The first factor is the R-symmetry, whereas the second is the topological
symmetry generated by the identically conserved three-dimensional current
$j_T=\frac{1}{2\pi}\star\Tr(\mathcal F-\widetilde{\mathcal F})$, where
$\mathcal F$ and $\widetilde{\mathcal F}$ are the field strengths of the two
gauge-group factors and $\star$ is the three-dimensional Hodge dual.

We introduce four chemical potentials $\Delta_a$ and background magnetic
fluxes $\mathfrak n_a$ for a convenient rank-four Abelian basis of
$\mathrm{SU}(4)_R\times\U(1)_T$. Following the standard index conventions,
we label them by $a=1,\ldots,4$ in the order $(A_1,A_2,B_1,B_2)$.
The elementary bifundamental
multiplets are neutral under $\U(1)_T$, which instead acts on monopole operators.
For $k=1,2$, monopole operators furnish additional conserved
currents that enhance the supersymmetry and R-symmetry to $\cN=8$ and
$\mathrm{SO}(8)_R$, respectively
\cite{Aharony:2008ug,Gustavsson:2009pm}. The reconstructed instanton
coefficients coincide at these two levels, a feature that emerges from the
data without using the symmetry enhancement as an input.

\subsection{Conventions}

We use angular chemical potentials normalized according to%
\footnote{Reference~\cite{Bobev:2022jte} instead uses
$\sum_a\Delta_a=2$ and $y_a=\e^{\ii\pi\Delta_a}$, so that
$\Delta_a^{\rm here}=\pi\Delta_a^{\rm there}$.}
\be
\sum_{a=1}^{4}\Delta_a=2\pi\,,
\qquad
y_a=\e^{\ii\Delta_a}\,.
\label{eq:Deltaconv}
\ee
On a genus-$\mathfrak g$ Riemann surface $\Sigma_{\mathfrak g}$, supersymmetry
requires the background magnetic fluxes to satisfy
\be
\sum_{a=1}^{4}\mathfrak n_a=2(1-\mathfrak g)\,.
\ee

We denote the angular Coulomb-branch holonomies by $u_i$ and $\wt u_i$,
with $i=1,\ldots,N$, and define
$x_i=\e^{\ii u_i}$ and $\wt x_i=\e^{\ii\wt u_i}$. The effective twisted
superpotential is~\cite{Benini:2015eyy}
\be
\Wt=\sum_{i=1}^{N}\Big[\frac{k}{2}\big(\wt u_i^{2}-u_i^{2}\big)
-2\pi\big(\wt n_i\wt u_i-n_i u_i\big)\Big]
+\sum_{i,j=1}^{N}\sum_{\sigma=\pm1}\sigma\sum_{a\in I_\sigma}
\Li_2\big(\e^{\ii(\wt u_j-u_i+\sigma\Delta_a)}\big) \, ,
\label{eq:Wdef}
\ee
where $I_-=\{1,2\}$, $I_+=\{3,4\}$, and
$\Li_s(z)=\sum_{m\geq1}z^m/m^s$, analytically continued with the
dilogarithm branch cut along $[1,\infty)$.
Because $\Wt$ is multivalued its
critical points exist only after a lift, and the integers
$n_i,\wt n_i\in\mathbb{Z}$ in~\eqref{eq:Wdef} record it; we write $n$
for lifts and reserve $\mathfrak{n}$ for fluxes. Throughout we take
\be
n_i=\wt n_i=-i\,,
\qquad i=1,\ldots,N\,,
\ee
and order the eigenvalues so that both $\im u_i$ and $\im\wt u_i$ increase
with $i$. With this lift, the long-range forces cancel and
the resulting on-shell twisted superpotential has
the $N^{3/2}$ scaling~\cite{Benini:2015eyy}.
This prescription selects the Bethe
vacuum studied below, and we adopt it as given.

\subsection{The index}

The topologically twisted index on $S^1\times\Sigma_{\mathfrak g}$ admits a
representation as a sum over solutions of the Bethe ansatz equations (BAEs)~\cite{Nekrasov:2009uh,Benini:2015noa}
\be
\resizebox{0.93\textwidth}{!}{$\displaystyle
Z_{S^1 \times \Sigma_\mathfrak g}
=\prod_{a=1}^4y_a^{-\frac{N^2}{2}\mathfrak n_a}
\sum_{\{x_i,\widetilde x_j\}\in\mathrm{BAE}}
\Bigg[\frac{1}{\det\mathbb B}
\frac{\prod_{i=1}^N x_i^N\widetilde x_i^N
\prod_{i\ne j}\big(1-\frac{x_i}{x_j}\big)\big(1-\frac{\widetilde x_i}{\widetilde x_j}\big)}
{\prod_{i,j=1}^N
\prod_{a=1}^{2} (\widetilde x_j-x_i y_a)^{1-\frac{\mathfrak n_a}{1-\mathfrak g}}
\prod_{a=3}^{4} (x_i-\widetilde x_j y_a)^{1-\frac{\mathfrak n_a}{1-\mathfrak g}}}
\Bigg]^{1-\mathfrak g}
$}\, .
\label{eq:BAformula}
\ee
The sum runs over BAE solutions $\exp(\ii\partial_{u_i}\Wt)=
\exp(\ii\partial_{\wt u_i}\Wt)=1$ modulo the Weyl group. Here $\mathbb{B}$ is
the $2N\times2N$ Jacobian of the logarithmic Bethe operators,
\bea
\mathbb{B}&=\begin{pmatrix}
\delta_{jl}\big(k-\textstyle\sum_{m} \mathscr G_{jm}\big) & \mathscr G_{jl}\\[2pt]
-\mathscr G_{lj} & \delta_{jl}\big(k+\textstyle\sum_{m} \mathscr G_{mj}\big)
\end{pmatrix} \, ,\\[4pt]
\mathscr G_{ij}&=\left[z\frac{\rd}{\rd z}\log \mathscr D(z)\right]_{z=\wt x_j/x_i} \, ,\qquad
\mathscr D(z)=\frac{(1-zy_3)(1-zy_4)}{(1-z/y_1)(1-z/y_2)} \, .
\label{eq:Bmatrix}
\eea
where $j,l,m=1,\ldots,N$ and $\delta_{jl}$ is the Kronecker delta.
This formula assumes that the contributing Bethe roots are admissible
and nondegenerate. They must avoid all singular loci of the summand and satisfy
$\det\mathbb B\neq0$.

\subsection{Universal twist, genus dependence, and the orbit factor}

We specialize to the universal twist, for which the background magnetic flux
is aligned with the exact superconformal R-symmetry. In our conventions,
\be
\Delta_a=\frac{\pi}{2}\,,
\qquad
\mathfrak n_a=\frac{1-\mathfrak g}{2}\,,
\qquad
a=1,\ldots,4\,.
\label{eq:universal}
\ee
All numerical computations below are performed at $\mathfrak g=0$, so that
$\mathfrak n_a=1/2$.

The flux assignment requires one qualification. The half-integer values
$\mathfrak n_a$ specify the superconformal R-symmetry direction. In the
integer-charge basis underlying the localization formula, Dirac quantization
requires integral $\mathfrak n_a$. The universal assignment therefore defines
an ordinary quantized background only for odd $\mathfrak g$. At
$\mathfrak g=0$ we use the Bethe-vacua expression by analytic continuation to
$\mathfrak n_a=1/2$.

At the universal twist, the genus dependence of~\eqref{eq:BAformula}
factorizes for each individual Bethe solution:
\be
1-\frac{\mathfrak n_a}{1-\mathfrak g}
=
\frac{1}{2}\,,
\qquad
\prod_{a=1}^4
y_a^{-\frac{N^2}{2}\mathfrak n_a}
=
\bigg(
\prod_{a=1}^4 y_a^{-\frac{N^2}{4}}
\bigg)^{1-\mathfrak g} \,.
\label{eq:universalfactorization}
\ee
The expression inside the square brackets in~\eqref{eq:BAformula} is
therefore independent of $\mathfrak g$, while the prefactor carries the
same overall power $1-\mathfrak g$. The Bethe equations themselves are
independent of $\mathfrak g$ and of the magnetic fluxes. Consequently, if
$Z_{\rm data}(N,k)$ denotes the genus-zero contribution of the selected
Bethe solution, the contribution of the same solution at
$\mathfrak g\neq1$ is
\be
\left[Z_{\rm data}(N,k)\right]^{1-\mathfrak g}\,.
\label{eq:genusrel}
\ee
Thus the genus-zero result determines the contribution of the same Bethe
vacuum at higher genus. This relation holds vacuum by vacuum and does not
imply the corresponding identity for the full index, since the sum over Bethe
vacua is performed only after the individual contributions have been raised
to the power $1-\mathfrak g$.

The case $\mathfrak g=1$ must be treated separately.
The universal magnetic flux then vanishes and the twisted
partition function computes the Witten index~\cite{Benini:2016hjo}, or
equivalently the number of admissible Bethe vacua counted with the appropriate
multiplicities.

The Bethe equations are invariant under the simultaneous multiplication
\be
\{x_i,\wt x_i\}
\longmapsto
\{\varpi_k x_i,\varpi_k\wt x_i\}\,,
\qquad
\varpi_k=\e^{2\pi\ii/k}\,,
\label{eq:Zkaction}
\ee
and the summand in~\eqref{eq:BAformula} is invariant under the same action.
For the solution studied below, this transformation generates a $k$-element
orbit of distinct solutions modulo the Weyl group. Our tables record one
representative. The selected orbit therefore contributes
\be
k\left[Z_{\rm data}(N,k)\right]^{1-\mathfrak g}\,,
\qquad \mathfrak g\neq1\,.
\ee
The multiplicity is $k$ for every genus because it counts distinct terms in
the Bethe-vacua sum; only the contribution of each vacuum is raised to the
power $1-\mathfrak g$.

At genus zero, we define the real free-energy contribution of this orbit by
\be
F(N,k)\equiv-\re\log\left[k Z_{\rm data}(N,k)\right]
=-\re\log Z_{\rm data}(N,k)-\log k\,.
\label{eq:orbit}
\ee
On this locus our computations give $\re\Wt=0$ to full working
precision at every computed point, and $\im\log Z_{\rm data}$ is an integer multiple of $\pi$ from the chosen
logarithm branch, and this phase is discarded in the real free energy.

\emph{Scope of $F$.} Throughout, $F$ denotes the contribution of this
solution and its orbit. Whether further Bethe solutions contribute at these
ranks is an open question, and one that lies outside what the present data can
settle.

\subsection{Natural variables and notation}
\label{sec:natural-variables}

The large-rank data are most naturally organized in terms of the shifted rank
$\Nh$, the associated instanton action $t$, and the physical nome $q$,
\be
\Nh=N-\frac{k}{24}+\frac{2}{3k}\,,\qquad
t=2\pi\sqrt{\frac{2\Nh}{k}}\,,\qquad q=\e^{-t}\,.
\label{eq:shift}
\ee
Section~\ref{sec:method} infers these variables from the data themselves.
Whenever $t$ and $q$ are used as real instanton variables, we
restrict to $\Nh>0$ and take the positive square root; every rank
entering the $q$-series analyses satisfies this condition.
We also write
\be
\lambda=\frac{N}{k}\,,
\qquad
\lamh=\lambda-\frac{1}{24}
\ee
for the 't~Hooft coupling and its shifted counterpart. The rank shift
in~\eqref{eq:shift} is the universal-twist shift of
Ref.~\cite{Bobev:2022jte}. It differs by $1/k$ from the Airy shift
$N-k/24-1/(3k)$ governing the $S^3$ free energy
\cite{Fuji:2011km,Marino:2011eh}.

We denote the divisor sums by
\be
\sigma_a(n)=\sum_{d\mid n}d^a\,,
\qquad
\sigma_a(x)=0
\quad\text{for}\quad
x\notin\mathbb Z_{>0}\,.
\label{eq:divisor}
\ee
For a positive integer $m$, we further define
\be
\sigma_1^{\rm odd}(m)
=
\sigma_1\!\left(\frac{m}{2^{v_2(m)}}\right),
\ee
namely the sum of the divisors of the odd part of $m$, where $v_2(m)$ is the
$2$-adic valuation of $m$.

The Clausen function is
\be
\Cltwo(\theta)
=
\im\Li_2(\e^{\ii\theta})
=
\sum_{n\geq1}\frac{\sin(n\theta)}{n^2}\,,
\ee
and Catalan's constant is $G=\Cltwo(\pi/2)$. We use the generalized harmonic
numbers
\be
\Hgen_z=\psi(1+z)+\gamma\,,
\qquad z\geq0\,,
\ee
where $\psi$ is the digamma function and $\gamma$ is the Euler--Mascheroni
constant. Our Bernoulli-number convention is
\be
\frac{x}{\e^x-1}
=
\sum_{m\geq0}\mathrm B_m\frac{x^m}{m!}\,,
\ee
so that $\mathrm B_2=1/6$, $\mathrm B_4=-1/30$, and $\mathrm B_6=1/42$.

For modular quantities, we reserve $q$ for the physical nome and denote the
shifted nome by $Q$. We use $[z^m]\mathscr A(z)$ for the coefficient of $z^m$
in a formal power series $\mathscr A(z)$. The two nomes are related by
\be
Q=-q\,.
\label{eq:Qnome}
\ee
We write
\be
q=\e^{2\pi\ii\tau}\,,
\qquad
Q=\e^{2\pi\ii\tau_2}=\e^{\ii\pi\tau_4}\,,
\qquad
\tau_2=\tau+\frac12\,,
\qquad
\tau_4=2\tau_2\,.
\label{eq:tauconv}
\ee
On the physical locus, $\tau=\ii t/(2\pi)$. Thus $q\mapsto Q=-q$ is the
half-period shift $\tau\mapsto\tau+1/2$, while a further sign change
$Q\mapsto-Q$ is $\tau_4\mapsto\tau_4+1$ in the theta-nome convention.
The modular parameters lie in the upper half-plane
$\mathbb H=\{\tau\in\mathbb C:\im\tau>0\}$. For a subgroup $\Gamma\subseteq\mathrm{SL}_2(\mathbb Z)$,
$M_w(\Gamma)$ denotes the space of holomorphic modular forms of weight $w$.
Writing
$\mathsf M=\left(\begin{smallmatrix}\mathsf a&\mathsf b\\
\mathsf c&\mathsf d\end{smallmatrix}\right)\in\mathrm{SL}_2(\mathbb Z)$,
for a positive integer $L$ we use
\begin{equation*}
\Gamma_0(L)
=
\left\{
\mathsf M\in\mathrm{SL}_2(\mathbb Z):
\mathsf c\equiv0\pmod L
\right\},
\qquad
\Gamma(2)
=
\left\{
\mathsf M\in\mathrm{SL}_2(\mathbb Z):
\mathsf M\equiv\mathbb I\pmod2
\right\},
\end{equation*}
where $\mathbb I$ is the $2\times2$ identity matrix. We denote the compact
modular curves associated with $\Gamma_0(L)$ and $\Gamma(2)$ by $X_0(L)$
and $X(2)$, respectively.

The normalized Eisenstein series have the Fourier expansion
\be
E_{2\mathsf k}(Q)
=
1-\frac{4\mathsf k}{\mathrm B_{2\mathsf k}}
\sum_{m\geq1}\sigma_{2\mathsf k-1}(m)Q^m\,,
\qquad \mathsf k\geq1\,.
\label{eq:Eisenstein-general}
\ee
In particular,
\bea
E_2(Q)=1-24\sum_{m\geq1}\sigma_1(m)Q^m\,,\qquad
E_4(Q)=1+240\sum_{m\geq1}\sigma_3(m)Q^m\,.
\label{eq:E2-E4}
\eea
The series $E_2$ is quasimodular, whereas $E_{2\mathsf k}$ is a modular form
of weight $2\mathsf k$ on $\mathrm{SL}_2(\mathbb Z)$ for every
$\mathsf k\geq2$. We use the logarithmic derivative
\be
D\equiv Q\partial_Q=q\partial_q\,,
\label{eq:Ddef}
\ee
where the equality follows from $Q=-q$. For any positive integer $r$, we
define the coefficientwise inverse derivative of a formal series with vanishing
constant term by
\be
D^{-r}\bigg(\sum_{m\geq1}a_mQ^m\bigg)
\equiv
\sum_{m\geq1}\frac{a_m}{m^r}Q^m\,.
\label{eq:Dinverse}
\ee
The constant term is fixed to zero. The same convention applies to a formal
$q$-series because $D=q\partial_q$. We also use the Euler product
\be
\Eul(z)=\prod_{n\geq1}(1-z^n)\,,
\qquad |z|<1\,.
\label{eq:Eul-def}
\ee
For every positive integer $r$,
\be
D\log\Eul(Q^r)
=
\frac{r}{24}\bigl[E_2(Q^r)-1\bigr] \, .
\label{eq:Eul-derivative}
\ee
Indeed, expanding the logarithmic derivative of the product gives
$-r\sum_{n\geq1}nQ^{rn}/(1-Q^{rn})
=-r\sum_{m\geq1}\sigma_1(m)Q^{rm}$, which agrees with
Eq.~\eqref{eq:E2-E4}. The Dedekind eta function is normalized by
\be
\eta(\upsilon)
=
\e^{\pi\ii\upsilon/12}\Eul(\e^{2\pi\ii\upsilon})\,.
\label{eq:eta-def}
\ee
We also use the standard $S$-transformation
\be
\eta(-1/\upsilon)
=
(-\ii\upsilon)^{1/2}\eta(\upsilon)\,,
\label{eq:eta-S}
\ee
where the square root is taken on the principal branch for
$\upsilon\in\mathbb H$.

For the theta-nome convention $z=\e^{\ii\pi\upsilon}$, we define
\bea
\vartheta_2(z)
&=
\sum_{n\in\mathbb Z}z^{(n+1/2)^2}
=
2z^{1/4}\sum_{n\geq0}z^{n(n+1)}\,,
\\
\vartheta_3(z)
&=
\sum_{n\in\mathbb Z}z^{n^2}
=
1+2\sum_{n\geq1}z^{n^2}\,,
\label{eq:theta-def}
\eea
and
\be
\lambda_{\rm mod}(z)
=
\frac{\vartheta_2(z)^4}{\vartheta_3(z)^4}
=
16z\,
\frac{\bigl(\sum_{n\geq0}z^{n(n+1)}\bigr)^4}
{\bigl(1+2\sum_{n\geq1}z^{n^2}\bigr)^4}\,.
\label{eq:lambdadef}
\ee
In particular, $\lambda_{\rm mod}(z)=16z+O(z^2)$ as $z\to0$.
When the argument is a modular parameter, $\lambda_{\rm mod}(\upsilon)$
abbreviates $\lambda_{\rm mod}(\e^{\ii\pi\upsilon})$. The $S$- and
$T$-transformations used below, with
$S:\upsilon\mapsto-1/\upsilon$ and $T:\upsilon\mapsto\upsilon+1$, are
\bea
\lambda_{\rm mod}(-1/\upsilon)
&=1-\lambda_{\rm mod}(\upsilon)\,,
\\
\lambda_{\rm mod}(\upsilon+1)
&=
\frac{\lambda_{\rm mod}(\upsilon)}
{\lambda_{\rm mod}(\upsilon)-1}\,,
\qquad
\lambda_{\rm mod}(\ii)=\frac12\,.
\label{eq:lambda-T}
\eea
The last value follows from the first transformation at its fixed point.
With the conventions of~\eqref{eq:tauconv},
$\lambda_{\rm mod}(q)=\lambda_{\rm mod}(2\tau)$ and
$\lambda_{\rm mod}(-Q)=\lambda_{\rm mod}(2\tau_2+1)$ are Hauptmoduls for
$\Gamma_0(4)$ in the physical and shifted parametrizations, respectively.%
\footnote{Let
$\mathsf M=\left(\begin{smallmatrix}\mathsf a&\mathsf b\\
\mathsf c&\mathsf d\end{smallmatrix}\right)\in\Gamma_0(4)$ and set
$w_0=2\upsilon$, $w_1=2\upsilon+1$. A direct fractional-linear calculation
gives
\begin{equation*}
2(\mathsf M\upsilon)
=
\begin{pmatrix}\mathsf a&2\mathsf b\\ \mathsf c/2&\mathsf d\end{pmatrix}w_0,
\qquad
2(\mathsf M\upsilon)+1
=
\begin{pmatrix}
\mathsf a+\mathsf c/2&-\mathsf a-\mathsf c/2+2\mathsf b+\mathsf d\\
\mathsf c/2&\mathsf d-\mathsf c/2
\end{pmatrix}w_1,
\end{equation*}
where a matrix acts by a fractional-linear transformation. Both matrices have
determinant one; because $\mathsf c\equiv0\pmod4$, their diagonal entries are
odd and their off-diagonal entries are even, so they lie in $\Gamma(2)$. The
first map is onto $\Gamma(2)$, with inverse
$\left(\begin{smallmatrix}\alpha&\beta\\ \chi&\delta\end{smallmatrix}\right)
\mapsto
\left(\begin{smallmatrix}\alpha&\beta/2\\2\chi&\delta\end{smallmatrix}\right)$.
If the first induced matrix is $\Delta$, the second is $T\Delta T^{-1}$, with
$T=\left(\begin{smallmatrix}1&1\\0&1\end{smallmatrix}\right)$. Normality of
$\Gamma(2)$ in $\mathrm{SL}_2(\mathbb Z)$ makes the second map onto as well.
These group isomorphisms identify the corresponding quotient curves with
$X(2)$. Since $\lambda_{\rm mod}$ generates the function field of $X(2)$, the
two compositions are Hauptmoduls for $X_0(4)$.}

Finally, $\mathfrak g$ always denotes the genus of the Riemann surface
$\Sigma_{\mathfrak g}$, whereas $h$ denotes the order in the type-IIA genus
expansion of Sections~\ref{sec:Wgenus} and~\ref{sec:Zgenus}.

\subsection{The data and their precision}

We generated the selected Bethe solutions using two numerically independent
procedures. In the first implementation, the logarithmic Bethe equations are
interpreted as force-balance conditions for eigenvalues in the complex plane.
A dissipative flow drives a seed configuration into the selected basin of
attraction~\cite{Herzog:2010hf,Benini:2015eyy}, after which
arbitrary-precision Newton iteration refines the solution. The second
implementation transports a known solution in rank using the continuation
algorithm of Ref.~\cite{Hosseini:2025jxb}. At every point where both calculations were
performed, the on-shell observables agree within the retained decimal precision.

The analysis draws on two complementary sets of values:
\begin{itemize}\itemsep2pt
\item a broad set, retained to $200$ significant decimal digits, covering
      $1\leq k\leq50$. At $k=1$ the even-rank values extend to $N=800$; for
      $2\leq k\leq50$ they extend to $N=500$, and the odd-rank values to
      $N=299$ at every level. Each observable is evaluated at $20\,100$ points;
\item a deep set, retained to $800$ significant decimal digits, covering every
      $2\leq N\leq100$ at $k=1,2,4$, giving $297$ values for each observable.
\end{itemize}
The computational analysis reported here begins from these four tables and
reconstructs every result that depends on computation, including the parents,
constant maps, sector coefficients, modular searches, figures, tables, and
numerical checks. The two Bethe-equation solver implementations used to generate
the tables lie outside this workflow, so the auditable computational boundary
is the on-shell scalar data.

Here ``$200$ digits'' and ``$800$ digits'' record how many digits are retained;
they are not certified forward-error bounds. Every primary value,
parent subtraction, integer-relation input, and high-precision comparison is
carried in arbitrary precision up to the point where an exact decision is
made. Ordinary floating point is confined to plots and diagnostic quantities
and does not enter inference-critical decisions. The dynamic range makes this
essential. At $(k,N)=(1,800)$, subtracting a parent of order
$5\times10^{4}$ exposes a residual of order $10^{-107}$, roughly one hundred
and twelve decimal orders below the leading contribution.

\subsection{Senses of exactness}
\label{sec:exactness}

The results below combine analytic proofs, exact arithmetic reconstruction
from finite data, high-precision numerical tests, and conjectural all-order
continuations. We use the following terminology throughout. For a primitive
integer relation, we define its height to be the maximum of the absolute
values of its integer coefficients. For a rational number $\mathsf p/\mathsf q$ in lowest
terms, with $\gcd(\mathsf p,\mathsf q)=1$, we define its height by
$H(\mathsf p/\mathsf q)\equiv \max\{|\mathsf p|,|\mathsf q|\}$.

\begin{description}[
leftmargin=0pt,
font=\normalfont\bfseries
]

\item[Proved.]
An analytic statement established from explicit assumptions in this paper or
from a cited result. Theorems~\ref{thm:Euler}, \ref{thm:C}, and~\ref{thm:f0}
have this status. Proposition~\ref{prop:radius} is proved under the assumption
of the all-order continuation~\eqref{eq:Rmodular} stated in its hypothesis.

\item[Exactly reconstructed.]
An exact arithmetic expression recovered from decimal data within a declared
finite ansatz, basis, or search class, with the relevant height bounds stated.
Its parameters are determined from specified data and then checked against
independent data or coefficients not used in that determination. The term
``exact'' refers to the recovered arithmetic statement within the declared
class; its identification with the underlying Bethe-vacuum observable remains
a finite-data statement. When the class itself is informed by the reconstructed
data, its provenance is stated explicitly. Examples include the index constant
map~\eqref{eq:f0arith} and the finite-range relation~\eqref{eq:Brel} among the
index sectors.

\item[Numerically established.]
A quantitative statement supported by direct numerical comparison at a stated
working precision over a specified parameter range. Examples include the
agreement between the two independent Bethe-solution computations wherever
both are available, and the vanishing of \(\operatorname{Re}\widetilde{W}\) to
working precision over the computed grid.

\item[Empirical pattern.]
A structural or asymptotic pattern inferred from the available finite data
without a derivation from the underlying Bethe equations. Examples include
the candidate reciprocal radii in~\eqref{eq:rhovalues}, inferred from
finite-order coefficient growth prior to assuming the modular-product
continuation, and the extrapolation~\eqref{eq:Uleading} for
$\cU_k(q)$.

\item[Conjectural continuation.]
A specified all-order extension beyond the reconstructed coefficient range.
The divisor series~\eqref{eq:Wnp}, the all-order ansatz for the exponentially
suppressed index sectors~\eqref{eq:Znp-continuation}, and the modular
product~\eqref{eq:Rmodular} have this status. For both observables, the
$k=1$ and $k=2$ sector coefficients agree exactly over their common
resolved orders. Any extension of these equalities beyond the reconstructed
ranges is conjectural. Analytic consequences derived from an all-order
continuation retain that continuation as an explicit hypothesis.

\item[Open.]
A question not settled by the present analysis. The scope of an exhaustive
bounded search is determined by its stated search class and size cap.
Section~\ref{sec:U} excludes representations of $\cU_k(q)$ within the
specified modular dictionary and size cap; representations outside that
class remain open.

\end{description}

Table~\ref{tab:status} summarizes the evidence, status, and limitations of the
paper's principal claims using this terminology.

\begingroup
\small
\renewcommand{\arraystretch}{1.16}
\begin{longtable}{@{}>{\raggedright\arraybackslash}p{0.21\linewidth}
                      >{\raggedright\arraybackslash}p{0.30\linewidth}
                      >{\raggedright\arraybackslash}p{0.43\linewidth}@{}}
\caption{Status of the principal claims, using the vocabulary introduced
above. When a class is itself data-informed, as for
$\cR_k$, the table states this separately.}
\label{tab:status}\\
\toprule
Claim & Computational evidence & Status and limitation \\
\midrule
\endfirsthead
\multicolumn{3}{l}{\small\tablename~\thetable\ continued}\\
\toprule
Claim & Computational evidence & Status and limitation \\
\midrule
\endhead
\midrule
\multicolumn{3}{r}{\small Continued on the next page}\\
\endfoot
\bottomrule
\endlastfoot

Shifted rank and perturbative parents
& The root scan is performed in four disjoint rank windows. The level law is
reconstructed from $k=1,\ldots,12$ and validated on
$k=13,\ldots,50$.
& The coefficients are reconstructed within the prescribed bases
$\{k,1/k\}$ and
$\{\Nh^{3/2},\Nh^{1/2},\log\Nh,1\}$. The choice of these bases is supplied as
physical input to the search. \\
\addlinespace

Twisted superpotential constant $\CW(k)$
& Levelwise integer relations are certified at
$k=1,\ldots,8,10,12$ using the four robustness checks of the companion
Letter~\cite{Hosseini:2026dkj}, together with the additional three-point
plateau-tail criterion. The resulting level formula is checked against all
plateaux for $1\leq k\leq50$.
& Exactly reconstructed from the Bethe-vacuum plateaux within the declared
basis; its identification with the twisted-superpotential constant term remains
a finite-data statement. Equality of the Clausen, Euler-sum, and integral
representations is proved for all positive integers $k$ in
Theorem~\ref{thm:C}. \\
\addlinespace

Twisted superpotential sectors
& $51$, $67$, and $93$ rational sectors are reconstructed at
$k=1,2,4$, respectively. At every resolved order, their coefficients agree
exactly with the stated divisor formula.
& The result is a finite-order reconstruction. The all-order divisor formula
remains conjectural. Its Eichler representation and convergence in the unit
disc follow analytically once this continuation is assumed. \\
\addlinespace

Index constant $\fzero(k)$
& Six coefficients are reconstructed from $k=1,\ldots,6$ and validated on nine
withheld levels through $k=50$. An independent large-$k$ comparison provides
an additional check.
& Exactly reconstructed from the index plateaux within the declared
six-feature space. Its identification with the rank-independent term of the
topologically twisted index remains a finite-data statement. Equality of the
arithmetic and one-kernel representations is proved for all real $k>0$ in
Theorem~\ref{thm:f0}. \\
\addlinespace

Index sectors and their relation
& $17$, $15$, and $15$ rational sectors are reconstructed at
$k=1,2,4$, respectively. The relation is determined from the minimal
low-order subset and verified over $\mathbb Q$ on all remaining reconstructed
sectors.
& Within the reconstructed coefficient ranges, the relation is an exact
identity over $\mathbb Q$. \\
\addlinespace

Modular product ansatz for $\cR_k$ and its radius of convergence
& Data-informed product families of modular levels two and four are identified
in the shifted nome $Q=-q$. Their integer parameters are determined exactly
from the first three coefficients of the common $k=1,2$ sequence and the first
two coefficients at $k=4$. At $k=4$, the candidate zero provides an
independent constraint on the chosen linear Hauptmodul factor. All remaining
resolved coefficients are then reproduced exactly.
& The integer parameters are unique within the stated product families.
At $k=1,2,4$, respectively, $14$, $12$, and $13$ coefficients are withheld
from the coefficient-based determination and reproduced exactly. The infinite
product \eqref{eq:Rmodular} remains conjectural. Conditional on this
continuation, Proposition~\ref{prop:radius} proves the location and logarithmic
nature of the singularity, the radius of convergence, and the leading
coefficient growth. \\
\addlinespace

Bounded modular search for $\cU_k$
& Rational linear combinations of at most four entries are considered from the
declared level-two and level-four library in the shifted nome $Q=-q$. For each
candidate support, the coefficients are determined from a nonsingular set of
instanton orders selected by exact row reduction and then tested against every
remaining resolved order.
& No candidate reproduces all reconstructed coefficients of $\cU_k$ at
$k=1,2,4$. The negative result is exhaustive within the declared library and
the four-entry cap. Representations outside this search class remain open. \\
\addlinespace

Cardy evaluation at the selected saddle
& The reconstructed $\im\Wt(N,k)$ is substituted into the Bethe-saddle Cardy
formula.
& The resulting expression applies to the selected Bethe saddle. Interpreting
it as an asymptotic formula for the corresponding coefficients additionally
requires a valid contour deformation and dominance of that saddle. \\
\addlinespace
\end{longtable}
\endgroup

\section{The effective twisted superpotential}
\label{sec:W}

The companion Letter~\cite{Hosseini:2026dkj} reconstructed the shifted-rank parent and
constant map, derived the type-IIA coefficients at arbitrary genus, and found
the special-level divisor sequence of the on-shell twisted superpotential,
together with its proposed Eisenstein--Eichler all-order continuation. We
include these results because they supply one side of the comparison with the index and enter its constant map.
The new analytic contribution in this paper is the proof, given in
Appendix~\ref{app:proof}, that the finite Clausen, alternating-Euler-sum, and
integral representations of the constant map are equivalent. The same proof
yields the closed rational-argument Euler sums stated there.

\subsection{The perturbative parent}
\label{sec:Wpert}

At fixed $k$ the perturbative part of $\im\Wt$ consists of two terms,
\be
\im\Wt(N,k)\big|_{\rm pert}
=
\frac{\pi^{2}}{3}\sqrt{\frac{k}{2}}\,\Nh^{3/2}
+\CW(k)\,.
\label{eq:Wparent}
\ee
Unrestricted four-term fits at the representative levels
$k=1,2,3,5,10,20,50$ reconstruct this form. At each of these levels, we solve
for the level-dependent fit coefficients
$\mathsf w_{3/2},\mathsf w_{1/2},\mathsf w_{\log},\mathsf w_0$ in the ansatz
\be
\mathsf w_{3/2}\Nh^{3/2}
+\mathsf w_{1/2}\Nh^{1/2}
+\mathsf w_{\log}\log\Nh
+\mathsf w_0
\ee
using the four largest available ranks.
The result gives
$\mathsf w_{3/2}=\frac{\pi^{2}}{3}\sqrt{k/2}$ and
$\mathsf w_0=\CW(k)$, while $\mathsf w_{1/2}$ and
$\mathsf w_{\log}$ are numerically consistent with zero. For example, at
$k=1$ we find $\mathsf w_{1/2}\sim6\times10^{-102}$ and
$\mathsf w_{\log}\sim-5\times10^{-101}$, against a leading contribution
of order $5\times10^{4}$. The corresponding estimates include
$\mathsf w_{1/2}\sim2\times10^{-32}$ at $k=5$ and
$\mathsf w_{1/2}\sim3\times10^{-21}$ at $k=10$.
The accuracy with which all four coefficients are reconstructed is limited by the
exponentially small remainder and by the conditioning of the four-parameter
solve. It consequently deteriorates with increasing $k$ in parallel with the
size of the resolved remainder.

The residual over the full computed rank range provides the stronger
test. After subtracting
\eqref{eq:Wparent}, the remainder is exponentially small in $\sqrt{\Nh}$ and
tracks the physical scale $q=\e^{-t}$. Any additional perturbative contribution
would produce a
power-law tail and therefore could not follow the $q$ scale uniformly across
the available ranks. At $(k,N)=(1,800)$, the residual is of order $10^{-107}$.
If one adds separately a term proportional to $\Nh^{1/2}$, $\log\Nh$, or
$\Nh^{-1}$, its coefficient is bounded at approximately the $10^{-104}$ level.
No additional algebraic term is resolved over the available ranks, and we
therefore adopt~\eqref{eq:Wparent} as the perturbative parent. Its proposed
all-orders validity is supported by the exponential behavior of the residual;
a derivation from the Bethe equations is not attempted here.

This structure differs from that of the $S^3$ free energy. At fixed level, the
latter has an infinite perturbative $1/N$ expansion whose resummation is
governed by an Airy function~\cite{Fuji:2011km,Marino:2011eh}. Both
observables have leading $N^{3/2}$ behavior and are naturally expressed in
terms of shifted ranks, but the universal-twist shift~\eqref{eq:shift} differs
by $1/k$ from the Airy shift relevant to the $S^3$ partition
function~\cite{Bobev:2022jte}. At leading order, and in the normalization
of~\eqref{eq:Wdef}, the two are related by \cite{Hosseini:2016tor}
\be
\im\Wt(N,k)\big|_{\rm pert}
=
\frac{\pi}{2}F_{\sthree}(N,k)\,,
\ee
where $F_{\sthree}(N,k)$ denotes the ABJM three-sphere free energy.
Beyond the leading large-$N$ asymptotics, their perturbative structures
differ.

\subsection{The constant map and its exact representations}
\label{sec:WC}

Subtracting the leading $\Nh^{3/2}$ term from $\im\Wt(N,k)$ leaves a sequence
that converges exponentially to the constant term $\CW(k)$, at a rate set by
the first instanton sector. For each attempted level $k\leq20$, we search for an exact
reconstruction of the resulting plateau in the basis
\be
\{\zeta(3)/\pi,\,G,\,\Cltwo(2\pi j/k):1\leq j<k\} \, .
\ee
Certification retains the four robustness checks of the companion
Letter~\cite{Hosseini:2026dkj}, namely the primitive height bound,
reconstruction at the drift-limited precision, stability under a reduction to
three quarters of that digit budget, and stability under deletion of each
unused basis element. We supplement these checks here with an independent
three-point plateau-tail criterion based on the final three plateau estimates.
The levels satisfying both the Letter protocol and this additional criterion
are $k=1,\ldots,8,10,12$. We organize the corresponding certified relations
according to the parity of $k$.

For odd $k=p$, the reconstructed expression is
\bea
\CW(p)
&=
-\left(\frac{p^{2}}{4}+\frac{3}{8p}\right)
\frac{\zeta(3)}{\pi}
+(-1)^{\frac{p+1}{2}}G
\\
&\quad
-\frac{3}{8}
\sum_{j=1}^{p-1}
\bigl(1-(-1)^j\bigr)\,j\,
\Cltwo\!\left(\frac{2\pi j}{p}\right)
-\frac{1}{4}
\sum_{j=1}^{p-1}
(-1)^j\,j\,
\Cltwo\!\left(\frac{4\pi j}{p}\right),
\label{eq:clausen-odd}
\eea
while for even $k=2p$ it is
\be
\CW(2p)
=
-\left(
p^{2}-\frac{1+7(-1)^p}{4p}
\right)
\frac{\zeta(3)}{\pi}
-
\sum_{j=1}^{p-1}
\frac{3-(-1)^{p+j}}{2}\,j\,
\Cltwo\!\left(\frac{2\pi j}{p}\right).
\label{eq:clausen-even}
\ee
Empty sums are understood to vanish. At the lowest levels, these formulas
reduce to
\be
\CW(1)
=
-G-\frac{5\zeta(3)}{8\pi}\,,
\qquad
\CW(2)
=
-\frac{5\zeta(3)}{2\pi}\,,
\qquad
\CW(4)
=
-\frac{3\zeta(3)}{\pi}\,.
\ee
Table~\ref{tab:CW} lists the results through $k=8$. Only the levels
$k\leq7$ are used to assemble
\eqref{eq:clausen-odd}--\eqref{eq:clausen-even}. The resulting formulas then
reproduce all $43$ withheld plateaux with $8\leq k\leq50$, in each case to
the precision set by the corresponding instanton remainder.

The range of the levelwise integer-relation search is limited by the precision
available in the finite-rank plateaux. Let $p_0(k)$, $p_1(k)$, and $p_2(k)$
denote the plateau estimates at the three largest available ranks, ordered from
largest to smallest rank. We assign the conservative digit budget
\be
d(k)=\left\lfloor-\log_{10}|p_0(k)-p_1(k)|\right\rfloor-1\,.
\label{eq:plateau-digits}
\ee
The reach of PSLQ is determined by the competition between this digit supply
and the precision required to certify an integer relation. We consequently
restrict the levelwise integer-relation search to $k\leq20$. This is a
conservative search range, not a negative statement about possible
lower-height relations at larger $k$. The quantitative precision criterion
underlying this choice is discussed in Section~\ref{sec:method}.
Figure~\ref{fig:constant_maps}(a) summarizes the available digit budget and
the outcome of the levelwise searches.

\begin{table}[t]
\centering
\small
\begin{tabular}{@{}clc@{}}
\toprule
$k$ & $\CW(k)$ & basis size / height \\
\midrule
$1$ & $-\dfrac{5\zeta(3)}{8\pi}-G$ & $2$ / $8$\\[6pt]
$2$ & $-\dfrac{5\zeta(3)}{2\pi}$ & $2$ / $5$\\[6pt]
$3$ & $-\dfrac{19\zeta(3)}{8\pi}+G-\tfrac32\Cltwo\!\big(\tfrac{2\pi}{3}\big)$ & $3$ / $19$\\[6pt]
$4$ & $-\dfrac{3\zeta(3)}{\pi}$ & $2$ / $3$\\[6pt]
$5$ & $-\dfrac{253\zeta(3)}{40\pi}-G+\tfrac12\Cltwo\!\big(\tfrac{2\pi}{5}\big)
      +\tfrac72\Cltwo\!\big(\tfrac{4\pi}{5}\big)$ & $4$ / $253$\\[6pt]
$6$ & $-\dfrac{19\zeta(3)}{2\pi}+2\,\Cltwo\!\big(\tfrac{\pi}{3}\big)$ & $3$ / $19$\\[6pt]
$7$ & $-\dfrac{689\zeta(3)}{56\pi}+G-\tfrac52\Cltwo\!\big(\tfrac{2\pi}{7}\big)
      +\tfrac{11}{2}\Cltwo\!\big(\tfrac{4\pi}{7}\big)
      -\tfrac12\Cltwo\!\big(\tfrac{6\pi}{7}\big)$ & $5$ / $689$\\[6pt]
$8$ & $-\dfrac{31\zeta(3)}{2\pi}+4G$ & $3$ / $31$\\
\bottomrule
\end{tabular}
\caption{
The twisted superpotential constant map at low levels. The final column gives
the number of basis elements remaining after pruning and the largest
coefficient height in the reconstructed relation. Both quantities enter the precision requirement of the levelwise
integer-relation search discussed above and in Section~\ref{sec:method}. The closed
forms~\eqref{eq:clausen-odd}--\eqref{eq:clausen-even} were assembled using
only the levels $k\leq7$. The result at $k=8$, and every result at higher
level, is therefore a withheld prediction.}
\label{tab:CW}
\end{table}

The parity-dependent Clausen formulas are unified by the following integral.

\begin{theorem}[Constant map]
\label{thm:C}
For every real $k>0$, define
\be
\Cint(k)=-\frac{k^{2}\zeta(3)}{16\pi}
+\frac{2k}{\pi}\int_{0}^{\infty}\log(2\cosh x)\,\log\big(1-\e^{-kx}\big)\,\rd x\,.
\label{eq:CW}
\ee
The integral converges absolutely. For every positive integer $k$,
$\Cint(k)$ equals~\eqref{eq:clausen-odd} when $k$ is odd and
\eqref{eq:clausen-even} when $k$ is even.
\end{theorem}

Appendix~\ref{app:proof} gives a self-contained proof. The first step reduces
\eqref{eq:CW} exactly to the convergent alternating Euler series
\be
\Cint(k)
=
-\frac{k^{2}\zeta(3)}{16\pi}
-\frac{2\zeta(3)}{\pi k}
-\frac{k}{\pi}
\sum_{m\geq1}
\frac{(-1)^{m+1}}{m^{2}}\,
\Hgen_{2m/k}\,.
\label{eq:CEuler}
\ee
For integer $k$, the rational-argument harmonic sum is then evaluated by a
root-of-unity filter. The odd- and even-level evaluations require separate
identities, stated in Theorem~\ref{thm:Euler}. That theorem evaluates the sums
$\cE_{1,p}$ and $\cE_{2,p}$ of Eq.~\eqref{eq:E-def} in closed form and is an
independent arithmetic result, with no input from the Bethe-vacuum data. At
$k=2$, the series reduces to a classical alternating harmonic sum and the
result follows directly.

The function $\Cint(k)$ is also related to the constant map
$A(k)$ of the ABJM $S^3$ partition function
\cite{Marino:2011eh,Hanada:2012si,Hatsuda:2015owa,Hatsuda:2014vsa},
\be
A(k)
=
\frac{2\zeta(3)}{\pi^{2}k}
\left(1-\frac{k^{3}}{16}\right)
+
\frac{k^{2}}{\pi^{2}}
\int_{0}^{\infty}
\frac{x\log(1-\e^{-2x})}{\e^{kx}-1}\,
\rd x\,,
\qquad
k>0\,.
\label{eq:Ak}
\ee
Proposition~\ref{prop:CA} proves that
\be
\frac{\rd}{\rd k}
\left[
\frac{\Cint(k)}{k}
\right]
=
\frac{2\pi}{k^{2}}
\left[
A\!\left(\frac{k}{2}\right)-A(k)
\right]
-\frac{\zeta(3)}{4\pi}\,,
\qquad
k>0\,.
\label{eq:CAw}
\ee
Thus, the level derivative of $\Cint(k)/k$ is expressed as a finite difference
of the sphere constant map. The same pair $A(k)$ and $A(k/2)$ appears in the
constant term of the index below. Equation~\eqref{eq:CAw} is an analytic
relation among constant-map contributions, not an identity between the full
observables.

A direct numerical evaluation provides an independent check of
Theorem~\ref{thm:C}. At $70$-digit working precision, the finite Clausen
formula~\eqref{eq:clausen-odd}--\eqref{eq:clausen-even}, the Euler
series~\eqref{eq:CEuler}, and the integral~\eqref{eq:CW}
agree for every $1\leq k\leq50$; the largest pairwise discrepancy is
$6.4\times10^{-69}$ and occurs at $k=45$, as shown in
Figure~\ref{fig:constant_maps}(b). This check uses no Bethe-vacuum data.

The comparison with the computed plateaux provides a separate test of the
identification of this analytic constant with the finite-rank observable.
After subtracting the leading term in~\eqref{eq:Wparent} and the closed
constant map, the residuals at the largest available ranks are
$1.3\times10^{-107}$ at $(k,N_{\max})=(1,800)$,
$9.2\times10^{-37}$ at $(5,500)$, and
$7.4\times10^{-9}$ at $(50,500)$. Their growth with $k$ follows from the
increase of the physical nome at the largest available rank and does not
indicate a loss of accuracy in the analytic constant-map identity.

\begin{figure}[t]
\centering
\begin{minipage}[t]{0.485\textwidth}
\centering
\textbf{(a)}\\[1pt]
\includegraphics[width=\linewidth]{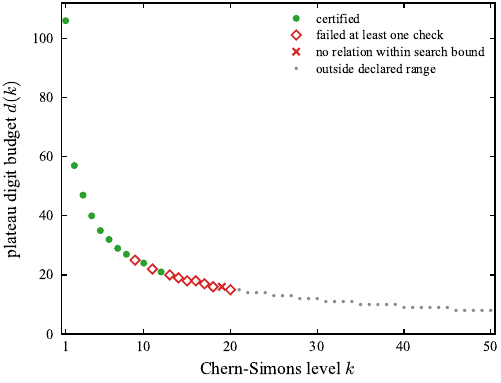}
\end{minipage}
\hfill
\begin{minipage}[t]{0.485\textwidth}
\centering
\textbf{(b)}\\[1pt]
\includegraphics[width=\linewidth]{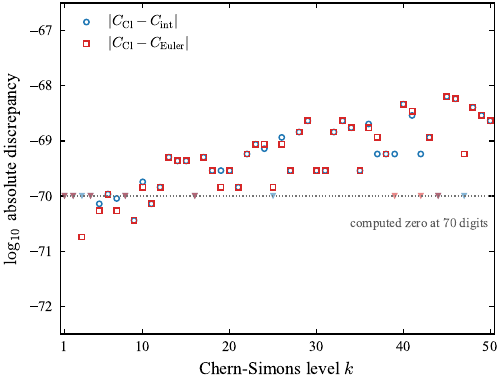}
\end{minipage}
\caption{
(a) The conservative plateau digit budget $d(k)$ decreases from $106$ digits
at $k=1$ to $8$ digits at $k=50$. Each marker represents one integer
Chern--Simons level. Filled circles denote relations that satisfy the four
robustness checks of the companion Letter~\cite{Hosseini:2026dkj} together
with the additional three-point plateau-tail criterion. Open diamonds denote
attempted candidates rejected by at least one required check. The cross at
$k=19$ indicates that no relation was found within the PSLQ search bound, and
grey points lie outside the declared levelwise search range $k\leq20$.
(b) The discrepancies between the finite Clausen
forms~\eqref{eq:clausen-odd}--\eqref{eq:clausen-even} and the Euler and
integral representations~\eqref{eq:CEuler} and~\eqref{eq:CW} are shown for
$1\leq k\leq50$ at $70$-digit working precision. Downward triangles mark
levels at which the computed discrepancy vanishes at this precision and is
placed at the censoring threshold. This panel tests the analytic identity of
Theorem~\ref{thm:C} without Bethe-vacuum data.}
\label{fig:constant_maps}
\end{figure}

\subsection{Type-IIA genus expansion}
\label{sec:Wgenus}

At fixed 't~Hooft coupling $\lambda=N/k$, the perturbative
parent~\eqref{eq:Wparent} admits the type-IIA genus expansion
\be
\im\Wt(N,k)\big|_{\rm pert}
\sim-\sum_{h\ge0}(2\pi\ii\lambda)^{2h-2}\,\cW_h(\lambda)\,N^{2-2h}\,,
\label{eq:Wgenus}
\ee
where $h$ denotes the type-IIA genus. Using
$\Nh=k(\lamh+2/(3k^2))$ and expanding at large $k$, one finds
\be
\cW_0(\lambda)
=
\frac{2\sqrt2\pi^{4}}{3}\lamh^{3/2}
-\frac{\pi\zeta(3)}{4}\,,
\qquad
\cW_1(\lambda)
=
-\frac{\pi^{2}}{3\sqrt2}\lamh^{1/2}
+\frac{\pi}{3}\log2\,,
\label{eq:W01}
\ee
and, for every $h=n+1\geq2$,
\bea
\cW_{n+1}(\lambda)
&=
\frac{(-1)^{n+1}\sqrt2}
{2^{n}3^{n+2}\pi^{2n-2}}
\binom{3/2}{n+1}
\lamh^{\frac12-n}
+ \omega_n \,,\\
\omega_n &= -
\frac{
2^{2n+1}(2^{2n}-1)\pi
\lvert \mathrm B_{2n}\mathrm B_{2n+2}\rvert
}{
n(2n+2)!
} \, .
\label{eq:Wgen}
\eea
For example,
\be
\cW_2(\lambda)
=
\frac{\sqrt2}{144}\lamh^{-1/2}
-\frac{\pi}{180}\,,
\qquad
\cW_3(\lambda)
=
\frac{\sqrt2}{5184\pi^{2}}\lamh^{-3/2}
-\frac{\pi}{3780}\,.
\ee

The two terms in~\eqref{eq:Wgen} have distinct origins. All dependence on
$\lamh$ follows from the binomial expansion
\be
 \left(\lamh+\frac{2}{3k^2}\right)^{3/2}
 =\sum_{r\ge0}\binom{3/2}{r}\lamh^{3/2-r}
 \left(\frac{2}{3k^2}\right)^r .
\ee
The genus-dependent constant $\omega_n$ descends from the constant map. After rescaling
$y=kx$ in~\eqref{eq:CW}, the expansion
\be
 \log(2\cosh z)=\log2+
 \sum_{n\ge1}\frac{2^{2n}(2^{2n}-1)\mathrm B_{2n}}{2n(2n)!}z^{2n} \, ,
\ee
together with
\be
 \int_0^\infty y^{2n}\log(1-\e^{-y})\,\rd y
 =-(2n)!\zeta(2n+2) \, ,
\ee
yields the large-$k$ asymptotic series in even powers of $1/k$. One Bernoulli
number arises from the Taylor expansion of $\log(2\cosh z)$ and the second
from the special value
\be
\zeta(2n+2)=\frac{(2\pi)^{2n+2}|\mathrm B_{2n+2}|}{2(2n+2)!} \, .
\ee
The integral identity follows by
expanding $-\log(1-\e^{-y})=\sum_{r\ge1}\e^{-ry}/r$, applying
Tonelli's theorem to the nonnegative exponential terms, and
evaluating
\be
\int_0^\infty y^{2n}\e^{-ry}\,\rd y=\frac{(2n)!}{r^{2n+1}} \, .
\ee

Since
$\lvert B_{2n}\rvert\sim2(2n)!/(2\pi)^{2n}$, the
$\lamh$-independent term in~\eqref{eq:Wgen} grows factorially with $n$.
The type-IIA genus expansion is therefore asymptotic. At fixed $(k,\lambda)$,
its truncation error decreases to an optimal order and then grows. Effects
exponentially small in the type-IIA string coupling lie beyond this
perturbative expansion.

Figure~\ref{fig:genusW} compares the truncated series with the closed
finite-rank parent for the first three cases listed below. At $\lambda=10$,
the optimal truncation occurs at $h=9$ for $k=10$, with error
$6.0\times10^{-8}$; at $h=16$ for $k=20$, with error $6.4\times10^{-15}$;
at $h=24$ for $k=30$, with error $7.9\times10^{-22}$; and at $h=48$ for
$k=60$, with error $1.9\times10^{-42}$.

\begin{figure}[t]
\centering
\includegraphics[width=0.67\linewidth]{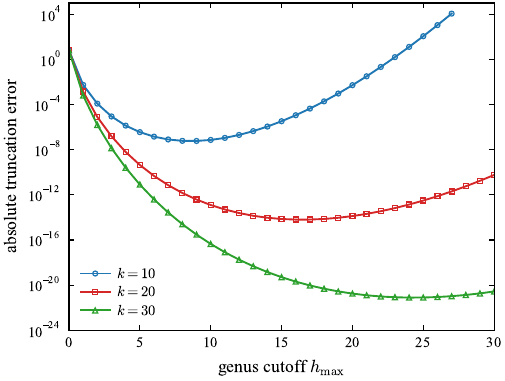}
\caption{Absolute truncation error of the type-IIA genus expansion of
$\im\Wt$ at $\lambda=10$ for the indicated values of $k$. The horizontal
axis is the integer genus cutoff $h_{\max}$. Each curve reaches a minimum at
finite order, as expected from the factorial growth of $\omega_n$ in
Eq.~\eqref{eq:Wgen}. The perturbative parent is evaluated directly, so the
comparison tests the large-$k$ coefficients derived from
Eqs.~\eqref{eq:Wparent} and~\eqref{eq:CW}.}
\label{fig:genusW}
\end{figure}

\subsection{The exponentially small tail}
\label{sec:Weichler}

At the three levels for which deep data are available, $k=1,2,4$, a
rank-differenced peeling procedure reconstructs, respectively,
\be
M_1^{\scriptscriptstyle(\mathcal W)}=51\,,
\qquad
M_2^{\scriptscriptstyle(\mathcal W)}=67\,,
\qquad
M_4^{\scriptscriptstyle(\mathcal W)}=93
\ee
rational exponential sectors. Through these orders, the data satisfy
\bea
\im\Wt(N,k)
&=\frac{\pi^{2}}{3}\sqrt{\frac{k}{2}}\,\Nh^{3/2}+\CW(k)
 +\frac{k^{2}}{2\pi}\sum_{m=1}^{M_k^{\scriptscriptstyle(\mathcal W)}}c^{(k)}_m
 \left(t+\frac1m\right)q^m
 \\
 &\quad
+\text{higher instanton orders } (m > M_k^{\scriptscriptstyle(\mathcal W)})\,,
\label{eq:Wfull}
\eea
where every reconstructed coefficient agrees with
\be
c_m^{(k)}
=
\frac{(-1)^{m+1}}{m^{2}}
\begin{cases}
5\left[
\sigma_3(m)-4\sigma_3(m/2)
\right],
& k=1,2,
\\[3pt]
\sigma_3(m)-16\sigma_3(m/4)\,,
& k=4.
\end{cases}
\label{eq:cm}
\ee
As usual, the convention $\sigma_3(x)=0$ for
$x\notin\mathbb Z_{>0}$ is understood.

We conjecture that the reconstructed coefficient formula continues to all
orders, giving
\be
\mathcal W_{\rm np}(N,k)
=
\frac{k^{2}}{2\pi}
\sum_{m\geq1}
c_m^{(k)}
\left(
t+\frac{1}{m}
\right)q^m\,.
\label{eq:Wnp}
\ee
The finite-order statement~\eqref{eq:Wfull}, including all instanton
orders $m\leq M_k^{\scriptscriptstyle(\mathcal W)}$, is reconstructed from
the computed data, while~\eqref{eq:Wnp} is the proposed continuation to
all $m\geq1$. Neither the combination $t+1/m$ nor the overall factor
$k^2/(2\pi)$ is imposed in the sector reconstruction. The design matrix
contains independent columns $q^m$ and $tq^m$, and the normalization is
selected from the cross-level family specified in
Section~\ref{sec:method}.

The divisor combinations in~\eqref{eq:cm} acquire a modular interpretation
only after the coefficients have been reconstructed. Work first in the shifted
nome $Q=-q$. Using \eqref{eq:E2-E4}, one has
\be
E_4(Q)-L^{2}E_4(Q^L)
=
1-L^{2}
+
240\sum_{m\geq1}
\left[
\sigma_3(m)
-
L^{2}\sigma_3\!\left(\frac{m}{L}\right)
\right]Q^m\,.
\label{eq:E4}
\ee
For $L=2$ and $L=4$, the bracketed coefficients reproduce the two divisor
factors in~\eqref{eq:cm}. The alternating sign is accounted for below when the
shifted nome is evaluated at $Q=-q$ and the overall signs of the physical
series are included. Thus, in the shifted $Q$-coordinate, the $k=1,2$
sequence is associated with the level-two modular form
\be
E_4(Q)-4E_4(Q^2)\in M_4(\Gamma_0(2))\,,
\ee
whereas the $k=4$ sequence is associated with the level-four modular form
\be
E_4(Q)-16E_4(Q^4)\in M_4(\Gamma_0(4))\,.
\label{eq:E4levelsQ}
\ee
Here ``level two'' and ``level four'' refer to modular level in the shifted
$Q$-coordinate. The spaces $M_4(\Gamma_0(2))$ and
$M_4(\Gamma_0(4))$ have dimensions $2$ and $3$, respectively. Their
cusp-form subspaces vanish~\cite{diamond2005first}.
The reconstructed divisor sequences
therefore select particular Eisenstein elements of these shifted-nome modular
spaces.

The same sequences have a different modular description when written in the
physical nome $q$. The half-period shift
$Q=-q=\e^{2\pi\ii(\tau+1/2)}$ changes the apparent modular level. Using
\cite{Dorigoni:2020oon}
\be
E_4(-q)
=
-E_4(q)+18E_4(q^2)-16E_4(q^4)\,,
\label{eq:E4halfperiod}
\ee
one obtains
\bea
E_4(-q)-4E_4(q^2)
&=
-E_4(q)+14E_4(q^2)-16E_4(q^4)\,,
\\
E_4(-q)-16E_4(q^4)
&=
-E_4(q)+18E_4(q^2)-32E_4(q^4)\,.
\label{eq:E4physical}
\eea
Consequently, the $k=1,2$ sequence has a level-two representation in the
shifted nome $Q$ and a level-four representation in the physical nome $q$,
whereas the $k=4$ sequence remains at level four in both coordinates. For
$k=4$, the half-period shift changes the Eisenstein combination but not the
modular subgroup, since all terms $E_4(q^d)$ with $d\mid4$ are forms on
$\Gamma_0(4)$. We retain the shifted nome $Q$ in the Eichler construction
below because it keeps the correspondence with the divisor combinations
in~\eqref{eq:cm} direct.

The form of the nonperturbative sectors naturally leads to an Eichler
primitive, as seen from
\begin{equation*}
\frac{1}{m^2}\left(t+\frac1m\right)
=
\frac{t}{m^2}+\frac{1}{m^3}\,,
\end{equation*}
and, with the logarithmic derivative $D$ defined in~\eqref{eq:Ddef}, a
coefficientwise third primitive of a weight-four Fourier series produces the
factor $m^{-3}$, while
the operator $tD$ restores the $t/m^2$ term. To make this explicit, define the
Lambert series
\be
\cL_3(Q)
=
\sum_{n\geq1}\Li_3(Q^n)
=
\sum_{m\geq1}\frac{\sigma_3(m)}{m^3}Q^m\,,
\qquad
|Q|<1\,.
\ee
The second equality follows by expanding
$\Li_3(Q^n) = \sum_{r \geq 1} Q^{n r}/r^3$ and collecting $m = n r$:
\begin{equation*}
\sum_{n \mid m} \frac{1}{(m/n)^3} = \frac{1}{m^3}\sum_{n \mid m} n^{3} = \frac{\sigma_3(m)}{m^3}\,.
\end{equation*}
It satisfies
\be
D^3\cL_3(Q)=\frac{E_4(Q)-1}{240}\,.
\ee
After removing the constant Fourier mode, the inverse-derivative convention
in~\eqref{eq:Dinverse} defines the third Eichler primitive of the nonconstant
Fourier part of the weight-four modular form
$E_4(Q)-L^2E_4(Q^L)$ on $\Gamma_0(L)$ by~\cite{pacsol2013modular}
\be
\widetilde E_L(Q)
\equiv
D^{-3}
\left[
E_4(Q)-L^2E_4(Q^L)-(1-L^2)
\right]
=
240
\left[
\cL_3(Q)
-\frac{1}{L}\cL_3(Q^L)
\right].
\label{eq:Etilde}
\ee

The conjectural continuation~\eqref{eq:Wnp} can then be written as
\bea
\mathcal W_{\rm np}(N,k)
&=
(1+tD)\Phi_k(q)\,,
\\
\Phi_1(q)
&=
-\frac{1}{96\pi}\widetilde E_2(-q)\,,
\qquad
\Phi_2(q)=4\Phi_1(q)\,,
\qquad
\Phi_4(q)
=
-\frac{1}{30\pi}\widetilde E_4(-q)\,.
\label{eq:Eichler}
\eea
By Eq.~\eqref{eq:Ddef}, the same operator acts as $D=q\partial_q$ on
functions of $-q$. Equation~\eqref{eq:Eichler} is then an
identity coefficient by coefficient once~\eqref{eq:cm} is assumed. The
substitution $Q=-q$ supplies the factor $(-1)^m$, while the overall signs in
$\Phi_k$ produce the alternating factor $(-1)^{m+1}$ appearing
in~\eqref{eq:cm}.

The candidate coefficients grow at most linearly. The inequality
$\sigma_3(m)<\zeta(3)m^3$, together with their explicit expressions,
implies $\lvert c_m^{(k)}\rvert\le \mathsf C_k m$ for some
constant $\mathsf C_k>0$. This gives the upper bound
\be
\limsup_{m\to\infty}\lvert c_m^{(k)}\rvert^{1/m}
\le
\lim_{m\to\infty}(\mathsf C_km)^{1/m}
=1\,.
\ee
For the matching lower bound, let $\mathfrak p$ run over the odd primes. Along this
subsequence,
\be
c_\mathfrak p^{(1)}=c_\mathfrak p^{(2)}
=
5\left(\mathfrak p+\mathfrak p^{-2}\right),
\qquad
c_\mathfrak p^{(4)}
=
\mathfrak p+\mathfrak p^{-2},
\ee
and therefore
\be
\limsup_{m\to\infty}\lvert c_m^{(k)}\rvert^{1/m}
\ge
\lim_{\substack{\mathfrak p\to\infty\\
                 \mathfrak p\ \mathrm{odd\ prime}}}
\lvert c_{\mathfrak p}^{(k)}\rvert^{1/\mathfrak p}
=1\,.
\ee
Combining the two bounds yields
\be
\limsup_{m\to\infty}
\big|c_m^{(k)}\big|^{1/m}
=1\,.
\label{eq:unitradius}
\ee
Since $m^{-1/m}\to1$ as $m\to\infty$, the same root-test limit holds for
$c_m^{(k)}/m$. Thus the two power series
$\sum_{m\geq1}c_m^{(k)}q^m$ and
$\sum_{m\geq1}c_m^{(k)}q^m/m$ appearing in~\eqref{eq:Wnp}
both have unit radius of convergence. Since $t=-\log q$, the complete
expression is well defined after choosing a branch of the logarithm on a
simply connected subdomain of $0<|q|<1$, but does not extend holomorphically
to $q=0$.

\begin{figure}[t]
\centering
\includegraphics[width=0.67\linewidth]{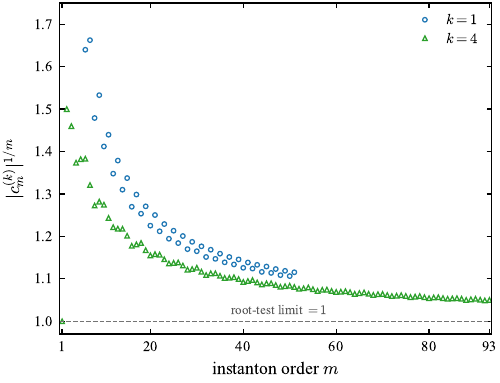}
\caption{Root-test values $|c_m^{(k)}|^{1/m}$ through $m=51$ at $k=1$ and
through $m=93$ at $k=4$. For $k=1$, the values at $m=1,\ldots,5$ lie above
the displayed vertical range. The $k=2$ coefficients coincide with the $k=1$
sequence through their $51$ common reconstructed orders and reproduce
Eq.~\eqref{eq:cm} through all $67$ orders reconstructed at $k=2$. They are
omitted to avoid overplotting. The data approach the unit-radius limit implied
by the divisor-sum continuation~\eqref{eq:unitradius}. This provides a
finite-depth diagnostic of the conjectural all-order formula.}
\label{fig:growthW}
\end{figure}

The coincidence of the reconstructed $k=1$ and $k=2$ sequences is suggestive
in view of the $\cN=8$ enhancement, although the present reconstruction does
not derive it from supersymmetry. The $k=4$ sequence is comparably simple
without a corresponding enhancement. In the shifted nome $Q=-q$, its divisor
law is the nonconstant Fourier part of the level-four form
$E_4(Q)-16E_4(Q^4)$, and the candidate exponential tail is obtained from the
corresponding third Eichler primitive. In the physical nome $q$, the same
sequence is represented by the level-four combination
$-E_4(q)+18E_4(q^2)-32E_4(q^4)$. The absence of cusp forms in
$M_4(\Gamma_0(4))$ strongly constrains this modular representation once the
divisor law is known; it does not explain why the selected Bethe-vacuum
observable realizes that law. The present deep reconstruction is restricted to $k=1,2,4$, so whether
comparable divisor-sum closures occur at other levels remains open.
Deriving the coefficient law from the Bethe equations, or identifying a
symmetry or geometric mechanism responsible for the $k=4$ simplification,
would clarify whether it is isolated or belongs to a broader non-$\cN=8$
arithmetic pattern among these Bethe observables.

\section{The Cardy limit of the superconformal index}
\label{sec:Cardy}

Evaluated on the same Bethe solution, the effective twisted
superpotential of Section~\ref{sec:W} also controls the leading term in the
Cardy limit of the superconformal index. Let
$\mathsf x$ denote the angular fugacity and set
\be
\mathsf x=\e^{\ii\pi\omega}=\e^{-\beta}\,,
\qquad
\beta=-\ii\pi\omega\in\mathbb R_{>0}\,.
\label{eq:SCIvariables}
\ee
The fugacities $y_a=\e^{\ii\Delta_a}$ are held fixed at the universal point
$\Delta_a=\pi/2$. In the finite-$N$ Cardy/Bethe relation below, the
$\beta$-independent term is the contribution of the same Bethe vacuum to the
topologically twisted index at the associated universal sphere fluxes
$\mathfrak n_a=1/2$.

\subsection{The Cardy coefficient}
\label{sec:CardyA}

Let $Z_{\rm vac}(N,k)\equiv Z_{\rm data}(N,k)$ denote the contribution of the selected Bethe vacuum to
the genus-zero topologically twisted index, and define its real free energy by%
\footnote{The free energy $F(N,k)$ of the full $k$-fold $\mathbb Z_k$ orbit,
defined in~\eqref{eq:orbit}, is related to the single-vacuum quantity by
$F(N,k)=F_{\rm vac}(N,k)-\log k$.}
\be
F_{\rm vac}(N,k)
\equiv
-\re\log Z_{\rm vac}(N,k)
=
-\log\big|Z_{\rm vac}(N,k)\big|\,.
\label{eq:Fvac}
\ee
Let $\cI_{\rm vac}(\omega;N,k)$ denote the contribution to the
superconformal index associated with the same Bethe vacuum. The finite-$N$
Bethe formulation gives
\cite{Choi:2019dfu,Bobev:2022wem,Bobev:2024mqw}
\bea
\log\cI_{\rm vac}(\omega;N,k)
&=
-\frac{1}{\pi\omega}\,\im\Wt(N,k)
-F_{\rm vac}(N,k)
+\ii\phi
+O(\omega)
\\
&=
\frac{\ii}{\beta}\,\mathcal C(N,k)
-F_{\rm vac}(N,k)
+\ii\phi
+O(\beta)\,,
\label{eq:CardyA}
\eea
as $\omega\to\ii0^+$, or equivalently $\beta\to0^+$. Here
\be
\mathcal C(N,k)
\equiv
\im\Wt(N,k)>0
\label{eq:CardyC}
\ee
on the selected branch. The quantity
$\phi\in\mathbb R/(2\pi\mathbb Z)$ denotes the remaining
$\beta$-independent phase. We keep it explicit because its value can depend
on localization conventions and logarithm branches.

The strict Cardy factor associated with this saddle is the local expression
\be
\cI_{\rm C}(\beta;N,k)
\equiv
\exp\!\left(\frac{\ii\mathcal C(N,k)}{\beta}\right).
\label{eq:Cardybeta}
\ee
It retains only the leading $1/\beta$ term in~\eqref{eq:CardyA}. On the
positive real $\beta$ ray this factor is a pure phase, so the leading Cardy
term does not by itself produce exponential growth of $|\cI_{\rm C}|$ along that ray.

At $k=1,2,4$, the Cardy coefficient is the same on-shell quantity reconstructed
in Section~\ref{sec:Weichler}. Adopting the all-order conjecture
\eqref{eq:Wnp}, we write
\be
\mathcal C(N,k)
=
\frac{\pi^{2}}{3}\sqrt{\frac{k}{2}}\,\Nh^{3/2}
+\CW(k)
+\frac{k^{2}}{2\pi}
\sum_{m\geq1}
c_m^{(k)}
\left(t+\frac1m\right)q^m \, ,
\label{eq:Aexact}
\ee
where the coefficients $c_m^{(k)}$ are given in~\eqref{eq:cm}. The series is
absolutely convergent for $0<q<1$.

\subsection{Large-charge coefficients}
\label{sec:CardyInv}
\label{sec:CardyLimits}

At fixed $y_a$, choose the normalization of $\mathsf x$ so that its conjugate
charge has unit spacing, and write the complete superconformal index as
\be
\cI_{\rm full}(\mathsf x;N,k)
=
\sum_{\ell\geq0}
a(\ell;N,k)\,\mathsf x^\ell\,.
\label{eq:Iexpansion}
\ee
Other charge normalizations can be treated by first passing to a
uniformizing fugacity. Let $\Gamma_{\mathsf x}$ be a positively oriented contour about
$\mathsf x=0$. Cauchy's formula reads
\be
a(\ell;N,k)
=
\frac{1}{2\pi\ii}
\oint_{\Gamma_{\mathsf x}}
\frac{\rd\mathsf x}{\mathsf x^{\ell+1}}\,
\cI_{\rm full}(\mathsf x;N,k)\,.
\label{eq:CauchySCI}
\ee

The strict Cardy factor~\eqref{eq:Cardybeta} is a local approximation near
$\mathsf x=1$. We therefore define $a^{\rm C}(\ell;N,k)$ as the local contribution obtained by
replacing the full integrand in~\eqref{eq:CauchySCI} by the strict Cardy factor
near $\beta=0$.

This step requires analytic assumptions. We assume that, on a chosen branch of
$\beta=-\log\mathsf x$, the Cardy expansion can be continued into a sector
containing the first-quadrant saddle found below, and that a corresponding
local cycle can be reached by deforming the coefficient contour without
crossing an obstruction. We also assume that the local Cardy approximation remains valid
along the relevant part of the deformed cycle. Using
$\rd\mathsf x=-\mathsf x\,\rd\beta$, and denoting the correspondingly oriented
local contour in the $\beta$-plane by $\Gamma_\beta$, gives
\be
a^{\rm C}(\ell;N,k)
=
\frac{1}{2\pi\ii}
\int_{\Gamma_\beta}
\rd\beta\,
\exp\!\left(
\ell\beta+\frac{\ii\mathcal C(N,k)}{\beta}
\right).
\label{eq:CardyIntegral}
\ee

For $\ell>0$, if the rescaled local contour is homologous to a positively
oriented loop about the origin, this integral evaluates to
\be
a^{\rm C}(\ell;N,k)
=
\sqrt{\frac{\ii\mathcal C(N,k)}{\ell}}\,
I_1\!\left(
2\sqrt{\ii\mathcal C(N,k)\ell}
\right).
\label{eq:degeneracy}
\ee
For a complex variable $z$ on the chosen square-root branch, the modified Bessel function is normalized here by
\be
\frac{1}{2\pi\ii}
\oint_{\mathscr C_0}
\rd s\,
\exp\!\left(s+\frac{z}{s}\right)
=
\sqrt z\,I_1(2\sqrt z)\,,
\label{eq:BesselDefinition}
\ee
where $\mathscr C_0$ encircles $s=0$ counterclockwise.

On the principal branch $\sqrt{\ii}=\e^{\ii\pi/4}$, define
\be
\chi_\ell
=
\sqrt{2\mathcal C(N,k)\ell}\,.
\label{eq:chiell}
\ee
For
\be
S(\beta)
=
\ell\beta+\frac{\ii\mathcal C(N,k)}{\beta}\,,
\ee
the first-quadrant saddle is
\be
\beta_{\rm s}
=
\e^{\ii\pi/4}
\sqrt{\frac{\mathcal C(N,k)}{\ell}}\,,
\qquad
S(\beta_{\rm s})
=
(1+\ii)\chi_\ell\,.
\label{eq:saddle}
\ee
The argument of the Bessel function is
$2\sqrt{\ii\mathcal C(N,k)\ell}=(1+\ii)\chi_\ell$, which lies in the sector where
the standard exponentially growing asymptotic of $I_1$ applies. Hence, at
fixed $(N,k)$ and as $\ell\to\infty$,
\be
a^{\rm C}(\ell;N,k)
\sim
\frac{\mathcal C(N,k)^{1/4}}
{2\sqrt{\pi}\,\ell^{3/4}}\,
\e^{\chi_\ell}\,
\e^{\ii(\chi_\ell+\pi/8)}
\left[1+O(\chi_\ell^{-1})\right].
\label{eq:degeneracyAsymptotic}
\ee
Its real part has the additive asymptotic expansion
\be
\re a^{\rm C}(\ell;N,k)
=
\frac{\mathcal C(N,k)^{1/4}}
{2\sqrt{\pi}\,\ell^{3/4}}\,
\e^{\chi_\ell}
\left[
\cos\!\left(\chi_\ell+\frac{\pi}{8}\right)
+O(\chi_\ell^{-1})
\right].
\label{eq:degeneracyReal}
\ee
The error is written additively because a relative estimate is not uniform
near the zeros of the cosine.

Self-consistency of the Cardy approximation at the saddle requires
$|\beta_{\rm s}|\ll1$, or
\be
\ell\gg\mathcal C(N,k)\,.
\label{eq:regime}
\ee
Since $\mathcal C(N,k)\sim N^{3/2}$ at fixed $k$, a necessary condition for a
Cardy limit uniform in rank is
\be
\frac{\ell}{N^{3/2}}\longrightarrow\infty\,.
\ee
This condition places the saddle in the Cardy region. Sufficiency also
requires uniform control of the $N$-dependence of the terms omitted
in~\eqref{eq:CardyA}; in particular, the large-charge and large-rank limits
need not commute.

The quantity $a^{\rm C}(\ell;N,k)$ is the local contribution obtained from the
strict $1/\beta$ Cardy factor, not the coefficient of the complete index.
Restoring the $\beta$-independent term in~\eqref{eq:CardyA} multiplies
\eqref{eq:degeneracy} by
$\e^{-F_{\rm vac}(N,k)+\ii\phi}$, while terms suppressed by positive powers of
$\beta$ generate subleading large-charge corrections. Combining
Eqs.~\eqref{eq:orbit} and~\eqref{eq:Zparent}, at fixed $k$ and large
$N$ one has
\be
\e^{-F_{\rm vac}(N,k)}
=
\exp\!\left[
-\frac{\pi\sqrt{2k}}{3}\,\Nh^{3/2}
+O(\Nh^{1/2})
\right].
\ee
This factor is independent of $\ell$ at fixed $(N,k)$ but is exponentially
suppressed in rank, and it must be retained as part of the complete saddle
amplitude.

The complete coefficient contour may contain conjugate or competing saddles,
or may fail to pass through the saddle~\eqref{eq:saddle} on the chosen sheet.
If a second saddle contributes the complex conjugate of the dressed local
term, the pair gives
\be
2\re\!\left[
\e^{-F_{\rm vac}(N,k)+\ii\phi}
a^{\rm C}(\ell;N,k)
\right],
\ee
whereas other saddles may modify or cancel this contribution. Before
specializing $y_a$, the fully charge-refined coefficients are integer-valued
supertraces. At $y_a=\e^{\ii\pi/2}$, the coefficients
$a(\ell;N,k)$ become phase-weighted sums and are not guaranteed to be
non-negative or real. Equation~\eqref{eq:degeneracy} exactly evaluates
the strict-factor contour integral under the stated contour assumption; its
identification with a contribution to a coefficient of the complete index is
conditional and asymptotic. It is therefore neither an exact integer
coefficient nor a non-negative degeneracy. Related issues of contour and sheet
dependence in three-dimensional indices are discussed in
Ref.~\cite{ArabiArdehali:2025bub}.

\section{The topologically twisted index}
\label{sec:Z}

The selected-orbit index is evaluated on the same Bethe root as the twisted
superpotential and has a different finite-rank completion. At fixed level, its
perturbative parent contains three rank-dependent terms, its exponential
sectors are quadratic in the instanton action, and one of its two independent
nonperturbative generators has a modular-product candidate. The candidate
families are motivated by the reconstructed arithmetic and finite-order growth;
within them, exact low-order algebra fixes the integer parameters, and every
remaining reconstructed coefficient verifies the result. Assuming the all-order
continuation, the convergence radius and leading coefficient growth follow
analytically.

\subsection{Fixed-level parent and rank-independent term}
\label{sec:Zpert}

At fixed $k$, the perturbative parent is
\be
F(N,k)\big|_{\rm pert}
=
\frac{\pi\sqrt{2k}}{3}
\left(
\Nh^{3/2}-\frac{3}{k}\Nh^{1/2}
\right)
+\frac12\log\Nh-\fzero(k)\,.
\label{eq:Zparent}
\ee
This agrees with~\cite[Eq.~(17)]{Bobev:2022jte} after specializing that result to
$\mathfrak g=0$. Independently of that input, unrestricted four-term fits at the
representative levels $k=1,2,3,5,10,20,50$ reconstruct the functional form
directly from the data. At each of these levels, we determine the
level-dependent fit coefficients
$\mathsf F_{3/2},\mathsf F_{1/2},\mathsf F_{\log},\mathsf F_0$ from the ansatz
\be
\mathsf F_{3/2}\Nh^{3/2}
+\mathsf F_{1/2}\Nh^{1/2}
+\mathsf F_{\log}\log\Nh
+\mathsf F_0
\ee
on the four largest available ranks. At low levels the result is
\be
\mathsf F_{3/2}=\frac{\pi\sqrt{2k}}{3}\,,
\qquad
\mathsf F_{1/2}=-\pi\sqrt{\frac{2}{k}}\,,
\qquad
\mathsf F_{\log}=\frac12
\ee
to working precision. The loss of digits at larger $k$ tracks
the increasing unpeeled instanton remainder. For example,
$\mathsf F_{\log}/(1/2)=1-5.5\times10^{-11}$ at $k=20$, while
$\mathsf F_{\log}/(1/2)=1-3.6\times10^{-4}$ at $k=50$. The $\Nh^{1/2}$ and
$\log\Nh$ terms, absent from~\eqref{eq:Wparent}, are therefore resolved by the
index data alone.

References~\cite{Bobev:2022jte,Bobev:2022eus} do not fix the rank-independent term. Within
the declared six-feature search space and height bounds, its plateau is exactly
reconstructed as
\be
\fzero(k)
=
2A(k)-6A\!\left(\frac{k}{2}\right)-\frac{3\CW(k)}{\pi}
-\frac{k^{2}\zeta(3)}{2\pi^{2}}+\frac12\log k-\frac52\log2\,.
\label{eq:f0arith}
\ee
Here $A(k)$ is defined in~\eqref{eq:Ak}, and $\CW(k)$ is given by
Theorem~\ref{thm:C}. In the ordered feature basis
\be
\left\{A(k),A(k/2),\CW(k)/\pi,k^2\zeta(3)/\pi^2,\log k,\log2\right\},
\ee
the coefficient vector
\be
\left(2,-6,-3,-\frac12,\frac12,-\frac52\right)
\ee
was fixed on $k=1,\ldots,6$ and tested on nine withheld levels. The residuals
are $3.3\times10^{-29}$ at $k=7$, $1.0\times10^{-23}$ at $k=10$,
$1.9\times10^{-15}$ at $k=20$, and $5.6\times10^{-8}$ at $k=50$. Each reaches
the instanton floor at that level, so the withheld tests constrain coefficient
errors down to the unpeeled nonperturbative tail and no further.

Within this search class, Eq.~\eqref{eq:f0arith} identifies the index constant
map as a combination of the sphere constant maps at levels $k$ and $k/2$, the
twisted-superpotential constant map, and elementary terms. The declared
six-dimensional feature space is guided by structures already present in
neighboring localization observables and by the scales appearing in the
type-IIA genus expansion. The half-level entry $A(k/2)$ has two specific precedents
available before fitting the index constant. Writing the $S^3$ rank shift as
$B_{S^3}(k)=k/24+1/(3k)$ and the corresponding Bethe-observable shift as
$B_{\rm Bethe}(k)=k/24-2/(3k)$, one finds
$B_{\rm Bethe}(k)=2[B_{S^3}(k)-B_{S^3}(k/2)]$. The same pair of levels also
appears in the constant maps through Eq.~\eqref{eq:CAw}, where the level
derivative of $C_{\rm int}(k)/k$ involves $A(k)-A(k/2)$. Together, these
occurrences provide a concrete precedent for including $A(k/2)$ among the
candidate features in the search for the index plateau. The rational
coefficient vector in~\eqref{eq:f0arith} is determined from the index data.

The arithmetic expression~\eqref{eq:f0arith} also admits the one-kernel
representation
\bea
\fzero(k)
&=
-\frac{3\zeta(3)}{8\pi^{2}}k^{2}
+\frac76\log k
+f_0
+\frac{k}{\pi^{2}}
\int_{0}^{\infty}
\cK(x)\log\bigl(1-\e^{-kx}\bigr)\,\rd x\,,
\\
\cK(x)
&=
6x\tanh x
+4\log\frac{\sinh x}{x}
+4\left(\frac{x}{\tanh x}-1\right),
\\
f_0
&=
-8\zeta'(-1)-\frac{23}{6}\log2-\frac23\log\pi\,.
\label{eq:f0kernel}
\eea
Here $\zeta'(s)=\rd\zeta(s)/\rd s$. Theorem~\ref{thm:f0} proves that
\eqref{eq:f0arith} and~\eqref{eq:f0kernel} are the same function of $k$. The proof substitutes the integral
representations of $\CW$ and $A$, integrates by parts once, and evaluates three
elementary moments. The residual $1/k$ dependence cancels identically, and the
remaining constant follows from the Riemann zeta functional equation. As a
numerical check, the two formulas agree to $2.5\times10^{-119}$ at
$120$-digit working precision over fourteen levels.

\subsection{Type-IIA genus expansion}
\label{sec:Zgenus}

With the convention of~\eqref{eq:Wgenus}, let $F_h(\lambda;k)$ denote the
coefficient at type-IIA genus $h$. The perturbative parent has the asymptotic
expansion
\be
F(N,k)\big|_{\rm pert}
\sim
-\sum_{h\geq0}
(2\pi\ii\lambda)^{2h-2}
F_h(\lambda;k)N^{2-2h}\,.
\ee
Expanding~\eqref{eq:Zparent} at fixed $\lambda=N/k$ gives
\be
F_0(\lambda)
=
\frac{4\sqrt2\pi^{3}}{3}\lamh^{3/2}
+\frac32\zeta(3)\,,
\qquad
F_1(\lambda;k)
=
\frac{2\sqrt2\pi}{3}\lamh^{1/2}
-\frac12\log\lamh
+\frac23\log k
+f_0\,,
\label{eq:F01}
\ee
and, for $h=n+1\geq2$,
\bea
F_{n+1}(\lambda)
&=
\frac{1}{n\,2^{n+1}3^{n}\pi^{2n}\lamh^{n}}
+
\frac{(-1)^{n}\sqrt2}{2^{n}3^{n}\pi^{2n-1}}
\left[
\binom{1/2}{n}-\frac29\binom{3/2}{n+1}
\right]
\lamh^{\frac12-n}
+d_n\,,
\\
d_n
&=
\frac{2^{2n+2}|\mathrm B_{2n}\mathrm B_{2n+2}|}{(2n+2)!}
\left(3\cdot4^{n}-1+\frac1n\right).
\label{eq:Fgen}
\eea

The three terms in~\eqref{eq:Fgen} have distinct origins. Since $\Nh=k(\lamh+2/(3k^2))$,
the logarithmic contribution in~\eqref{eq:Zparent} expands as
\bea
\frac12\log\Nh
&=
\frac12\log k
+\frac12\log\lamh
+\frac12\log\left(
1+\frac{2}{3k^2\lamh}
\right)\\
&=
\frac12\log k
+\frac12\log\lamh
+\frac12\sum_{n\geq1}
\frac{(-1)^{n+1}}{n}
\left(
\frac{2}{3k^2\lamh}
\right)^n.
\label{eq:logNhexp}
\eea
For $h=n+1\geq2$, its term of order $k^{-2n}$ is
\be
\frac{(-1)^{n+1}2^{n-1}}
{n\,3^n\lamh^n}\,\frac{1}{k^{2n}}\,.
\ee
Using $N=\lambda k$ to convert to the normalization of the genus expansion
gives
\be
\frac{1}
{n\,2^{n+1}3^n\pi^{2n}\lamh^n}\,,
\ee
which is the first term in~\eqref{eq:Fgen}. The second comes from the large-$k$
expansion of the shifted $\Nh^{3/2}$ and $\Nh^{1/2}$ terms in
\eqref{eq:Zparent}, while $d_n$ descends from the constant map
\eqref{eq:f0kernel}.

Only $F_1$ retains explicit level dependence. The $\frac12\log k$ from
$\frac12\log\Nh$ combines with the $-\frac76\log k$ from $-\fzero(k)$ to
produce $-\frac23\log k$ in the parent and hence $+\frac23\log k$ in
$F_1$. The constants $d_n$ grow factorially, so the type-IIA genus expansion
is asymptotic. At $\lambda=10$, the optimal truncation occurs at
$h=8,16,24,47$ for $k=10,20,30,60$, with minimum errors
$8.8\times10^{-7}$, $1.9\times10^{-13}$, $3.5\times10^{-20}$, and
$1.7\times10^{-40}$, respectively. Figure~\ref{fig:genusZ} displays the first
three cases.

\begin{figure}[t]
\centering
\includegraphics[width=0.67\linewidth]{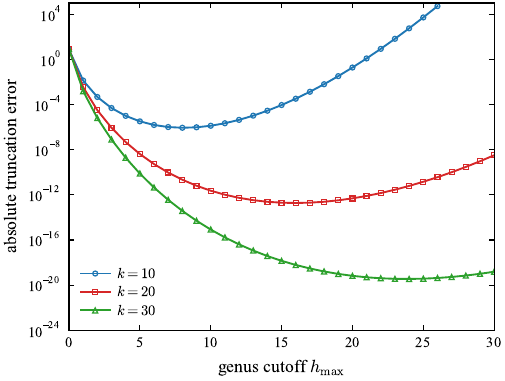}
\caption{Absolute truncation error of the type-IIA expansion of the
selected-orbit index parent at $\lambda=10$ for $k=10,20,30$, measured against
the closed expression~\eqref{eq:Zparent}. The horizontal axis is the integer
genus cutoff $h_{\max}$. Each curve reaches a finite optimal order and then
turns upward, as expected from the factorial growth of $d_n$ in
Eq.~\eqref{eq:Fgen}. The comparison tests the derived large-$k$ coefficients
and their implementation.}
\label{fig:genusZ}
\end{figure}

\subsection{Instanton sectors}
\label{sec:Zsectors}

For $k=1,2,4$, the finite-depth data obey
\bea
F(N,k)
&=
F(N,k)\big|_{\rm pert}
-
\sum_{m=1}^{M_k^{\scriptscriptstyle(F)}}
\left[
R_m^{(k)}
+\frac{k^2}{\pi^2}
\left(
A_m^{(k)}t^2
+B_m^{(k)}\left(t+\frac1m\right)
\right)
\right]q^m
\\
&\quad
+\text{higher instanton orders } (m>M_k^{\scriptscriptstyle(F)})\,,
\label{eq:Znp}
\eea
where $M_k^{\scriptscriptstyle(F)}$ denotes the reconstructed depth of the
index sectors, with
\be
(M_1^{\scriptscriptstyle(F)},M_2^{\scriptscriptstyle(F)},M_4^{\scriptscriptstyle(F)})=(17,15,15)\,.
\ee
We conjecture that the same sector structure continues to all orders,
\be
F_{\rm np}(N,k)
=
-\sum_{m\geq1}
\left[
R_m^{(k)}
+\frac{k^2}{\pi^2}
\left(
A_m^{(k)}t^2+B_m^{(k)}\left(t+\frac1m\right)
\right)
\right]q^m\,.
\label{eq:Znp-continuation}
\ee
Only the coefficients through $m=M_k^{\scriptscriptstyle(F)}$ are
reconstructed from the data. The reduction below to the two generators
$\cR_k(q)$ and $\cU_k(q)$ separates the part admitting a
modular-product candidate from the part that remains unresolved.

Three structural features of~\eqref{eq:Znp} emerge from the search itself.
First, each sector is quadratic in $t$. Fitting the first sector
with trial polynomial degrees $P=0,1,2,3$ and testing on withheld ranks gives
residual ratios $66$, $48$, $3.3\times10^{-5}$, and $3.3\times10^{-5}$ at
$k=1$, while the fitted cubic coefficient is $4.6\times10^{-215}$. Degree two
is therefore selected, and degree three adds no resolved contribution. Second,
the common normalization $k^2/\pi^2$ is selected from the declared cross-level
family, with the power of $k$ fixed across levels.

Third, the fitted linear and constant terms combine as $t+1/m$. To state this
test explicitly, write the unrestricted $m$-th sector as
\be
\left(
\mathsf f_2t^2+\mathsf f_1t+\mathsf f_0
\right)q^m\,.
\ee
Comparison with~\eqref{eq:Znp} gives
\be
A_m^{(k)}
=-\frac{\pi^2}{k^2}\mathsf f_2\,,
\qquad
B_m^{(k)}
=-\frac{\pi^2}{k^2}\mathsf f_1\,,
\qquad
R_m^{(k)}
=\frac{\mathsf f_1}{m}-\mathsf f_0\,.
\label{eq:ABRfromfit}
\ee
Thus the $t+1/m$ structure is tested by the rational reconstruction of
$\mathsf f_1/m-\mathsf f_0$. It succeeds in every resolved sector, whereas the
unconstrained constant coefficient $\mathsf f_0$ does not pass the same
rationality test.

We reconstruct $17$, $15$, and $15$ sectors at $k=1,2,4$, respectively. The
two leading sectors are shown in Table~\ref{tab:leading}; the complete set of
reconstructed coefficients is listed in Table~\ref{tab:Zsectors}.

\begin{table}[ht]
\centering\small
\renewcommand{\arraystretch}{1.2}
\begin{tabular}{@{}crrrrrr@{}}
\toprule
 & \multicolumn{3}{c}{$k=1,2$} & \multicolumn{3}{c}{$k=4$}\\
\cmidrule(lr){2-4}\cmidrule(lr){5-7}
$m$ & $A_m^{(1,2)}$ & $B_m^{(1,2)}$ & $R_m^{(1,2)}$ & $A_m^{(4)}$ & $B_m^{(4)}$ & $R_m^{(4)}$\\
\midrule
$1$ & $-\tfrac{13}{2}$   & $-10$  & $76$    & $-\tfrac12$      & $-1$  & $20$\\
$2$ & $\tfrac{1441}{4}$  & $316$  & $-3590$ & $\tfrac{49}{4}$  & $10$  & $-262$\\
\bottomrule
\end{tabular}
\caption{The two leading exactly reconstructed index sectors in the
normalization of~\eqref{eq:Znp}. The sequences at $k=1$ and $k=2$ coincide,
whereas those at $k=4$ do not. Further orders are listed in
Table~\ref{tab:Zsectors}.}
\label{tab:leading}
\end{table}

The $k=1$ and $k=2$ sequences coincide term by term, with no relative factor,
through all fifteen sectors reconstructed at both levels. This is the
index counterpart of the shared twisted superpotential sequence in~\eqref{eq:cm}. The
smaller reconstruction depth relative to $\Wt$ has a structural origin. Each
index sector contains three independent features rather than two, and its
coefficient heights grow much faster. The numerator of $A_{17}^{(1)}$, for
example, already has $32$ digits.

\subsection{A finite-order relation among the sectors}
\label{sec:Brel}

The reconstructed coefficients satisfy
\be
m\,B_m^{(k)}
=
\begin{cases}
3A_m^{(k)}+\dfrac{R_m^{(k)}}{8}\,, & k=1,2\,,\\[10pt]
3A_m^{(k)}+\dfrac{R_m^{(k)}}{16}
-\dfrac{3\,\sigma_1^{\rm odd}(m)}{4m}\,, & k=4\,.
\end{cases}
\label{eq:Brel}
\ee
For $k=1,2$, the relation is fixed on the two lowest sectors and verified on
the remaining fifteen sectors at $k=1$ and thirteen at $k=2$. At $k=4$, it is
fixed on the three lowest sectors and verified on the remaining twelve. These
are exact identities in $\mathbb Q$ over the reconstructed ranges
$m\leq17,15,15$. The appearance
of $\sigma_1^{\rm odd}$ at $k=4$ and its absence at $k=1,2$
is a further trace of the level-two versus level-four distinction of
\eqref{eq:cm} in the shifted $Q$-coordinate.

Introduce the level-dependent formal generating functions
\be
\cR_k(q)=\sum_{m\geq1}R_m^{(k)}q^{m}\,,
\qquad
\cU_k(q)=\sum_{m\geq1}U_m^{(k)}q^{m}\,,
\qquad
U_m^{(k)}\equiv m\,A_m^{(k)}\,.
\label{eq:RUdef}
\ee
Only their coefficients through $m=M_k^{\scriptscriptstyle(F)}$ are currently known.
Over these ranges, Eq.~\eqref{eq:Brel} reduces the
independent nonperturbative data from three sequences to the pair
$(\cR_k,\cU_k)$. We now identify a modular-product continuation for $\cR_k$ within
two data-informed integer-parameter families, and then exhaust a specified
library for $\cU_k$.

\subsection{A modular-product conjecture for the first generator}
\label{sec:Rmodular}

Only a finite initial segment of $\cR_k$ is known, consisting of
$(M_1^{\scriptscriptstyle(F)},M_2^{\scriptscriptstyle(F)},
M_4^{\scriptscriptstyle(F)})=(17,15,15)$ coefficients at $k=1,2,4$.
These coefficients do not by themselves determine a unique all-order
continuation or its radius of convergence. We therefore distinguish the
finite-order information supplied by the reconstruction, the choice of
modular families, and the analytic consequences that follow once a
continuation is specified.

\paragraph{Arithmetic and finite-order analytic clues.}
Four observations constrain a useful candidate class. First, the signs of
$R_m^{(k)}$ alternate over every resolved order. This selects the shifted nome
$Q=-q$ as the natural series variable and agrees with the half-period shift
that organizes the twisted superpotential. The conventions for $q$, $Q$,
$\tau_2$, and $\tau_4$ are fixed in Eqs.~\eqref{eq:Qnome}--\eqref{eq:tauconv}.

Second, the reconstructed coefficients of the twisted superpotential display
modular organisation at the same Chern--Simons levels. Equation~\eqref{eq:cm}
contains the divisor combinations
$\sigma_3(m)-4\sigma_3(m/2)$ at $k=1,2$ and
$\sigma_3(m)-16\sigma_3(m/4)$ at $k=4$, which are the nonconstant Fourier
coefficients of weight-four Eisenstein forms on $\Gamma_0(2)$ and
$\Gamma_0(4)$, respectively. This does not imply that the index is modular,
but it makes product blocks of shifted-nome modular levels two and four a
natural restricted class to examine.

Third, the finite relation~\eqref{eq:Brel} contains additional weight-two
level-four arithmetic at $k=4$. Its inhomogeneous term is generated by
\be
\mathscr H_{\rm odd}(z)
=
\sum_{m\geq1}\frac{\sigma_1^{\rm odd}(m)}{m}\,z^m
=
\log\frac{\Eul(z^2)}{\Eul(z)}
\label{eq:sigmaodd}
\ee
as an identity of formal power series.%
\footnote{Expanding the Euler product gives
$\log\Eul(z)=-\sum_{m\geq1}\sigma_1(m)z^m/m$. Hence the coefficient of
$z^m$ in $\log[\Eul(z^2)/\Eul(z)]$ is
$[\sigma_1(m)-2\sigma_1(m/2)]/m$. Writing
$m=2^{v_2(m)}u$ with $u$ odd and using multiplicativity of $\sigma_1$ gives
$\sigma_1(m)-2\sigma_1(m/2)=\sigma_1(u)=\sigma_1^{\rm odd}(m)$, with the
convention~\eqref{eq:divisor}.}
Using Eq.~\eqref{eq:Eul-derivative} and the Euler-product identity
\be
\Eul(-Q)=\frac{\Eul(Q^2)^3}{\Eul(Q)\Eul(Q^4)}
\ee
one obtains
\be
1+24D\mathscr H_{\rm odd}(-Q)
=
E_2(Q)-4E_2(Q^2)+4E_2(Q^4)
=
-\bigl(2E_2(Q^2)-E_2(Q)\bigr)
+2\bigl(2E_2(Q^4)-E_2(Q^2)\bigr)\,.
\label{eq:sigmaoddmod}
\ee
The first bracket is the weight-two Eisenstein series on $\Gamma_0(2)$ in
$\tau_2$; the second is its pullback under $\tau_2\mapsto2\tau_2$ and lies on
$\Gamma_0(4)$. Thus the $k=4$ relation supplies shifted-nome level-four data
that are absent at $k=1,2$.

Fourth, the reconstructed coefficients provide finite-order information about
the exponential scale and the type of a possible nearest singularity. For an
all-order series $\cR_k(q)=\sum_{m\geq1}R_m^{(k)}q^m$, the
Cauchy--Hadamard theorem gives the reciprocal radius of convergence as
\be
\limsup_{m\to\infty}|R_m^{(k)}|^{1/m}\,.
\ee
The available coefficient windows are finite and cannot determine this
limsup. We therefore use the root and successive-ratio quantities
\be
\bigl|R_m^{(k)}\bigr|^{1/m},
\qquad
\bigl|R_m^{(k)}/R_{m-1}^{(k)}\bigr|
\ee
as finite-order diagnostics of the exponential scale. Neither diagnostic is
guaranteed to converge for a general power series.

The resolved coefficients contain an exact factor $1/m$. More precisely,
\be
r_m^{(k)}=(-1)^{m+1}mR_m^{(k)}
\label{eq:rdef}
\ee
is a positive integer throughout each reconstructed range. Removing the
$1/m$ factor produces stable finite-order estimates of the exponential scale.
At $k=1$,
\be
\bigl(r_{17}^{(1)}\bigr)^{1/17}=85.019695223\ldots,
\qquad
\frac{r_{17}^{(1)}}{r_{16}^{(1)}}=85.019695223\ldots,
\ee
which agree with $\e^{\pi\sqrt2}=85.019695223\ldots$ to $23$ and $20$
significant digits, respectively. The corresponding root and ratio at $k=4$
agree with $\e^\pi=23.140692633\ldots$ to $17$ and $16$ significant digits.
These data identify
\be
\rho_1=\rho_2=\e^{\pi\sqrt2},
\qquad
\rho_4=\e^\pi
\label{eq:rhovalues}
\ee
as candidate reciprocal radii.

The raw diagnostics provide an independent finite-order check of these
candidates. For an exact logarithmic coefficient law
\be
|R_m|=\frac{\rho^m}{m}\,,
\ee
one has identically
\be
|R_m|^{1/m}
=
\rho\,m^{-1/m}\,,
\qquad
\left|\frac{R_m}{R_{m-1}}\right|
=
\rho\,\frac{m-1}{m}\,.
\label{eq:rawdeficit}
\ee
Both indicators lie below $\rho$, with relative deficits
$1-m^{-1/m}\simeq(\log m)/m$ and $1/m$. At $m=17$ these deficits are
$15.4\%$ and $5.9\%$. Substituting $\rho_1=\e^{\pi\sqrt2}$ predicts
\be
\rho_1\,17^{-1/17}=71.9681262253\ldots,
\qquad
\rho_1\,\frac{16}{17}=80.0185366807\ldots.
\ee
The reconstructed $k=1$ coefficients give
$|R_{17}^{(1)}|^{1/17}=71.9681262253\ldots$ and
$|R_{17}^{(1)}/R_{16}^{(1)}|=80.0185366807\ldots$, in agreement with these
predictions to ten digits. This agreement supports the candidate scale selected by the normalized
coefficients $r_m^{(k)}$, but does not establish a radius of convergence.

The same normalization also provides finite-order evidence for the type of
singularity. Since $r^{(k)}_m>0$ across the reconstructed range, the sequence
shows no modulation beyond the overall alternation. A single
dominant singularity on the negative $q$ axis would have this
sign pattern, while a conjugate pair would generally produce
modulation. Finitely many terms cannot exclude a modulation that appears only
at higher order. Among the elementary singular models, a coefficient prefactor
$m^{-1}$ is characteristic of a logarithm.%
\footnote{For
$\mathscr F(q)=C_\star(1-q/q_\star)^{-\alpha}$ with
$\alpha\notin\{0,-1,-2,\ldots\}$, the binomial series gives
$[q^m]\mathscr F\sim
[C_\star/\Gamma(\alpha)]m^{\alpha-1}q_\star^{-m}$, whereas
$[q^m]\log(1-q/q_\star)^{-1}=m^{-1}q_\star^{-m}$. Transfer to a general
function with the same local expansion requires the usual analytic-continuation
hypotheses of singularity analysis~\cite[Secs.~VI.2--VI.5]{flajolet2009analytic}.
These hypotheses are not assumed at this stage.}
The normalized ratio $r_m^{(k)}/\rho_k^m$ approaches one rapidly over the
available range. Its deviation from unity is $1.4\times10^{-22}$ at the last
$k=1$ order and $8.4\times10^{-17}$ at the last $k=4$ order. This
matches the logarithmic model with unit strength. Since
\be
\log(1+\rho_k q)=\sum_{m\geq1}(-1)^{m+1}\frac{\rho_k^{m}}{m}\,q^{m}\,,
\ee
the local model suggested by the data is
\be
\cR_k(q)\approx\log(1+\rho_k q)
\quad\text{near}\quad
q_\star=-\rho_k^{-1}.
\label{eq:darboux}
\ee
In the shifted nome, the corresponding candidate zero of the logarithm's
argument lies at $Q_\star=\rho_k^{-1}$, namely
$\e^{-\pi\sqrt2}$ for $k=1,2$ and $\e^{-\pi}$ for $k=4$.

\paragraph{Restricted modular families.}
The preceding observations motivate an outer logarithm, an Euler-product
block of the appropriate shifted-nome level, and an argument capable of
vanishing at the candidate point. They do not determine a unique function.
No automated search over a broad modular dictionary was performed for
$\cR_k$. Instead, the analysis considered the following two restricted
families after the full coefficient windows had been examined.

For the common $k=1,2$ sequence, define
\be
P_p(Q)=\prod_{n\geq1}(1+Q^n)^p
=\left(\frac{\Eul(Q^2)}{\Eul(Q)}\right)^p\,,
\qquad p\in\mathbb Z\,.
\label{eq:Ppdef}
\ee
Equation~\eqref{eq:Eul-derivative} gives
\be
D\log P_p(Q)
=
\frac{p}{24}\bigl[2E_2(Q^2)-E_2(Q)-1\bigr]\,.
\label{eq:Pp-level2}
\ee
Thus $D\log P_p$ is the nonconstant part of the weight-two Eisenstein series
$2E_2(Q^2)-E_2(Q)$ on $\Gamma_0(2)$. The product $P_p$ is nonzero for
$|Q|<1$, so $\log P_p$ alone cannot produce the inferred interior logarithmic
singularity. We restricted the search to the two-term Laurent deformation
\be
P_p(Q)-\frac{cQ}{P_p(Q)}\,.
\ee
It has unit constant term, uses the same level-two product block and its
inverse, and can acquire a simple zero inside the unit disc. The restriction
reflects a choice of candidate family and carries no claim of uniqueness for
the architecture. The integrality of the normalized reconstructed coefficients
motivates $p,c\in\mathbb Z$.

For $k=4$, the divisor structure in Eqs.~\eqref{eq:cm} and
\eqref{eq:sigmaoddmod} instead points to a shifted-nome level-four block. The
unit-normalized Euler quotient supported on the divisors one and four is
\be
\frac{\Eul(Q)}{\Eul(Q^4)}\,,
\qquad
D\log\frac{\Eul(Q)}{\Eul(Q^4)}
=
\frac{1}{24} \bigl[E_2(Q)-4E_2(Q^4)+3\bigr]\,.
\label{eq:Eul14-level4}
\ee
This quotient is nonzero for $|Q|<1$ and therefore cannot supply the candidate
zero. Within this product architecture, a separate weight-zero factor is required.

The modular lambda function enters at this point. By the conventions and group
argument in Section~\ref{sec:natural-variables},
\be
\mathfrak h_4(Q)\equiv\lambda_{\rm mod}(-Q)
=\lambda_{\rm mod}(2\tau_2+1)
\label{eq:h4def}
\ee
is a Hauptmodul for $\Gamma_0(4)$ in the shifted parametrization and satisfies
$\mathfrak h_4(Q)=O(Q)$ at the cusp. Since the compact modular curve $X_0(4)$ has genus
zero~\cite{diamond2005first}, rational weight-zero modular functions at this
level can be expressed rationally in $\mathfrak h_4$. Within this function field, the
simplest nonconstant polynomial factor with value one at the cusp is
$1+c\mathfrak h_4(Q)$. This choice defines a restricted linear family and is not an
exhaustive classification.

The candidate zero gives a sharper clue. In the theta convention
$Q=\e^{\ii\pi\tau_4}$, $Q_\star=\e^{-\pi}$ corresponds to $\tau_4=\ii$,
while $-Q_\star$ corresponds to $\tau_4=\ii+1$. Using
Eq.~\eqref{eq:lambda-T},
\be
\mathfrak h_4(Q_\star)
=
\lambda_{\rm mod}(\ii+1)
=
\frac{\lambda_{\rm mod}(\ii)}{\lambda_{\rm mod}(\ii)-1}
=-1\,.
\label{eq:h4star}
\ee
Within the chosen linear factor, the condition
$1+c \mathfrak h_4(Q_\star)=0$ therefore predicts $c=1$ exactly. The point
$Q_\star$ is still inferred from finite-order coefficients here, so this is a
data-informed prediction inside the restricted family. We keep $c$ symbolic
below in order to test that prediction independently using exact coefficient
algebra.

The two candidate families are
\bea
\cR^{[2]}_{p,c}(Q)
&=
\log\!\left[P_p(Q)-\frac{cQ}{P_p(Q)}\right],
\\[3pt]
\cR^{[4]}_{p,c}(Q)
&=
\log\!\left[
\bigl(1+c\lambda_{\rm mod}(-Q)\bigr)
\left(\frac{\Eul(Q)}{\Eul(Q^4)}\right)^p
\right],
\qquad p,c\in\mathbb Z\,.
\label{eq:Rfamilies}
\eea
The logarithm branches are fixed to vanish at $Q=0$. The first family is used
for the common $k=1,2$ sequence and the second for $k=4$.

\paragraph{Exact parameter determination within the chosen families.}
For the shared $k=1,2$ sequence, let
$a_m(p,c)=[Q^m]\cR^{[2]}_{p,c}(Q)$. The first three coefficients are
\be
a_1=p-c\,,
\qquad
a_2=\frac{p+4pc-c^2}{2}\,,
\qquad
a_3=\frac{4p+3pc-6p^2c+6pc^2-c^3}{3}\,.
\label{eq:Rparameter12}
\ee
Because $Q=-q$, the reconstruction gives
$(a_1,a_2,a_3)=(-76,-3590,-614704/3)$. The first equation yields
$c=p+76$, and the second reduces to
\be
(p+12)(p+39)=0\,.
\ee
The two integral solutions are $(-12,64)$ and $(-39,37)$. The
corresponding values of $a_3$ are $-614704/3$ and $-713146/3$, respectively,
so the reconstructed third coefficient selects
\be
(p,c)=(-12,64)\,.
\ee
With these parameters fixed, the remaining fourteen coefficients at $k=1$
and twelve at $k=2$ are reproduced exactly over $\mathbb Q$ by
$\cR^{[2]}_{-12,64}(Q)$.

At $k=4$, let $b_m(p,c)=[Q^m]\cR^{[4]}_{p,c}(Q)$. The first two coefficients
are
\be
b_1=-p-16c\,,
\qquad
b_2=-\frac{3p}{2}-128c-128c^2\,.
\label{eq:Rparameter4}
\ee
There are two complementary parameter determinations. The candidate-zero
condition gives $c=1$ by Eq.~\eqref{eq:h4star}, and the first reconstructed
coefficient $b_1=-20$ then fixes $p=4$. The coefficients
$m=2,\ldots,15$ are subsequently reproduced exactly over $\mathbb Q$.

Independently of the candidate zero,
matching the first two exact coefficients $(b_1,b_2)=(-20,-262)$ gives
\be
p=20-16c\,,
\qquad
(c-1)(16c+29)=0\,.
\ee
The two rational solutions are $(p,c)=(4,1)$ and $(49,-29/16)$, and
$(4,1)$ is the unique integral pair.
The remaining thirteen coefficients are then reproduced exactly over
$\mathbb Q$ by $\cR^{[4]}_{4,1}(Q)$.
The coefficient-only calculation therefore recovers the same integer
parameters selected by the candidate-zero condition.

The uniqueness statements are global over integer pairs inside the two
families in Eq.~\eqref{eq:Rfamilies}. They do not establish uniqueness among
other modular, quasimodular, theta-product, or nonmodular architectures. The
family choice is data-informed, while the parameter determination and holdout
checks are exact finite-order statements.

With the selected parameters, define
\be
\mathcal A(Q)=P_{-12}(Q)\,,
\qquad
\Xi(Q)=1+\lambda_{\rm mod}(-Q)
=
\frac{2\lambda_{\rm mod}(Q)-1}{\lambda_{\rm mod}(Q)-1}\,.
\label{eq:Xidef}
\ee
The second equality follows from the $T$-transformation in
Eq.~\eqref{eq:lambda-T}. The exact agreement over the resolved coefficients
motivates the conjecture that $\cR_k(q)$ admits the all-order continuation
\be
\cR_k(q)
=
\begin{cases}
\displaystyle
\log\!\left[
\mathcal A(Q)-\dfrac{64Q}{\mathcal A(Q)}
\right],
& k=1,2,
\\[12pt]
\displaystyle
\log\!\left[
\Xi(Q)
\left(\dfrac{\Eul(Q)}{\Eul(Q^4)}\right)^4
\right],
& k=4,
\end{cases}
\qquad Q=-q\,.
\label{eq:Rmodular}
\ee
The branch is fixed by $\cR_k(0)=0$. Equation~\eqref{eq:Rmodular} agrees
coefficientwise with the reconstruction modulo
$q^{M_k^{\scriptscriptstyle(F)}+1}$ for
$(M_1^{\scriptscriptstyle(F)},M_2^{\scriptscriptstyle(F)},
M_4^{\scriptscriptstyle(F)})=(17,15,15)$.

\subsection{The nearest singularity}
\label{sec:sing}

For $k=1,2$, the conjecture can be rewritten as an eta quotient, making its
nearest singularity accessible analytically. With the nome convention
$Q=\e^{2\pi\ii\tau_2}$ from~\eqref{eq:tauconv} and the eta normalization
in~\eqref{eq:eta-def}, define
\be
\mathfrak f_2(\tau_2)
=
\left(
\frac{\eta(\tau_2)}{\eta(2\tau_2)}
\right)^{24}
\label{eq:hauptmodul}
\ee
the Hauptmodul of $\Gamma_0(2)$~\cite{van_Ekeren_2018}. Since
\be
\prod_{n\geq1}(1+Q^n)
=
Q^{-1/24}\frac{\eta(2\tau_2)}{\eta(\tau_2)}\,,
\ee
one obtains the branch-independent identity
$\mathcal A(Q)^2=\mathfrak f_2(\tau_2)Q$ and hence
\be
\mathcal A(Q)-\frac{64Q}{\mathcal A(Q)}
=
\mathcal A(Q)\,\frac{\mathfrak f_2(\tau_2)-64}{\mathfrak f_2(\tau_2)}\,.
\label{eq:factorized}
\ee
This factorization avoids introducing separate square roots of $Q$ and
$\mathfrak f_2(\tau_2)^{-1}$.
Because $\mathcal A(Q)$ and $\mathfrak f_2(\tau_2)$ are nonzero in the upper half-plane,
the zeros of the logarithm's argument are exactly the solutions of
$\mathfrak f_2(\tau_2)=64$.

\begin{proposition}[Nearest singularity and coefficient growth of $\cR_k$]
\label{prop:radius}
Assume the all-order continuation~\eqref{eq:Rmodular}. Then $\cR_k(q)$ is
holomorphic on $|q|<|q_\star|$ and singular at $q_\star$, where
\be
q_\star=-\e^{-\pi\sqrt2}
\quad (k=1,2)\,,
\qquad
q_\star=-\e^{-\pi}
\quad (k=4)\,.
\label{eq:qstar}
\ee
Consequently,
\be
\limsup_{m\to\infty}|R_m^{(k)}|^{1/m}
=|q_\star|^{-1}\,,
\ee
namely $\e^{\pi\sqrt2}\approx85.0197$ for $k=1,2$ and
$\e^{\pi}\approx23.1407$ for $k=4$. Moreover, the singularity at
$q_\star$ is a simple logarithmic branch point. In a neighborhood of $q_\star$,
\be
\cR_k(q)=\log\Bigl(1-\frac{q}{q_\star}\Bigr)+\mathcal H_k(q)\,,
\label{eq:localform}
\ee
with $\mathcal H_k$ holomorphic at $q_\star$, and the coefficients obey
\be
R^{(k)}_m=(-1)^{m+1}\frac{\rho_k^{m}}{m}\bigl(1+O(\theta_k^{m})\bigr)\,,
\qquad
\theta_k=
\begin{cases}
\e^{-2\pi\sqrt2/3},&k=1,2,\\
\e^{-4\pi/5},&k=4,
\end{cases}
\label{eq:Rasym}
\ee
with $\rho_k=|q_\star|^{-1}$ as in~\eqref{eq:rhovalues}.
\end{proposition}

\begin{proof}
Consider first $k=1,2$. By~\eqref{eq:factorized}, the argument of the
logarithm is
$\mathcal A(Q)\bigl(\mathfrak f_2(\tau_2)-64\bigr)/\mathfrak f_2(\tau_2)$.
The eta function is non-vanishing in the upper half-plane, so both
$\mathcal A(Q)$ and $\mathfrak f_2(\tau_2)$ are nonzero there. Near $Q=0$,
one has $\mathfrak f_2(\tau_2)\sim Q^{-1}$ and the argument tends to one,
consistently with $\cR_k(0)=0$. Its zeros in $0<|q|<1$ are therefore
exactly the solutions of $\mathfrak f_2(\tau_2)=64$.

Let $X_0(2)$ and $X(2)$ denote the compact modular curves associated with
$\Gamma_0(2)$ and $\Gamma(2)$, respectively. Because $X_0(2)$ has genus
zero~\cite{diamond2005first} and $\mathfrak f_2$ is a Hauptmodul~\cite{van_Ekeren_2018}, the solution set is a
single $\Gamma_0(2)$-orbit.
Under the Fricke involution
$W_2:\tau_2\mapsto-1/(2\tau_2)$, Eq.~\eqref{eq:eta-S} gives
\be
\mathfrak f_2\!\left(-\frac{1}{2\tau_2}\right)
=
\frac{2^{12}}{\mathfrak f_2(\tau_2)}\,.
\ee
At the fixed point $\tau_2=\ii/\sqrt2$, this gives
$\mathfrak f_2(\ii/\sqrt2)^2=2^{12}$. Positivity of the eta quotient on the
positive imaginary axis selects $\mathfrak f_2(\ii/\sqrt2)=64$, so the
solution set is the $\Gamma_0(2)$-orbit of
$\ii/\sqrt2$.

The point of this orbit closest to $q=0$ is the one with largest
$\im\tau_2$, because $|q|=|Q|=\e^{-2\pi\im\tau_2}$. For
$\mathsf M=\begin{pmatrix}\mathsf a&\mathsf b\\\mathsf c&\mathsf d\end{pmatrix}\in\Gamma_0(2)$,
\be
\im(\mathsf M\tau_2)
=
\frac{\im\tau_2}{|\mathsf c\tau_2+\mathsf d|^2}\,,
\label{eq:orbit-im}
\ee
with $\mathsf c$ even. If $\mathsf c=0$, $\mathsf M$ is a translation and the imaginary part
remains $1/\sqrt2$. If $|\mathsf c|\geq2$, then at $\tau_2=\ii/\sqrt2$,
\be
|\mathsf c\tau_2+\mathsf d|^2
=
\frac{\mathsf c^2}{2}+\mathsf d^2
\geq3\,,
\ee
where $\mathsf d$ is odd. Every non-translational image therefore has strictly smaller
imaginary part. The maximum occurs at $\tau_2=\ii/\sqrt2+\mathsf b$, $\mathsf b\in\mathbb Z$,
which gives $Q=\e^{-\pi\sqrt2}$ and
$q=-Q=-\e^{-\pi\sqrt2}$. The logarithm's argument is analytic and nonzero on
$|q|<\e^{-\pi\sqrt2}$ and tends to one at the origin, so it has a unique
analytic logarithm there with $\cR_k(0)=0$. At $q_\star$ the argument vanishes,
and $\cR_k$ is singular.

For $k=4$, the Euler-product factor in~\eqref{eq:Rmodular} is nonzero for
$|Q|<1$. The modular lambda function omits the value $1$ in the upper
half-plane, so the denominator in~\eqref{eq:Xidef} produces no interior pole.
The zeros are therefore the solutions of $\lambda_{\rm mod}(Q)=1/2$. With
$Q=\e^{\ii\pi\tau_4}$, $\lambda_{\rm mod}$ is a Hauptmodul for $\Gamma(2)$ and
$\lambda_{\rm mod}(\ii)=1/2$, so the solution set is the $\Gamma(2)$-orbit of
$\ii$. For $\mathsf M\in\Gamma(2)$, $\mathsf c$ is even and $\mathsf d$ is odd. If $\mathsf c=0$, the
image differs from $\ii$ by an even translation and retains imaginary part
one. If $|\mathsf c|\geq2$, then at $\tau_4=\ii$,
\be
|\mathsf c\tau_4+\mathsf d|^2=\mathsf c^2+\mathsf d^2\geq5\,,
\ee
and hence $\im(\mathsf M\tau_4)\leq1/5$. The maximum is attained at
$\tau_4=\ii+\mathsf b$, $\mathsf b\in2\mathbb Z$, giving $Q=\e^{-\pi}$ and
$q=-Q=-\e^{-\pi}$. The same analytic-logarithm argument proves that this is
the nearest singularity, now using $|q|=|Q|=\e^{-\pi\im\tau_4}$.

It remains to show that the zero at $q_\star$ is simple. Both dominant points
have trivial projective stabilizer in the relevant group. For $\mathsf M\in
\Gamma_0(2)$ fixing $\tau_{2,\star}=\ii/\sqrt2$, the fixed-point equation
$\mathsf c\tau^2+(\mathsf d-\mathsf a)\tau-\mathsf b=0$ with $\tau^2=-\tfrac12$ separates into $\mathsf a=\mathsf d$ and
$\mathsf c=-2\mathsf b$, and the determinant condition becomes $\mathsf a^2+2\mathsf b^2=1$, whose only
integer solutions are $\mathsf b=0$, $\mathsf a=\pm1$, leaving only $\mathsf M=\pm\mathbb I$.
Similarly, for $\mathsf M\in\Gamma(2)$ fixing $\tau_{4,\star}=\ii$ one finds
$\mathsf a=\mathsf d$, $\mathsf c=-\mathsf b$ and $\mathsf a^2+\mathsf b^2=1$; since $\mathsf a$ must be odd, again only
$\mathsf M=\pm\mathbb I$ survives. Neither point is an elliptic point of
its group, and neither is a cusp, so the quotient maps are unramified there.
Since $\mathfrak f_2$ and $\lambda_{\rm mod}$ are Hauptmoduls for $X_0(2)$ and $X(2)$
respectively, each is a local coordinate near the corresponding point, whence
$\mathfrak f_2'(\tau_{2,\star})\neq0$ and $\lambda_{\rm mod}'(\tau_{4,\star})\neq0$. Consequently
$\mathfrak f_2-64$ and $\lambda_{\rm mod}-\tfrac12$, and with them the arguments of the two
logarithms in~\eqref{eq:Rmodular}, have simple zeros at $Q_\star=-q_\star$.

Writing that argument as $(Q-Q_\star)\,g_k(Q)$ with $g_k$ holomorphic and nonzero
at $Q_\star$, and using
$Q-Q_\star=-Q_\star(1-Q/Q_\star)=-Q_\star(1-q/q_\star)$,
gives~\eqref{eq:localform} with $\mathcal H_k=\log\bigl(-Q_\star g_k\bigr)$ holomorphic at
$q_\star$. Extracting coefficients from $\log(1-q/q_\star)$ yields
$-q_\star^{-m}/m$, and $q_\star=-\rho_k^{-1}$ gives the leading term in
\eqref{eq:Rasym}.

To control the remainder, organize the other zeros by modular orbit shells.
For $k=1,2$, write $\mathsf c=2n$ with $n\in\mathbb Z$. At $\tau_{2,\star}=\ii/\sqrt2$ the denominator in
\eqref{eq:orbit-im} is
\be
 |\mathsf c\tau_{2,\star}+\mathsf d|^2=2n^2+\mathsf d^2.
\ee
For a non-translational image $n\neq0$ and $\mathsf d$ is odd. The smallest value is
$3$, attained for example by
$\mathsf M_{\rm sh}=\begin{psmallmatrix}1&0\\2&1\end{psmallmatrix}\in\Gamma_0(2)$;
the next possible value is at least $9$. Hence the second shell lies at
$\varrho_{2,1}=\varrho_{2,2}=\e^{-\pi\sqrt2/3}$. For $k=4$, the corresponding denominator is
$4n^2+\mathsf d^2$; its smallest non-translational value is $5$, again attained by
$\mathsf M_{\rm sh}\in\Gamma(2)$, and the next possible value is at least $13$. Thus
$\varrho_{2,4}=\e^{-\pi/5}$. Modulo translations, each of these shells contains only
finitely many orbit points. Triviality of the projective stabilizer is
preserved under conjugation, so all associated zeros are simple. The
discreteness of the next denominator
value provides a radius $\varrho_{3,k}>\varrho_{2,k}$ before the following shell.

Let $\mathcal Z_{2,k}$ denote the finite set of zeros on $|q|=\varrho_{2,k}$. With
all logarithms normalized to vanish at the origin, the function
\be
 \mathcal G_k(q)=\cR_k(q)-\log\!\left(1-\frac{q}{q_\star}\right)
 -\sum_{q_j\in\mathcal Z_{2,k}}
  \log\!\left(1-\frac{q}{q_j}\right)
\label{eq:second-shell-subtraction}
\ee
extends holomorphically to $|q|<\varrho_{3,k}$. This follows after differentiation.
The logarithmic derivative of the argument in \eqref{eq:Rmodular}
is meromorphic, with residue one at every simple zero, and
the displayed subtraction cancels all poles in that disc. Consequently,
\be
 R_m^{(k)}=-\frac{q_\star^{-m}}{m}
 -\frac{1}{m}\sum_{q_j\in\mathcal Z_{2,k}}q_j^{-m}
 +[q^m]\mathcal G_k(q)\,.
\label{eq:second-shell-coefficients}
\ee
For any $r$ with $\varrho_{2,k}<r<\varrho_{3,k}$, Cauchy's estimate gives
$[q^m]\mathcal G_k=O(r^{-m})$. Dividing~\eqref{eq:second-shell-coefficients} by the
leading term shows that the second shell contributes
$O((|q_\star|/\varrho_{2,k})^m)$, while the holomorphic remainder contributes
$O(m(|q_\star|/r)^m)$, which is smaller. Since
$|q_\star|/\varrho_{2,k}=\theta_k$, this proves
\eqref{eq:Rasym}.
\end{proof}

The rate in~\eqref{eq:Rasym} is rapid. Since $\theta_1=\theta_2\approx0.0517$
and $\theta_4\approx0.0810$, each additional instanton order
suppresses the second-shell scale by more than an order of magnitude. This accounts for the
accuracy of the finite-order diagnostics of Section~\ref{sec:Rmodular}. At
$k=1$ and $m=17$, $\theta_1^{17}\approx1.4\times10^{-22}$, while the
reconstructed coefficient differs from $(-1)^{m+1}\rho_1^m/m$ by
$1.4\times10^{-22}$ in relative terms. At $k=4$ and $m=15$, the corresponding
numbers are $4\times10^{-17}$ and $8\times10^{-17}$. Their agreement in scale
provides a consistency check of the all-order conjecture~\eqref{eq:Rmodular}.

We also locate the first zero numerically without using the analytic value
in~\eqref{eq:qstar} to initialize the search. Let $\mathscr P_k(q)$ denote
the argument of the logarithm in~\eqref{eq:Rmodular}. The product
representations above show that $\mathscr P_k$ is holomorphic for $|q|<1$.
For any $0<r<1$ such that $\mathscr P_k$ has no zero on $|q|=r$, the argument
principle gives
\be
\mathsf N_k(r)
=
\frac{1}{2\pi\ii}
\oint_{|q|=r}
\rd q\,
\frac{\mathscr P_k'(q)}{\mathscr P_k(q)}\,,
\label{eq:argument-principle-check}
\ee
where $\mathsf N_k(r)$ is the number of zeros of $\mathscr P_k$ in
$|q|<r$, counted with multiplicity. We evaluate $\mathsf N_k(r)$ on circles
of increasing radius to identify the first annulus containing a zero. The
first nonzero count is one, and bisection on the negative real axis then
isolates that zero.
The result reproduces~\eqref{eq:qstar} with the discrepancies shown below.
\begin{center}
\small
\renewcommand{\arraystretch}{1.2}
\begin{tabular}{@{}crr@{}}
\toprule
Working digits
& \multicolumn{1}{c}{$k=1,2$}
& \multicolumn{1}{c}{$k=4$}\\
\midrule
$80$  & $2.5\times10^{-83}$  & $2.9\times10^{-83}$\\
$150$ & $5.0\times10^{-153}$ & $7.3\times10^{-153}$\\
$250$ & $1.2\times10^{-253}$ & $3.7\times10^{-253}$\\
\bottomrule
\end{tabular}
\end{center}

Assuming the all-order continuation~\eqref{eq:Rmodular},
Proposition~\ref{prop:radius} establishes
$\rho_1=\rho_2=\e^{\pi\sqrt2}$ and $\rho_4=\e^\pi$ as the exact reciprocal
radii; equivalently, the radii of convergence are $\e^{-\pi\sqrt2}$ and
$\e^{-\pi}$. The corresponding convergence discs are strictly smaller than the unit
discs of the two ordinary power series underlying the twisted-superpotential
conjecture~\eqref{eq:Wnp}. The proof is analytic and does not
rely on convergence of the finite-order root or ratio
diagnostics. Those diagnostics provide
quantitative checks of the coefficient asymptotics~\eqref{eq:Rasym}.

At $k=1,2,4$, every computed physical point lies inside these convergence
discs. On the physical locus, the condition $|q|<|q_\star|$ is equivalent to
$\Nh>k/4$ for $k=1,2$ and to $\Nh>1/2$ for $k=4$.

The second generator has the same candidate exponential scale but a different
power-law prefactor. Its successive-coefficient ratios approach
$-\rho_k$, reaching $-85.16$ at $k=1$ and $-23.21$ at $k=4$ at the largest
reconstructed orders. The coefficients $U_m^{(k)}$ lack the additional $1/m$ suppression
present in $R_m^{(k)}$. The finite data are consistent with
$(-1)^{m+1}U_m^{(k)}/\rho_k^m$ approaching a constant.
Over the reconstructed range, the ratios
$U_m^{(k)}/(m R_m^{(k)})$ drift smoothly, and a second-order extrapolation in
$1/m$ gives
\be
\frac{U_m^{(k)}}{mR_m^{(k)}}
\;\stackrel{\rm extrap.}{\longrightarrow}\;
-\frac18 \quad (k=1,2)\,,
\qquad
-\frac1{16} \quad (k=4)\,,
\label{eq:Uleading}
\ee
to between thirteen and eighteen significant digits, the accuracy degrading with
the shallower reconstruction at $k=4$. Under this extrapolation,
\eqref{eq:Uleading} would imply that the leading singular behavior of $\cU_k$
at $q_\star$ is that of $-\frac18 D\cR_k$ and
$-\frac1{16}D\cR_k$, respectively, with $D$ as in~\eqref{eq:Ddef}. The rational
prefactors are those multiplying $R_m$ in~\eqref{eq:Brel}. Only $17$, $15$,
and $15$ coefficients are available, so \eqref{eq:Uleading} is an
extrapolation from finitely many terms. It suggests the candidate leading
singular behavior of $\cU_k$ but does not constrain the singularity structure
of the unknown all-order function without an additional hypothesis class.

\begin{figure}[t]
\centering
\includegraphics[width=0.67\linewidth]{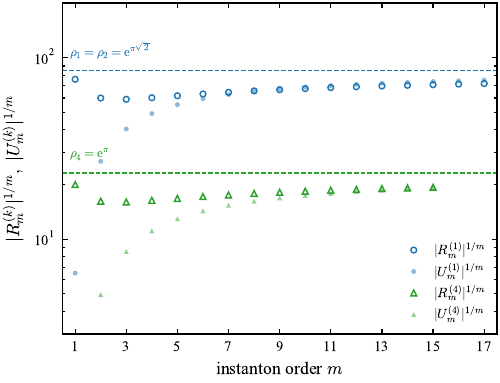}
\caption{Finite-order root-test diagnostics for the reconstructed index
generators. Open symbols show $|R_m^{(k)}|^{1/m}$, and the lighter filled
symbols show $|U_m^{(k)}|^{1/m}$, through $m=17$ at $k=1$ and through $m=15$
at $k=4$. The $k=2$ sequences coincide with those at $k=1$ through their
fifteen common orders and are omitted to avoid overplotting. Dashed lines mark
the candidate reciprocal radii
$\rho_1=\rho_2=\e^{\pi\sqrt2}$ and $\rho_4=\e^\pi$, equivalently
$|q_\star|^{-1}$ in Proposition~\ref{prop:radius}, under the all-order
continuation~\eqref{eq:Rmodular}. The behavior of $U_m^{(k)}$ is finite-order
evidence only.}
\label{fig:growthZ}
\end{figure}

\subsection{The unresolved generator}
\label{sec:U}

The finite-order behavior in~\eqref{eq:Uleading}, together with the modular
structure found for $\cR_k$, makes it natural to examine a bounded
low-complexity $q$-series library in the shifted nome $Q=-q$. The arguments
$Q,Q^2,Q^4$ accommodate the level-two and level-four structures identified
above. For $a=1,3$, define the Lambert series
\be
\mathscr L_a(x)
=
\sum_{n\geq1}\frac{n^a x^n}{1-x^n}
=
\sum_{m\geq1}\sigma_a(m)x^m\,,
\qquad |x|<1\,.
\ee
The raw library contains the nonconstant parts of the Eisenstein series
$E_2,E_4,E_6$ evaluated at $Q,Q^2,Q^4$, the series
$\mathscr L_1,\mathscr L_3$ and their images under $D$ at the same arguments,
the entries
$\cR_k,D\cR_k,D^2\cR_k,D^{-1}\cR_k$, and the two mixed terms
$\cR_k[E_2(Q)-1]$ and $(D\cR_k)[E_2(Q)-1]$.
Here $D^{-1}$ is understood as in~\eqref{eq:Dinverse}.

For $r\in\{1,2,4\}$,
\be
\mathscr L_1(Q^r)=-\frac{E_2(Q^r)-1}{24}\,,
\qquad
\mathscr L_3(Q^r)=\frac{E_4(Q^r)-1}{240}\,.
\ee
The six Lambert-series entries $\mathscr L_a(Q^r)$ are therefore exact scalar
multiples of Eisenstein entries. We retain the Eisenstein representatives and
remove these six duplicates, leaving the $21$ entries listed in
Table~\ref{tab:Ulibrary}.

\begin{table}[ht]
\centering
\small
\renewcommand{\arraystretch}{1.1}
\begin{tabular}{@{}llc@{}}
\toprule
Class & Retained entries & Count \\
\midrule
Eisenstein entries
&
$E_2(Q^r)-1,\ E_4(Q^r)-1,\ E_6(Q^r)-1$
&
$9$
\\
Lambert-series derivatives
&
$D\mathscr L_1(Q^r),\ D\mathscr L_3(Q^r)$
&
$6$
\\
First-generator entries
&
$\cR_k,\ D\cR_k,\ D^2\cR_k,\ D^{-1}\cR_k$
&
$4$
\\
Mixed terms
&
$\cR_k[E_2(Q)-1],\ (D\cR_k)[E_2(Q)-1]$
&
$2$
\\
\midrule
Total && $21$ \\
\bottomrule
\end{tabular}
\caption{Deduplicated library used in the bounded search for $\cU_k$, with
$r\in\{1,2,4\}$.}
\label{tab:Ulibrary}
\end{table}

Let $F_j$, $j=1,\ldots,21$, denote the retained entries after substituting
$Q=-q$. At each level we form the exact rational coefficient matrix
\be
\mathsf M_k
=
\bigl([q^m]F_j\bigr)_{
1\leq m\leq M_k^{(F)},\,
1\leq j\leq21}\,.
\ee
Its dimensions are $17\times21$, $15\times21$, and $15\times21$ at
$k=1,2,4$, and its exact ranks are respectively
\be
\operatorname{rank}\mathsf M_1=17\,,
\qquad
\operatorname{rank}\mathsf M_2=15\,,
\qquad
\operatorname{rank}\mathsf M_4=15\,.
\ee
Thus the full library has maximal row rank on each available coefficient
window. The low-complexity constraint in the search is the restriction to at
most four library entries.

More precisely, for every subset $S\subset\{1,\ldots,21\}$ with
$1\leq|S|\leq4$, we test whether rational coefficients $c_j$ exist such that
\be
[q^m]\cU_k(q)
=
\sum_{j\in S}c_j[q^m]F_j
\qquad
\text{for all}\quad
1\leq m\leq M_k^{(F)}\,.
\ee
There are
\be
\sum_{s=1}^{4}\binom{21}{s}=7{,}546
\ee
such subsets. For a subset with full column rank, a nonsingular square
subsystem determines the coefficients $c_j$ exactly over $\mathbb Q$, and
all remaining reconstructed coefficients provide exact checks. A
rank-deficient subset introduces no additional span beyond a smaller subset
already included in the search. No subset reproduces $\cU_k$ through all
reconstructed orders at $k=1,2,$ or $4$.

The entries involving $\cR_k$ are evaluated only through the same reconstructed
depths $M_k^{(F)}$. Equation~\eqref{eq:Rmodular} reproduces the reconstructed
coefficients $R_m^{(k)}$ exactly over these ranges, so no coefficients of
$\cR_k$ beyond the reconstructed windows enter the finite search. Extending
these entries to higher orders would require the conjectural continuation
in~\eqref{eq:Rmodular}.

The negative result is exact over the reconstructed ranges within the library
of Table~\ref{tab:Ulibrary} and the four-entry cap. It leaves open
representations involving weight-two modular forms of level eight in the
shifted nome, quasimodular forms of higher depth, mock-modular objects, theta
or product representations, and derivatives of $\cR_k$ beyond second order.

The rank-dependent perturbative structure of the selected-orbit index is
independently reconstructed from the data and agrees with the literature. Its
rank-independent term is exactly reconstructed within the declared basis,
and Theorem~\ref{thm:f0} proves the equivalence of its two analytic
representations. In the nonperturbative sector, $\cR_k$ has a unique
modular-product candidate within the families above over the reconstructed
orders, and Proposition~\ref{prop:radius} gives an exact radius theorem under
its all-order continuation. The all-order form of $\cU_k$ remains open.
Section~\ref{sec:validation} quantifies the residual associated with truncating
the index sectors at the reconstructed depth.

\section{The analysis framework}
\label{sec:method}

The computational analysis begins from the four high-precision scalar tables
described in Section~\ref{sec:setup}, which contain the on-shell values
$\im\Wt(N,k)$ and $\log Z_{\rm data}(N,k)$. The upstream Bethe solvers are not
part of the reproducibility workflow considered here. Starting from these
tables, we reconstruct the shifted rank, extract the perturbative parents,
reconstruct the constant maps and exponentially small sectors, and carry out
the arithmetic and modular searches developed in
Sections~\ref{sec:Weichler} and~\ref{sec:Rmodular}. This section describes the
logic of these inferences and the conditions under which each conclusion is
drawn.

Two features distinguish the analysis from unrestricted curve fitting. First,
the quantities being reconstructed are exact objects, such as rational numbers
or rational combinations of prescribed transcendental constants. Recognition
is therefore performed within explicitly stated bases, coefficient bounds, or
search classes. Failure of such a test excludes a candidate only within the
declared class and does not constitute a broader negative result. Second, the
hypothesis classes are supplied either by physical input or by arithmetic
structure identified in the data, and their provenance is stated explicitly.
In most stages, the class is fixed before the data used for parameter
determination are introduced.

The modular family for $\cR_k$ is the exception. Its architecture is informed
by the full reconstructed coefficient window, after which exact low-order
algebra determines the integer parameters within the chosen families. The
parameter-determining coefficients are kept disjoint from those used to test
the resulting parameter values. The choice of family itself, however, is
data-informed. We therefore do not describe the remaining coefficients as
out-of-sample tests of that choice. Table~\ref{tab:blind} summarizes the inference
or proof status of each stage, the class searched, the working precision, and
the independent checks that are applied.

\begingroup
\small
\renewcommand{\arraystretch}{1.16}
\begin{longtable}{@{}>{\raggedright\arraybackslash}p{0.21\linewidth}
                      >{\raggedright\arraybackslash}p{0.17\linewidth}
                      >{\raggedright\arraybackslash}p{0.56\linewidth}@{}}
\caption{Inference and proof status of each stage of the analysis, using the
vocabulary of Section~\ref{sec:exactness}. Working precisions are quoted in
decimal digits, while exact-rational stages are carried out over
$\mathbb Q$. Throughout the table, ``reconstructed'' refers to the stated
finite ansatz, basis, or search class.}
\label{tab:blind}\\
\toprule
Stage & Status & Declared search class and independent check\\
\midrule
\endfirsthead
\multicolumn{3}{l}{\footnotesize\tablename~\thetable\ continued}\\
\toprule
Stage & Status & Declared search class and independent check\\
\midrule
\endhead
\midrule
\multicolumn{3}{r}{\footnotesize Continued on the next page}\\
\endfoot
\bottomrule
\endlastfoot

Shift, Eq.~\eqref{eq:shift}
& reconstructed within a prescribed two-feature level law
& The scan uses $220$-digit arithmetic. At each $k=1,\ldots,50$, all sign
changes detected on a grid of spacing $1/4$ over $[-10,5]$ are bracketed, and
the calculation is repeated in four disjoint high-rank windows. The
coefficients in
$s_k=-\mu_{\rm sh}k-\nu_{\rm sh}/k$
are determined from $k=1,\ldots,12$ and recognized as rationals by continued
fractions with denominators at most $10^4$. The resulting law is then
validated on the withheld levels $k=13,\ldots,50$. Downstream stages use the
recognized rational law, while the individual roots are retained only as
window-stability diagnostics. \\
\addlinespace

Perturbative parents, Eqs.~\eqref{eq:Wparent} and~\eqref{eq:Zparent}
& reconstructed within a prescribed four-feature basis
& With the recognized shift law fixed, unconstrained four-parameter fits in
the basis
$\{\Nh^{3/2},\Nh^{1/2},\log\Nh,1\}$
are performed at $k=1,2,3,5,10,20,50$. At each level, the fit uses the four
largest available ranks and $220$-digit arithmetic, with selected diagnostics
recomputed at $240$ digits. The fitted coefficients agree with the closed
parent coefficients, with the agreement becoming less precise at larger $k$
because the unpeeled nonperturbative remainder is less strongly suppressed
over the available rank range. The resulting parent formulas are subsequently
tested by the residual and literature checks described below. \\
\addlinespace

Twisted superpotential constant $\CW(k)$,
Eqs.~\eqref{eq:clausen-odd}--\eqref{eq:clausen-even}
& levelwise integer-relation reconstruction
& Master arithmetic uses $240$ digits, while the levelwise PSLQ budget is
$d=\lfloor-\log_{10}|p_0-p_1|\rfloor-1$. Internal relations are removed from
the declared basis
$\{\zeta(3)/\pi,G,\Cltwo(2\pi j/k):1\leq j<k\}$
using the basis values alone. For $k\leq20$, the PSLQ search uses coefficient
bound $10^6$. The four robustness checks of the companion
Letter~\cite{Hosseini:2026dkj} require primitive height at most $10^5$,
reconstruction error below $10^{-(d-2)}$, stability at the reduced digit
budget $\lfloor3d/4\rfloor$, and stability under deletion of each unused basis
element. These checks are supplemented by the independent three-point
plateau-tail criterion \eqref{eq:plateau-test}. The Letter protocol certifies
$k=1,\ldots,8,10,12$, and the additional criterion leaves this set unchanged.
The closed parity formula is assembled using only $k\leq7$ and does not enter
the levelwise certification. \\
\addlinespace

Index constant $\fzero(k)$, Eq.~\eqref{eq:f0arith}
& rational reconstruction within a prescribed six-feature space
& At $120$-digit working precision, the largest-rank plateaux at
$k=1,\ldots,6$ are used to reconstruct the coefficients in the feature set
$\{A(k),A(k/2),\CW(k)/\pi,k^2\zeta(3)/\pi^2,\log k,\log 2\}$.
Continued-fraction recognition uses a relative tolerance of $10^{-20}$ and
allows at most six decimal digits in each numerator and denominator. The nine
levels $k=7,8,9,10,12,16,20,30,50$ are withheld from coefficient
determination and used only for validation. Equality of the arithmetic and
one-kernel representations is then proved in
Theorem~\ref{thm:f0}. \\
\addlinespace

Twisted superpotential sector structure, Eq.~\eqref{eq:Wfull}
& exact finite-range reconstruction within a prescribed basis linear in $t$
& The sector peel uses $900$-digit arithmetic in the two columns $q^m$ and
$tq^m$. For each instanton order $m$, the coefficients are reconstructed from
two interleaved fitting blocks and checked on a disjoint third block. A sector
is retained only if the two reconstructions agree, rational reconstruction
succeeds within the declared height bound, and the rank-differenced residual
on the held-out block is reduced by at least a factor of four. In every
retained sector, the reconstructed coefficients satisfy the $1/m$ relation in
Eq.~\eqref{eq:Wfull}. The overall normalization is selected from the family
$\pi^a k^b$, with $a,b\in\mathbb Z$ and $|a|,|b|\leq3$. The exponent of
$\pi$ is selected by coefficient height, while the exponent of $k$ is
determined by requiring the independently reconstructed $k=1$ and $k=2$
leading-coefficient sequences to coincide. \\
\addlinespace

Index sector structure, Eq.~\eqref{eq:Znp}
& reconstructed within bounded polynomial and normalization families
& The sector peel uses $900$-digit arithmetic. The polynomial degree is
determined from the $m=1$ sector by comparing $P=0,1,2,3$ in the basis
$\{t^p q:0\leq p\leq P\}$ on a disjoint held-out rank block. Degree two is
the lowest degree beyond which the held-out residual does not improve, while
the fitted cubic coefficient at $P=3$ is numerically negligible. Writing an
unconstrained sector as
$(\mathsf f_2 t^2+\mathsf f_1 t+\mathsf f_0)q^m$, the combination
$\mathsf f_1/m-\mathsf f_0$ admits the bounded rational reconstruction at
every retained instanton order, whereas $\mathsf f_0$ does not. This selects
the $t+1/m$ decomposition in Eq.~\eqref{eq:Znp}. The leading-coefficient
normalization is selected by the same bounded $\pi^a k^b$ search and
cross-level $k=1,2$ matching used for the twisted-superpotential sectors. \\
\addlinespace

Relation among sector coefficients, Eq.~\eqref{eq:Brel}
& exact finite-range reconstruction over $\mathbb Q$
& At $k=1,2$, the two coefficients multiplying $A_m^{(k)}$ and
$R_m^{(k)}$ are determined from the sectors $m=1,2$. At $k=4$, all three
coefficients, including the one multiplying
$\sigma_1^{\rm odd}(m)/m$, are determined from $m=1,2,3$. The resulting
relation is then verified exactly on every remaining reconstructed sector. \\
\addlinespace

Modular product ansatz for $\cR_k$, Eq.~\eqref{eq:Rmodular}
& exact parameter determination within data-informed product families
& The sign pattern, normalized coefficient growth, candidate zeros, and the
level-two and level-four arithmetic in the shifted nome $Q=-q$ motivate the
two modular-product families of Section~\ref{sec:Rmodular}. The first three
exact coefficients of the common $k=1,2$ sequence determine the integer
parameters $(p,c)$ uniquely within the stated family. At $k=4$, the
finite-order candidate zero independently selects $c=1$, while the first two
exact coefficients determine $(p,c)=(4,1)$. The remaining $14$, $12$, and
$13$ reconstructed coefficients at $k=1,2,4$, respectively, are withheld
from these determinations and reproduced exactly. The product families are
data-informed, and their continuation beyond the reconstructed orders remains
conjectural. \\
\addlinespace

Radius of convergence, Proposition~\ref{prop:radius}
& analytic theorem conditional on the conjectural product continuation
& The normalized reconstructed coefficients indicate candidate reciprocal
radii, while the raw root and ratio diagnostics test the predicted finite-$m$
deficits. Assuming the all-order continuation in Eq.~\eqref{eq:Rmodular},
modular transformations determine the relevant zero orbit, and comparison over
its modular images identifies the zero nearest to the origin. An independent
numerical check based on the argument principle and bisection locates the same
zero without using its analytic position as input. Under the same assumption,
Proposition~\ref{prop:radius} proves the logarithmic singularity, the radius of
convergence, and the leading coefficient growth. \\
\addlinespace

Bounded modular search for $\cU_k$
& exhaustive negative result within a prescribed finite library
& After six exact proportional duplicates are removed, the $27$-entry raw
library reduces to the $21$ entries of Table~\ref{tab:Ulibrary}. The resulting
coefficient matrices have exact ranks $17$, $15$, and $15$ at
$k=1,2,4$, respectively. At each level, all $7{,}546$ subsets containing at
most four library entries are examined. Whenever the selected entries are
linearly independent, an invertible coefficient subsystem determines the
rational coefficients exactly, and the resulting combination is tested
against all remaining reconstructed coefficients. No candidate reproduces
$\cU_k$ through all available orders. The negative result is exhaustive within
the declared library and the four-entry cap. Representations outside this
search class remain open. \\
\addlinespace
\end{longtable}
\endgroup

\subsection{Precision and conditioning}

Stored precision, working precision, and numerical accuracy are distinct. The
$200$- and $800$-digit labels refer to the decimal digits retained in the input
tables. They are not certified error bounds, as discussed in
Section~\ref{sec:setup}. The table loaders preserve these entries as decimal
strings and convert them to arbitrary-precision numbers only after the working
precision has been set. Numerical calculations are then performed at the
precision specified for each stage. Working above the stored width provides
guard digits against arithmetic loss, but it does not increase the information
content of the archived data.

We use the term ``empirical resolution limit'' for the scale beyond which the
stored data no longer support stable finite-data tests. In the sector peels
based on the $800$-digit tables, this limit is implemented conservatively by
requiring the smallest exponential factor in a complete fitting window to
satisfy
\begin{equation*}
q(N,k)^{m+W-1}>10^{-780}
\end{equation*}
at every usable rank. Here $m$ is the first sector in the window and $W$ is
its width. This operational threshold leaves a margin below the nominal $800$-digit width.
It should not be interpreted as a certified error estimate.

At high instanton order, the sector design matrices develop a large dynamic
range because the entries in the columns $t^p q^m$ span many decimal orders.
Before each arbitrary-precision solve, every column is divided by its largest
absolute entry on the corresponding fitting block. The two fitting blocks are
disjoint and interleaved over the same rank range. Their agreement determines
the precision demanded of the rational reconstruction. A separate held-out
block tests the effect of subtracting the reconstructed sector. The complete
acceptance rule is described in Section~\ref{sec:peel}.

Once a coefficient has been recognized, its rational part is stored and
manipulated exactly. The accompanying powers of $\pi$ and $k$ are kept
explicit, while $t$ and $q$ continue to be evaluated in arbitrary-precision
arithmetic. The resulting statements have the finite-data status ``exactly
reconstructed'' defined in Section~\ref{sec:exactness}.

Rigorous certification of these inferences would require certified enclosures
for the archived observables and a proof that those enclosures propagate
through every numerical solve. It would also require rigorous exclusion of
competing rational or integer relations within the declared search bounds.
Ball arithmetic provides a natural framework for the propagation of certified
enclosures~\cite{johansson2017arb}, but it does not by itself supply the
required uniqueness and exclusion arguments.

\subsection{Reconstruction of the shifted rank and perturbative parents}

At fixed $k$, we write the high-rank twisted-superpotential data provisionally
in the form
\be
\mathcal Y_N^{(k)}
\equiv
\im\Wt(N,k)
=
\alpha_k^{\rm par}(N+s_k)^{3/2}
+\delta_k^{\rm par}
+\varepsilon_N^{(k)}\,,
\ee
where $\varepsilon_N^{(k)}$ denotes the finite-rank remainder. For each
three-rank window, we temporarily neglect this remainder and eliminate
$\alpha_k^{\rm par}$ and $\delta_k^{\rm par}$. This gives a single scalar
equation for $s_k$. The calculation is performed at $220$-digit working
precision.

We scan the interval $[-10,5]$ on a grid of spacing $1/4$, bracket every sign
change detected between adjacent grid points, and refine each bracket by
bisection. No information from Eq.~\eqref{eq:shift} is used to initialize the
search. The procedure is repeated in four offset high-rank windows. The
endpoint of each window is ten ranks below that of the preceding window, and
the three sampling points within each window span forty ranks. The windows are
therefore distinct but not disjoint.

For every $1\leq k\leq50$, the scan detects exactly one sign-changing bracket
in each of the four windows. The root from the highest-rank window is retained
for the level-law reconstruction, while the other three roots provide
stability diagnostics. The largest spread among the four roots is
$1.15\times10^{-8}$ at $k=50$. This is also the level at which the retained
root has the largest discrepancy from the recognized level law, consistent
with the weaker suppression of the finite-rank remainder over the available
rank range.

The retained roots are fitted to the prescribed two-feature law
\begin{equation*}
s_k=-\mu_{\rm sh}k-\frac{\nu_{\rm sh}}{k}\,.
\end{equation*}
An overdetermined fit on $k=1,\ldots,12$, followed by continued-fraction
recognition with denominator at most $10^4$, yields
\begin{equation*}
\mu_{\rm sh}=\frac1{24}\,,
\qquad
\nu_{\rm sh}=-\frac23\,.
\end{equation*}
Relative to the recognized rational law, the maximum absolute residual is
$3.48\times10^{-21}$ on the determining levels and
$9.18\times10^{-9}$ on the thirty-eight withheld levels
$k=13,\ldots,50$. The latter occurs at $k=50$ and is comparable to the
cross-window spread at that level. Equation~\eqref{eq:shift} is therefore
exactly reconstructed within the prescribed two-feature law.

The feature space $\{k,1/k\}$ is fixed before the level data are fitted. It is
motivated by the analogous $k$- and $1/k$-dependence of the Airy shift in the
closely related $S^3$ localization problem
\cite{Fuji:2011km,Marino:2011eh}. The two coefficients of the present shift
law are nevertheless determined entirely from the Bethe-vacuum data.

Once the rational law
\be
s_k=-\frac{k}{24}+\frac{2}{3k}
\ee
has been recognized, it is used in every subsequent stage that depends on
$\Nh$, $t$, or $q$. The individual levelwise roots are retained only as
stability diagnostics.

With the shift fixed, unconstrained four-parameter fits are performed in the
prescribed basis
\begin{equation*}
\{\Nh^{3/2},\Nh^{1/2},\log\Nh,1\}
\end{equation*}
at $k=1,2,3,5,10,20,50$. Each fit uses the four largest available ranks for
the corresponding observable and $220$-digit arithmetic. All four
coefficients are determined independently at each level. Selected diagnostics
are recomputed separately at $240$ digits.

These fits test the perturbative structures in
Eqs.~\eqref{eq:Wparent} and~\eqref{eq:Zparent}. The fitted coefficients agree
with the closed parent coefficients, although the agreement becomes less
precise at larger $k$. Over the available rank range, the unpeeled
nonperturbative remainder is then less strongly suppressed. The fitted
rank-independent coefficients are not used as the constant maps in the
subsequent analysis. Those terms are reconstructed separately. The parent
formulas used downstream are supported jointly by these fits,
their agreement with the perturbative ABJM results of
Refs.~\cite{Bobev:2022jte,Bobev:2022eus}, and the residual tests of
Section~\ref{sec:validation}.

\subsection{Recognition of constants}

A constant plateau poses a different inference problem from a sector
coefficient. For $\CW(k)$, the principal modelling choice is the
transcendental basis in which the plateau is tested. The dilogarithmic
structure of the twisted superpotential naturally suggests Clausen values,
since
\be
\Cltwo(\theta)
=
\im\Li_2(\e^{\ii\theta})
=
\sum_{n\geq1}\frac{\sin(n\theta)}{n^{2}}\,,
\ee
for real $\theta$.
while the integer level singles out the rational angles $2\pi j/k$. The
observed level dependence of the numerical plateaux also motivates a
$\zeta(3)/\pi$ component. Their leading quadratic growth is consistent with
$-k^2\zeta(3)/(16\pi)$. We therefore take as the starting basis
\be
\{\zeta(3)/\pi,\,G,\,\Cltwo(2\pi j/k):1\leq j<k\} \, .
\ee

The raw Clausen list contains elementary redundancies. For even $k$, the entry
with $j=k/2$ vanishes because $\Cltwo(\pi)=0$. The entries labelled by $j$
and $k-j$ differ only by a sign,
\be
\Cltwo\!\left(\frac{2\pi(k-j)}{k}\right)
=
-\Cltwo\!\left(\frac{2\pi j}{k}\right).
\ee
Catalan's constant is retained as a separate basis element at every level.
When $4\mid k$, the rational-angle Clausen list already contains this value,
since $\Cltwo(\pi/2)=G$, and the corresponding Clausen entry is removed.
Further rational dependencies can remain among $G$ and the surviving Clausen
values. One source is the duplication identity
\be
\Cltwo(2\theta)
=
2\Cltwo(\theta)+2\Cltwo(\theta+\pi)\,.
\ee

These additional dependencies are removed before the plateau is introduced.
The internal-dependency search is performed at $200$-digit working precision
by applying PSLQ to $G$ and the surviving Clausen values, with tolerance
$10^{-150}$ and coefficient bound $10^5$. The distinguished entries
$\zeta(3)/\pi$ and $G$ are retained by construction. Whenever PSLQ finds a
bounded dependence, only a participating Clausen value is eligible for
removal. Among the eligible entries, the one with the smallest reduced
denominator is removed, with ties broken by the largest numerator. The search
is repeated until PSLQ finds no further relation within these bounds.

Using the plateau estimates and the conservative digit budget $d(k)$ defined
in Eq.~\eqref{eq:plateau-digits}, the levelwise relation search is adapted to
the information available at each level. Master arithmetic uses $240$ digits.
For a digit budget $d$, each PSLQ run uses fifteen guard digits, subject to a
minimum working precision of thirty digits. The tolerance is
$10^{-(d-4)}$, and the coefficient search bound is $10^6$
\cite{ferguson1999analysis,ferguson1998polynomial}. After primitive
normalization, a relation is certified under the four robustness checks of the
companion Letter~\cite{Hosseini:2026dkj} only if all of the following
conditions hold.
\begin{enumerate}[label=(\roman*),leftmargin=2.1em,itemsep=1pt,topsep=2pt]
\item Its integer height is at most $10^5$.
\item Solving the relation for $p_0(k)$ gives an error smaller than
$10^{-(d-2)}$.
\item Repeating the search with digit budget $\lfloor3d/4\rfloor$ returns the
same relation up to overall scale.
\item After each basis element with zero coefficient is deleted in turn, PSLQ
returns the same relation with that entry removed.
\end{enumerate}

We supplement these four robustness checks with an independent criterion for
the unresolved finite-rank tail. At fixed $k$, define
\be
\chi_{\rm tail}(k)
=
\frac{|p_0(k)-p_1(k)|}{|p_1(k)-p_2(k)|}\,,
\qquad
T(k)
=
|p_0(k)-p_1(k)|
\frac{\chi_{\rm tail}(k)}{1-\chi_{\rm tail}(k)}\,,
\qquad
0\leq\chi_{\rm tail}(k)<1\,.
\label{eq:plateau-test}
\ee
Under a locally geometric model for the final plateau errors, $T(k)$
estimates the remaining distance from $p_0(k)$ to the limiting plateau. We
require the reconstruction error at $p_0(k)$ to be no larger than $5T(k)$.
If $\chi_{\rm tail}(k)\notin[0,1)$, no finite geometric-tail estimate is
assigned and the criterion is not satisfied.

The four-check Letter protocol certifies
$k=1,\ldots,8,10,12$. Every one of these levels also satisfies the
three-point tail criterion, so the final certified set is unchanged. The tail
criterion by itself also passes at $k=9$, but the relation at that level fails
the three-quarter-precision check. At $k=19$, no relation is found within the
PSLQ search bound. All other attempted levels that are not certified fail at
least one of the required checks.

We attempt the levelwise integer-relation search only for $k\leq20$. This is a
conservative operational cutoff based on the competition between the available
plateau precision and the precision required for a bounded PSLQ search. For an
integer relation among $n$ numbers with coefficient height $H$, the required
precision is expected to scale roughly as
$n\log_{10}H$~\cite{ferguson1999analysis}. At $k=21$, the digit budget is
$d(21)=15$, while the pruned PSLQ vector has nine components. At the
certification height cap $H=10^5$, the corresponding heuristic scale is about
$45$ digits. Covering the full PSLQ search range up to $H=10^6$ would instead
require about $54$ digits. A null search at this level would therefore not
exclude relations throughout the declared coefficient range. Since the
plateau budget decreases further at larger $k$, negative searches become
progressively less informative. The cutoff at $k=20$ does not imply the
absence of lower-height relations at higher levels.

The index constant is reconstructed separately in the prescribed six-feature
space
\be
\left\{
A(k),\,
A(k/2),\,
\CW(k)/\pi,\,
k^2\zeta(3)/\pi^2,\,
\log k,\,
\log2
\right\}.
\ee
At $120$-digit working precision, the largest-rank plateaux at
$k=1,\ldots,6$ form a square linear system for the six coefficients.
Continued-fraction recognition is then applied with relative tolerance
$10^{-20}$. The absolute numerator and denominator of each recognized
coefficient are required to contain at most six decimal digits. The nine
levels
$
k=7,8,9,10,12,16,20,30,50
$
are excluded from coefficient determination and used only for validation.
Their residuals are consistent with the corresponding unpeeled instanton
remainders. This procedure yields Eq.~\eqref{eq:f0arith}.

Theorem~\ref{thm:f0} proves that the arithmetic and one-kernel
representations define the same function for every real $k>0$. Their equality
is therefore an analytic result. The identification of this common function
with the rank-independent term of the index remains exactly reconstructed
from finite data within the prescribed six-feature space.

\subsection{Adaptive sector peeling}
\label{sec:peel}

After subtracting a perturbative parent at fixed level $k$, write the residual
as
\be
\mathcal D(N)
=
\mathcal C_{\rm res}
+\sum_{m\geq1}\mathcal S_m(t)q^m\,,
\ee
where $\mathcal C_{\rm res}$ is the rank-independent term and
$\mathcal S_m(t)$ is the polynomial multiplying the $m$th instanton sector.

Every sector solve is performed on rank differences. For the solve beginning
at sector $m$, the code chooses a reference rank $N_{\rm ref}$ from the usable
ranks and fits
\begin{equation*}
\mathcal D(N)-\mathcal D(N_{\rm ref})\,.
\end{equation*}
The reference rank is fixed within that solve, but it can change when the
usable rank window changes at a later instanton order. The constant
$\mathcal C_{\rm res}$ cancels algebraically from every such system.

Neither $\CW(k)$ nor $\fzero(k)$ is supplied to the corresponding sector
peel. After the finite set of retained sectors has been subtracted, the
remaining residual at the largest deep-grid rank is recorded as a plateau
diagnostic. Its agreement with the separately reconstructed constant map is a
cross-check, not an input to the peel.

For a peel beginning at sector $m$, the design matrix contains a look-ahead
window of sectors
\begin{equation*}
m,m+1,\ldots,m+W-1\,.
\end{equation*}
The code first tries $W=8$ and decreases the width, if necessary, down to
$W=2$. A width is admissible only when the complete window remains above the
empirical resolution limit and enough ranks remain for two fitting blocks, a
held-out block, and the reference rank.

If the polynomial degree is $P$, each fitting block contains
$(P+1)W$ ranks. The two fitting blocks are disjoint and interleaved over the
same high-rank range. A third block of four ranks is kept separate for the
held-out test. Each column is rescaled by its largest absolute entry before
the square linear system is solved.

The two fitting blocks produce independent numerical estimates of the
coefficients in the first sector of the window. Their minimum relative
agreement determines the tolerance used by the stage-specific rational
recognizer. If the agreement corresponds to $d_{\rm agr}$ decimal digits, the
recognition tolerance is
\begin{equation*}
\tau
=
\max\!\left\{
\min\!\left(d_{\rm agr}/2,300\right),10
\right\},
\end{equation*}
and the allowed decimal height of the rational numerator and denominator is
\begin{equation*}
h_{\rm cap}
=
\max\!\left\{
\left\lfloor\tau/2.5\right\rfloor,6
\right\}.
\end{equation*}
The precise quantities subjected to rational recognition depend on the
observable. For the twisted superpotential they are the normalized
coefficient of $tq^m$ and the ratio of the $q^m$ and $tq^m$ coefficients.
For the index they are the normalized $t^2q^m$ and $tq^m$ coefficients and
the combination that tests the $t+1/m$ structure.

A sector is retained only if this recognition succeeds and its exact
subtraction reduces the largest rank-differenced residual on the held-out
block by at least a factor of four. The reconstructed rational coefficients
are then subtracted from all ranks before the next sector is considered. The
rational coefficients and the selected powers of $\pi$ and $k$ are stored
exactly. Only the evaluation of $t$ and $q$ remains numerical. The acceptance
rule is fixed throughout the peel, while the window width, recognition
tolerance, and height cap adapt deterministically to the information available
at each order.

The reconstruction depths are outputs of this procedure. For the twisted
superpotential,
\begin{equation*}
\left(
M_1^{\scriptscriptstyle(\mathcal W)},
M_2^{\scriptscriptstyle(\mathcal W)},
M_4^{\scriptscriptstyle(\mathcal W)}
\right)
=
(51,67,93)\,.
\end{equation*}
At $k=1$, the peel stops after sector $51$ because no admissible rank window of
width at least two remains. At $k=2$ and $k=4$, the next sector fails the
bounded rational-reconstruction test after sectors $67$ and $93$,
respectively.

For the index,
\begin{equation*}
\left(
M_1^{\scriptscriptstyle(F)},
M_2^{\scriptscriptstyle(F)},
M_4^{\scriptscriptstyle(F)}
\right)
=
(17,15,15)\,.
\end{equation*}
At all three levels, the next sector fails the stage-specific bounded
rational-reconstruction test. These depths describe the finite ranges over
which the coefficients are exactly reconstructed. They are not termination
statements for the underlying instanton series.

\subsection{Cross-level normalization and structure searches}

At a fixed level, a factor $k^b$ in the normalization can be absorbed into a
rational rescaling of the coefficient sequence. The power of $k$ therefore
cannot be determined from fixed-level data alone. We begin from the prescribed
finite family
\begin{equation*}
\pi^a k^b\,,
\qquad
a,b\in\mathbb Z\,,
\qquad
|a|,|b|\leq3\,.
\end{equation*}
For the twisted superpotential, this family is tested on the coefficient of
$tq^m$ in the first four sector solves. For the index, it is tested on the
coefficient of $t^2q^m$ in the first three sector solves. At each level,
bounded rational reconstruction selects the representation of smallest
coefficient height. This determines the power of $\pi$, while the power of
$k$ remains degenerate with a rational rescaling.

The remaining ambiguity is resolved across levels. After removing a candidate
factor $k^b$, we require the independently reconstructed sequences at $k=1$
and $k=2$ to coincide over the sectors used in the normalization search.
Exactly one value of $b$ satisfies this condition in each case. The resulting
prefactors are $k^2/(2\pi)$ for the twisted-superpotential sectors and
$k^2/\pi^2$ for the index sectors, as written in
Eqs.~\eqref{eq:Wnp} and~\eqref{eq:Znp}. No prior knowledge of the rational
coefficient sequences is used to select these normalizations. Their subsequent
arithmetic structure is analyzed only after the sector coefficients have been
reconstructed.

Once rational recognition has succeeded, the remaining finite-range structure
tests are carried out exactly over $\mathbb Q$. The relation
\eqref{eq:Brel} is determined from the minimum number of low-order sectors and
verified on every remaining reconstructed sector, as described in
Section~\ref{sec:Brel}.

The modular-product analysis of $\cR_k$ has a different provenance. Its family
architecture is informed by the reconstructed sign pattern, coefficient
growth, divisor arithmetic, and candidate zero locations discussed in
Section~\ref{sec:Rmodular}. Once the two families have been specified, their
integer parameters are determined by exact low-order equations, and all
remaining reconstructed coefficients test those parameter values. Because the
families themselves are data-informed, these checks validate the parameter
determination within the chosen families rather than the choice of family.

The search for $\cU_k$ is an exhaustive negative test within a prescribed
finite class. After exact proportional duplicates are removed from the
library, every subset containing at most four retained entries is examined.
For each subset with linearly independent columns, exact row pivots select an
invertible coefficient subsystem. The resulting rational combination is then
tested against the complete reconstructed coefficient window. The construction
and the scope of the resulting negative statement are given in
Section~\ref{sec:U}. No conclusion is drawn about representations outside the
declared library or the four-entry cap.

\subsection{Relation to symbolic regression}

The analysis is related to physics-informed symbolic regression in a limited
sense. Candidate expressions are sought within classes motivated by physical
or arithmetic structure and are confronted with high-precision data
\cite{Udrescu:2019mnk,Cranmer:2023pysr,Raayoni:2021ram}. Unlike generic
symbolic regression, however, the search does not range over a broad
expression space or rely on heuristic optimization. Each admissible class is
stated explicitly and, where relevant, bounded by its dimension, polynomial
degree, coefficient height, or library size. Numerical fits first isolate
candidate coefficients, after which continued fractions or PSLQ are used for
bounded recognition. Once the relevant coefficients have been rationally
reconstructed, the subsequent relation and library searches are carried out
exactly over $\mathbb Q$.

The interpretation of a successful check depends on how the candidate class
was chosen. When the class is specified independently of the data used for
parameter determination, separate data are reserved for validation whenever
possible. The modular family for $\cR_k$ is different because its architecture
is informed by the reconstructed coefficient window. The remaining
coefficients therefore test the parameter values within that family rather
than the choice of family itself. The status of each result is classified
according to Section~\ref{sec:exactness}, while analytic proofs are stated
separately from these data-driven inferences.

\section{Quantitative validation}
\label{sec:validation}

Each quantitative statement in this section is recomputed from the
high-precision values using the precision and comparison criteria stated below.
The resulting tests quantify the numerical consistency of the reconstructed formulas.

\subsection{Residual tests of the twisted-superpotential candidate}

We test the candidate obtained by combining the perturbative expression
\eqref{eq:Wparent} with the all-order divisor series
\eqref{eq:Wnp}. In this calculation, $\CW(k)$ is evaluated from the finite
Clausen formulas
\eqref{eq:clausen-odd}--\eqref{eq:clausen-even}, while the coefficients
$c_m^{(k)}$ are generated from the divisor formula \eqref{eq:cm}.

At each $(N,k)$, the instanton sum is truncated at
\be
M_{\rm sat}(N,k)
=
\left\lceil
\frac{830\log 10}{t(N,k)}
\right\rceil+8\,.
\label{eq:W-saturation-cutoff}
\ee
With $q=\e^{-t}$, this prescription gives
$q^{M_{\rm sat}}\leq10^{-830}q^8$, and the first omitted exponential factor
is smaller still. The eight-sector margin places the truncation scale below
the resolution of the $800$-digit archive.

Let
\be
\Delta_{\mathcal W}(N,k)
=
\im\Wt_{\rm data}(N,k)
-
\im\Wt_{\rm cand}^{(M_{\rm sat})}(N,k)
\ee
denote the resulting residual. Across all $297$ entries of the deep
twisted-superpotential table, the levelwise maxima are
\begin{center}
\small
\renewcommand{\arraystretch}{1.2}
\begin{tabular}{@{}ccc@{}}
\toprule
$k$
& rank $N_\star$
& $\displaystyle
   \max_{2\leq N\leq100}
   |\Delta_{\mathcal W}(N,k)|$ \\
\midrule
$1$ & $71$ & $4.79\times10^{-797}$ \\
$2$ & $66$ & $4.98\times10^{-797}$ \\
$4$ & $94$ & $4.95\times10^{-797}$ \\
\bottomrule
\end{tabular}
\end{center}
where $N_\star$ is the rank at which the maximum is attained. The overall
maximum is $4.98\times10^{-797}$ at $(k,N)=(2,66)$.

The archive stores each value to a fixed number of significant digits, so
the number of digits retained after the decimal point varies. One unit in the
final stored digit of the value at $(N,k)$ is therefore
\be
u_{\rm arch}(N,k)
=
10^{
\left\lfloor
\log_{10}
\left|
\im\Wt_{\rm data}(N,k)
\right|
\right\rfloor-799
}\,.
\ee
This unit grows with the magnitude of the observable, which behaves as
$N^{3/2}$ at large rank. The high-rank locations of the absolute maxima above
should therefore not be interpreted as a deterioration of the candidate.
Figure~\ref{fig:W_saturation} plots the dimensionless ratio
$|\Delta_{\mathcal W}|/u_{\rm arch}$. All $297$ ratios are smaller than
$0.5$, so the candidate agrees with every archived value within half a unit of
its final stored digit.

\begin{figure}[t]
\centering
\includegraphics[width=0.62\textwidth]{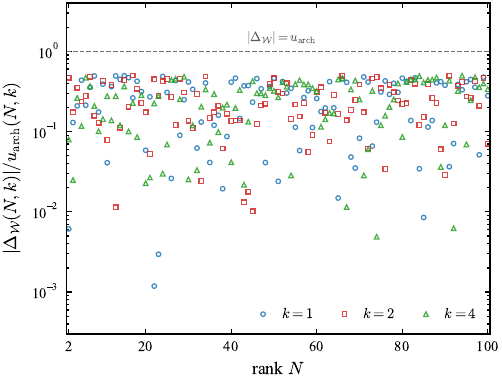}
\caption{Residual test of the all-order twisted-superpotential candidate
\eqref{eq:Wnp} at every $2\leq N\leq100$ and $k=1,2,4$. The divisor series
is truncated adaptively using $830$-digit working arithmetic. The vertical
axis is $|\Delta_{\mathcal W}(N,k)|/u_{\rm arch}(N,k)$, where
$u_{\rm arch}$ is one unit in the final stored digit of the corresponding
$800$-digit value. The dashed line marks
$|\Delta_{\mathcal W}|=u_{\rm arch}$. All $297$ residuals lie below half a
final-digit unit.}
\label{fig:W_saturation}
\end{figure}

The residual remains sensitive to the truncation prescription until the
omitted instanton scale has fallen below the archive resolution. To exhibit
this effect, the validation test also applies a uniform cutoff $M=60$ on a
thirteen-rank sample spanning $2\leq N\leq100$ at $k=2$ and $k=4$. The
largest residuals on this sample occur at $N=2$, where the nome is largest and
the omitted tail is least suppressed. They are
\begin{equation*}
3.82\times10^{-247}
\quad\text{at }k=2\,,
\qquad
3.43\times10^{-164}
\quad\text{at }k=4\,.
\end{equation*}
With the adaptive cutoff \eqref{eq:W-saturation-cutoff}, these small-rank
truncation effects are suppressed below the archival resolution. The largest
absolute residuals then occur at $N=66$ and $N=94$, respectively, where the
absolute size of the final stored digit is larger. Residuals at this precision
must therefore be quoted together with the truncation prescription used to
obtain them.

We also fix $(k,N)$ and examine the residual as a function of the truncation
order $M$. The three cases shown in Figure~\ref{fig:W_convergence} are
summarized below.
\begin{center}
\small
\renewcommand{\arraystretch}{1.2}
\begin{tabular}{@{}ccccc@{}}
\toprule
$k$ & $N$ & $M_{\max}$ &
$M_k^{\scriptscriptstyle(\mathcal W)}$ &
comparison with the reconstructed range \\
\midrule
$4$ & $20$ & $108$ & $93$ &
continuation tested for $94\leq m\leq108$ \\
$2$ & $30$ & $67$ & $67$ &
reconstructed endpoint reached \\
$1$ & $40$ & $45$ & $51$ &
within the reconstructed range \\
\bottomrule
\end{tabular}
\end{center}
In each case, the logarithm of the residual follows the geometric slope
$\log_{10}q$ until it reaches the resolution of the archived datum. The
$k=4$ curve uses the continued divisor coefficients through $m=108$, fifteen
orders beyond the independently reconstructed depth. No resolvable departure
from the geometric decay is observed over this interval. The $k=2$ curve
reaches the final reconstructed order, while the $k=1$ curve remains within
the reconstructed range. Thus all three curves test the expected truncation
behavior and its saturation at the archive resolution, while the $k=4$ curve
alone directly probes the all-order continuation beyond the reconstructed
coefficients.

\begin{figure}[!ht]
\centering
\includegraphics[width=0.62\textwidth]{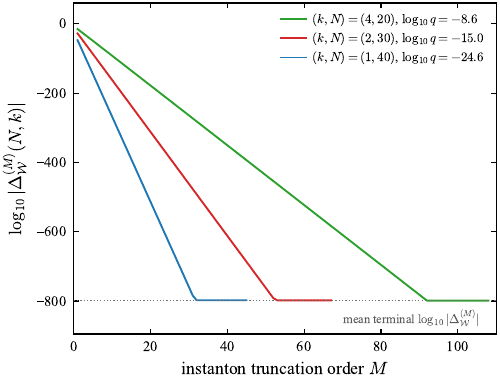}
\caption{Dependence of the twisted-superpotential residual on the integer
instanton truncation order $M$ at
$(k,N)=(4,20),(2,30),(1,40)$. Before saturation, each curve follows the
geometric slope $\log_{10}q$ given in the legend. The horizontal dotted line
marks the mean of the three terminal values of
$\log_{10}|\Delta_{\mathcal W}^{(M)}(N,k)|$. The curves terminate at
$M=108,67,45$, respectively, while the corresponding reconstructed depths are
$93,67,51$. The $k=4$ curve therefore probes the continued divisor formula
for $94\leq m\leq108$. The $k=2$ curve reaches its reconstructed endpoint,
whereas the $k=1$ curve remains within the reconstructed range.}
\label{fig:W_convergence}
\end{figure}

\subsection{Cross-checks of the constant maps}

The finite Clausen forms~\eqref{eq:clausen-odd}--\eqref{eq:clausen-even},
the Euler representation~\eqref{eq:CEuler}, and the integral
representation~\eqref{eq:CW} of $\CW(k)$ agree for every $1\le k\le50$ at
$70$-digit working precision. The largest pairwise discrepancy is
$6.4\times10^{-69}$ at $k=45$. No Bethe data enter this comparison. It tests
the numerical evaluation and mutual consistency of the three representations
proved equivalent in Theorem~\ref{thm:C}. Under a precision scan, the largest
discrepancy over $40\le k\le50$ tracks the working precision, taking the values
$8.6\times10^{-49}$, $6.4\times10^{-69}$, and $4.1\times10^{-99}$ at $50$,
$70$, and $100$ digits, respectively.

At $120$-digit working precision, the arithmetic
representation~\eqref{eq:f0arith} and the one-kernel
representation~\eqref{eq:f0kernel} of $\fzero(k)$ are compared at
$k=1,\ldots,10,12,16,20,30,50$.
The largest discrepancy is $2.5\times10^{-119}$ and occurs at $k=50$.
This numerical comparison complements Theorem~\ref{thm:f0}, which proves the
equality of the two representations for every real $k>0$.

As an independent check of Theorem~\ref{thm:Euler}, direct numerical summation
of the alternating Euler series~\eqref{eq:E-def} agrees with the closed
forms~\eqref{eq:E1-theorem} and~\eqref{eq:E2-theorem} for
$\cE_{1,p}$ through $p=12$ and for $\cE_{2,p}$ at odd $p\leq11$,
respectively. At $60$-digit working precision, the largest discrepancy is
$7.8\times10^{-61}$.

\subsection{Finite reconstruction depths}

The adaptive peeling procedure returns the depths shown in
Table~\ref{tab:depth}.

\begin{table}[ht]
\centering
\small
\begin{tabular}{@{}lccc@{}}
\toprule
 & $k=1$ & $k=2$ & $k=4$ \\
\midrule
Twisted superpotential,
$M_k^{\scriptscriptstyle(\mathcal W)}$
& $51$ & $67$ & $93$ \\
Selected-orbit index,
$M_k^{\scriptscriptstyle(F)}$
& $17$ & $15$ & $15$ \\
\bottomrule
\end{tabular}
\caption{Finite reconstruction depths returned by the adaptive peeling
procedure of Section~\ref{sec:peel}. At $k=1$, the
twisted-superpotential peel stops after sector $51$ because no admissible rank
window remains. At $k=2$ and $k=4$, the next sector fails bounded rational
recognition after sectors $67$ and $93$, respectively. The three index peels
stop by the same recognition criterion after sectors $17$, $15$, and $15$.
These depths do not imply termination of either instanton expansion.}
\label{tab:depth}
\end{table}

At every reconstructed twisted-superpotential order, the coefficient of
$q^m$ divided by the coefficient of $tq^m$ is $1/m$. The sector therefore
has the $t+1/m$ structure of Eq.~\eqref{eq:Wfull}, and the recognized
coefficient $c_m^{(k)}$ agrees with the divisor expression
\eqref{eq:cm}. Every reconstructed index sector has the corresponding
$t+1/m$ structure in Eq.~\eqref{eq:Znp}. The relation
\eqref{eq:Brel} also holds exactly over $\mathbb Q$ throughout the
reconstructed ranges.

The index peel fits three features per sector, compared with two for the
twisted superpotential. In all three index runs, the first unretained sector
fails the bounded rational-recognition step. The depths in
Table~\ref{tab:depth} record where the fixed reconstruction criteria cease to
support another sector and should not be interpreted as intrinsic termination
orders.

\subsection{Impact of the unresolved generator}

After subtracting all reconstructed index sectors, we difference the remaining
residual against the reference rank $N_{\rm ref}=100$ so that the
rank-independent term cancels. Denote the absolute rank-differenced residual
at level $k$ by $\mathfrak r_M^{(k)}(N)$, where
$M=M_k^{\scriptscriptstyle(F)}$ is the reconstructed depth.

At $N=30,40,60,80$, the rank dependence of
$\mathfrak r_M^{(k)}(N)$ follows that of $q^{M+1}$ for each of
$k=1,2,4$, as shown in Figure~\ref{fig:TTI_scaling}. At fixed $k$, the ratio
\begin{equation*}
\frac{|\mathfrak r_M^{(k)}(N)|}{q^{M+1}}
\end{equation*}
varies by less than a factor of $2.7$, while the residual itself decreases by
many decimal orders. The first omitted exponential scale therefore controls
the rank dependence of the unresolved contribution, with a level-dependent
coefficient.

For the twisted superpotential, the divisor formula~\eqref{eq:cm} supplies
coefficients beyond the independently reconstructed depth and permits the
continued sum in Eq.~\eqref{eq:Wnp}. No corresponding all-order expression
for $\cU_k(q)$ has been identified within the bounded search of
Section~\ref{sec:U}. The index reconstruction therefore cannot presently be
continued beyond the finite range in Table~\ref{tab:depth}.

\begin{figure}[t]
\centering
\includegraphics[width=0.62\textwidth]{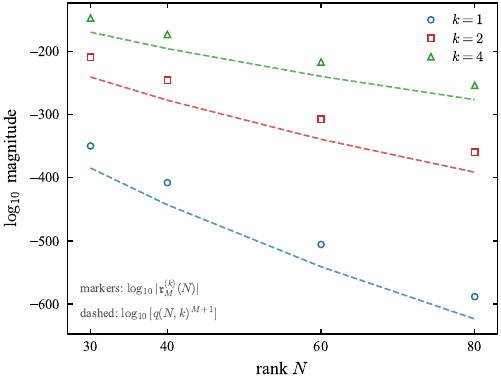}
\caption{First-omitted-sector scaling of the selected-orbit index. For each
level, $M$ denotes the corresponding reconstructed-sector cutoff
$\smash{M_k^{\scriptscriptstyle(F)}}$. After the reconstructed sectors are subtracted
and the rank-independent term is removed by rank differencing, the markers
show $\smash{\log_{10}|\mathfrak{r}_M^{(k)}(N)|}$ at $N=30,40,60,80$, while the dashed
curves show $\log_{10}[q(N,k)^{M+1}]$. At fixed $k$, the ratio
$\smash{|\mathfrak{r}_M^{(k)}(N)|/q(N,k)^{M+1}}$ varies by less than a factor of $2.7$.
The near-parallel decay is consistent with the exponential scale of the first
unreconstructed sector.}
\label{fig:TTI_scaling}
\end{figure}

\subsection{Type-IIA genus expansions}

At $90$-digit working precision, the genus
expansions~\eqref{eq:Wgen} and~\eqref{eq:Fgen} are compared with their
closed perturbative parents on the grid
\begin{equation*}
\lambda=5,10,20\,,
\qquad
k=10,20,30,60\,.
\end{equation*}
In all twelve cases, the truncation error reaches a minimum at finite order
and subsequently increases, as expected for an asymptotic expansion. At
$\lambda=10$, the minimum twisted-superpotential error decreases from
$6.0\times10^{-8}$ at $k=10$ to $1.9\times10^{-42}$ at $k=60$. For the
index, the corresponding errors decrease from
$8.8\times10^{-7}$ to $1.7\times10^{-40}$.

For each fixed value of $k$, the minimizing truncation order is the same at
$\lambda=5,10,20$ on the tested grid. This observed independence is
consistent with the late-order growth being governed by the
$\lambda$-independent Bernoulli-number contributions inherited from the
constant maps. The $\lambda$-dependent terms contribute at low and
intermediate order but do not share the dominant factorial growth.

Figures~\ref{fig:genusW} and~\ref{fig:genusZ} display the representative
slice $\lambda=10$ for $k=10,20,30$. Together with the complete grid
described above, they test the derived large-$k$ coefficient formulas against
the closed perturbative parents.

\section{Discussion and outlook}
\label{sec:discussion}

\subsection{Arithmetic organization on the selected Bethe orbit}

The comparison between the twisted superpotential and the topologically
twisted index is especially informative because both quantities are evaluated
on the same selected Bethe solution and its $k$-fold orbit. Their different
finite-rank structures therefore reflect the different information carried by
the two observables rather than a change of saddle. For the twisted
superpotential, the exponentially small sectors reduce at $k=1,2,4$ to a
single rational sequence $c_m^{(k)}$ governed by the divisor
formula~\eqref{eq:cm}. The finite-range relation~\eqref{eq:Brel} instead
reduces the three reconstructed index sequences to two independent generators,
$\cR_k$ and $\cU_k$. The first admits a modular-product ansatz within the
data-informed families of Section~\ref{sec:Rmodular}. No representation for
the second appears within the finite library and four-entry cap of
Section~\ref{sec:U}.

This difference points toward the additional local data entering the
Bethe-vacua formula for the index. Besides the critical value of $\Wt$, the
index depends on the Jacobian $\det\mathbb B$, the effective dilaton, and the
remaining one-loop factors. The extra generator may encode arithmetic
information carried by these fluctuation factors and absent from the on-shell
twisted superpotential. A direct decomposition of the reconstructed
coefficients into contributions from the individual factors in
Eq.~\eqref{eq:BAformula} would make this interpretation testable.

The rank-independent term of the index provides a complementary connection
between the two observables and the sphere partition function. The
rank-dependent parent~\eqref{eq:Zparent} agrees with the perturbatively exact
large-$N$ index of Refs.~\cite{Bobev:2022jte,Bobev:2022eus}. At the universal
twist, Ref.~\cite{Bobev:2022eus} determined the large-$k$ expansion of the
remaining rank-independent function through five inverse-power coefficients
and a numerically fitted constant. The one-kernel
representation~\eqref{eq:f0kernel} reproduces those five coefficients and
identifies the fitted constant as
\be
f_0
=
-8\zeta'(-1)
-\frac{23}{6}\log2
-\frac{2}{3}\log\pi\,.
\ee
Its numerical value differs from the fitted value quoted in
Ref.~\cite{Bobev:2022eus} by $5.5\times10^{-17}$.

More structurally, the arithmetic representation~\eqref{eq:f0arith} places
the sphere constants $A(k)$ and $A(k/2)$ and the twisted-superpotential
constant $\CW(k)$ in the same rank-independent index contribution. This
combination was reconstructed from the index plateaux.
It suggests that the effective dilaton, Jacobian, and
one-loop factors assemble the constant terms of the sphere and twisted
observables in a constrained way whose analytic origin remains to be
understood.

A second common feature is the shifted nome $Q=-q$. It exposes the
weight-four Eisenstein structure of the twisted-superpotential coefficients
and also organizes the modular-product families for $\cR_k$. In this
coordinate, the $k=1$ and $k=2$ twisted-superpotential sequences coincide and
are associated with level two, while the $k=4$ sequence has a distinct
level-four form. The equality at $k=1,2$ may reflect the $\cN=8$
enhancement, but the separate closure at $k=4$ shows that enhanced
supersymmetry is not by itself sufficient to explain the observed modular
organization. These special-level structures invite comparison with the
nonperturbative and spectral descriptions of the ABJM sphere partition
function.

\subsection{Relation to the sphere transseries and spectral modularity}

The modified grand potential provides the most direct setting for comparing
the exponential structures. In the Fermi-gas description of the sphere
partition function,
\be
J_k(\mu)
=
J_k^{\rm pert}(\mu)
+
J_k^{\rm np}(\mu)\,,
\qquad
J_k^{\rm pert}(\mu)
=
\frac{2\mu^3}{3\pi^2k}
+
\left(
\frac{k}{24}+\frac{1}{3k}
\right)\mu
+
A(k)\,.
\ee
The generic nonperturbative sector has the schematic form
\be
J_k^{\rm np}(\mu)
=
\sum_{\substack{n_{\rm ws},n_{\rm M2}\geq0\\
(n_{\rm ws},n_{\rm M2})\neq(0,0)}}
f_{n_{\rm ws},n_{\rm M2}}^{(k)}(\mu)
\exp\!\left[
-\left(
\frac{4n_{\rm ws}}{k}+2n_{\rm M2}
\right)\mu
\right].
\label{eq:S3-action-lattice}
\ee
The terms with $n_{\rm M2}=0$ are worldsheet instantons, those with
$n_{\rm ws}=0$ are membrane instantons, and terms with both integers nonzero
are mixed bound states. In the HMO organization, the bound-state contributions
are absorbed into the worldsheet sector through the effective chemical
potential $\mu_{\rm eff}$. Contributions that are separately singular at
integer $k$ then combine into finite coefficients
\cite{Hatsuda:2012dt,Hatsuda:2013gj,Hatsuda:2013oxa}.

At the three levels relevant to the present reconstruction, the resulting
nonperturbative grand potential can be organized in the single exponential
variable
\be
x_k=\e^{-4\mu_{\rm eff}/k}\,,
\qquad
k\in\{1,2,4\}.
\ee
Each coefficient in this expansion encodes the worldsheet, membrane, and
bound-state data contributing at the corresponding exponential order. The
sphere grand potential and the twisted superpotential may therefore both be
viewed as one-variable exponential expansions at these levels. Their
structural difference lies in the coefficient systems carried by the common
one-dimensional organization.

The sphere series is naturally grand canonical, whereas the twisted
superpotential studied here is a fixed-rank quantity. Their exponential
actions can nevertheless be compared through the perturbative saddle
underlying the Airy representation. In terms of the shifted 't~Hooft coupling
$\lamh=N/k-1/24$, the worldsheet action at the sphere saddle and the action
reconstructed from the twisted-superpotential data are
\be
t_{\sthree}
=
2\pi
\sqrt{
2\left(
\lamh-\frac{1}{3k^2}
\right)
}\,,
\qquad
t_{\mathcal W}
=
2\pi
\sqrt{
2\left(
\lamh+\frac{2}{3k^2}
\right)
}\,.
\label{eq:action-comparison}
\ee
They obey
\be
t_{\mathcal W}^{\,2}-t_{\sthree}^{\,2}
=
\frac{8\pi^2}{k^2}\,.
\ee
Both actions approach $2\pi\sqrt{2\lamh}$ at large $\lamh$, so the
twisted-superpotential instanton variable has the same leading action as the
sphere worldsheet sector. Their finite-rank shifts are different. The
exponentially small difference between $\mu_{\rm eff}$ and $\mu$ does not
modify this leading saddle comparison.

The agreement of the leading actions does not by itself identify the
coefficient structures. A coefficient at fixed order in the sphere grand
potential is obtained after combining worldsheet, membrane, and bound-state
contributions on the common exponential lattice. The worldsheet part is
governed by the ordinary topological string on local
$\mathbb P^1\times\mathbb P^1$, while the membrane part is governed by the
refined topological string in the Nekrasov--Shatashvili limit
\cite{Hatsuda:2013oxa}. These combined coefficients are naturally expressed
in terms of $\mu_{\rm eff}$. Before transforming to fixed rank, they are
re-expanded in the bare chemical potential $\mu$. Polynomial factors in
$\mu$ then become derivatives with respect to $N$ acting on Airy functions
whose arguments contain both the usual perturbative shift and the additional
shift generated by the corresponding instanton exponential
\cite[Eq.~(6.2)]{Codesido:2014oua}. The canonical expansion therefore does
not inherit the grand-potential coefficients term by term.

The twisted-superpotential reconstruction has a different coefficient
algebra. Its proposed continuation is
\be
\mathcal W_{\rm np}(N,k)
=
\frac{k^2}{2\pi}
\sum_{m\geq1}
c_m^{(k)}
\left(
t_{\mathcal W}+\frac1m
\right)
\e^{-m t_{\mathcal W}}\,,
\ee
where the $c_m^{(k)}$ are rational divisor sums. At the reconstructed
special levels, these coefficients are organized by weight-four Eisenstein
forms and their Eichler primitives. Both the sphere grand potential and the
twisted superpotential therefore admit one-dimensional exponential
organizations, but the coefficients encode different data. On the sphere
side they arise from the pole-cancelled combination of ordinary and refined
topological-string sectors. On the twisted-superpotential side they form a
single rational divisor sequence with the universal factor
$t_{\mathcal W}+1/m$.

This sharper comparison leads to a concrete open question. It remains to
determine whether the divisor sequence and its universal $t_{\mathcal W}+1/m$
dependence can be obtained from a distinguished projection or reorganization
of the ordinary and Nekrasov--Shatashvili topological-string data appropriate
to the selected Bethe observable. An alternative possibility is that the two
problems share the leading worldsheet action and the integer-level
exponential lattice while their coefficient structures have independent
origins.

The special levels provide a second comparison through spectral modularity.
For ABJM theory at $k=1,2$, the higher-genus worldsheet contributions simplify
and the spectral construction can be expressed in terms of ordinary Jacobi
theta functions associated with the mirror curve
\cite{Kallen:2013qla,Grassi:2014zfa,Codesido:2014oua}. At $k=4$, the
all-genus worldsheet contribution remains, but the generalized theta function
again reduces to a Jacobi theta function after an elliptic resummation. That
resummation was conjectured and checked order by order in
Ref.~\cite{Grassi:2014uua}. Within the same construction, the $k=4$ quantum
volume is related to that of the maximally supersymmetric
$(k,M)=(2,1)$ ABJ theory. The $k=4$ example is important because it shows that
special-level modular simplification is not confined to the
$\cN=8$ theories.

The modular structure found here acts on a different object. In the shifted
nome $Q=-q$, the proposed continuation of the common $k=1,2$
twisted-superpotential sequence is governed by
\be
E_4(Q)-4E_4(Q^2)
\in M_4\!\left(\Gamma_0(2)\right),
\ee
while the $k=4$ sequence is governed by
\be
E_4(Q)-16E_4(Q^4)
\in M_4\!\left(\Gamma_0(4)\right).
\ee
These forms organize the instanton coefficients of an on-shell Bethe
observable. The recurrence of the levels
$k=1,2,4$ is therefore a genuine point of contact with the sphere spectral
problem, while the modular quantities and the observables they control remain
different.

The index supplies a further variation. Within the data-informed families of
Section~\ref{sec:Rmodular}, $\cR_k$ admits a modular-product ansatz. Assuming
its all-order continuation, the modular orbit of its nearest zero determines
the logarithmic singularity, the radius of convergence, and the leading
coefficient growth in Proposition~\ref{prop:radius}. In the Fermi-gas spectral
problem, zeros of the spectral determinant determine energy levels. In the
present index problem, the relevant zero instead governs the analytic
structure of a coefficient-generating function. The analogy concerns the
organizing role of modular zeros, not an identification of the two spectral
objects.

\subsection{Outlook}

The first priority is an analytic derivation of the divisor
law~\eqref{eq:cm} and of the half-period shift $Q=-q$ from the Bethe equations.
Such a derivation would determine whether the Eisenstein coefficients arise
from a reorganization of the ordinary and refined topological-string data
entering the sphere grand potential. It would also clarify why the $k=1$ and
$k=2$ sequences coincide and why a separate level-four form appears at
$k=4$. Extending the deep reconstruction to $k=3,5$ and other levels would
help distinguish structures that are generic at integer $k$ from the stronger
modular closures observed at the three levels studied here.

For the index, the immediate problem is to understand the second generator
$\cU_k(q)$. Deeper finite-rank data would enlarge the reconstructed
coefficient window, but a derivation from the Jacobian, effective dilaton, and
one-loop factors in Eq.~\eqref{eq:BAformula} would be more decisive. It could
explain why $\cR_k$ admits the modular-product ansatz found here and reveal
the arithmetic structure governing $\cU_k$. An analytic derivation of
Eq.~\eqref{eq:f0arith} would likewise clarify why $A(k)$, $A(k/2)$, and
$\CW(k)$ enter the same rank-independent index term.

A further direction is to move beyond the selected orbit. Enumerating Bethe
vacua at small rank would determine whether other admissible solutions
contribute at comparable size and whether their finite-rank coefficients
display related arithmetic structures. The same analysis can be extended to
other Seifert manifolds~\cite{Closset:2018ghr} and to other
three-dimensional $\cN=2$ quiver gauge theories.

The perturbative parents, rank-independent terms, and type-IIA coefficients
provide field-theory benchmarks for higher-derivative
$\mathrm{AdS}_4$ holography
\cite{Bobev:2020egg,Bobev:2021oku}. The exponentially small sectors probe a
different part of the bulk expansion. Their actions are compatible with
brane-instanton scales, while their arithmetic coefficients remain without a
bulk derivation. Establishing such an interpretation requires an analytic
dictionary between the selected finite-$N$ Bethe solution and Euclidean
fundamental strings, wrapped branes, or more general nonperturbative saddles
in the dual geometry.

\section*{Acknowledgments}
It is a pleasure to thank Alberto Zaffaroni for valuable discussions and for his careful reading of and comments on the manuscript. I am also grateful to Parastoo Salah for many helpful discussions and insightful comments on the manuscript.
This research was supported by UK Research and Innovation (UKRI) under the UK government's Horizon Europe funding guarantee (Grant No.~EP/Y027604/1) and by the Science and Technology Facilities Council (STFC) under Grant No.~ST/X000494/1.

\textbf{Statement on use of AI.}
The ideas, arguments, and writing in this paper came from the author. An OpenAI model provided useful assistance with code development, computational checks, and copy editing.

\section*{Data and code availability}
The data and code supporting the findings of this manuscript are not publicly available and may be obtained from the author upon reasonable request.
\appendix

\section{Proof of the constant-map theorem}
\label{app:proof}

This appendix proves Theorem~\ref{thm:C}, supplying the argument deferred
in~\cite{Hosseini:2026dkj}. The proof proceeds in two steps. We first reduce the
integral representation~\eqref{eq:CW} exactly to alternating Euler sums at
rational arguments. We then evaluate those sums by root-of-unity methods,
obtaining the finite Clausen expressions of
Eqs.~\eqref{eq:clausen-odd}--\eqref{eq:clausen-even}. The required Euler-sum
identities are stated separately as Theorem~\ref{thm:Euler}.

\paragraph{Roadmap.}
Section~\ref{sec:prelim} establishes the convergence and interchange results
used throughout. Section~\ref{sec:integral-reduction} reduces the defining integral to
the sums $\cE_{1,k}$ and $\cE_{2,k}$.
Sections~\ref{sec:roots}--\ref{sec:E2} construct the relevant
root-of-unity kernels and prove Theorem~\ref{thm:Euler}, and
Section~\ref{sec:assembly} combines the two parity cases to prove
Theorem~\ref{thm:C}. Branch conventions are fixed below and kept unchanged
throughout the proof.

For positive integers $a,p$, define
\be
\cE_{a,p}
=
\sum_{m=1}^{\infty}
\frac{(-1)^{m+1}}{m^{2}}\,
\Hgen_{am/p}\,.
\label{eq:E-def}
\ee
For real $k>0$, let
\be
\mathfrak I_k
=
\int_0^{\infty}
\log(2\cosh x)\,
\log\bigl(1-\e^{-kx}\bigr)\,
\rd x\,,
\label{eq:Ik-def}
\ee
and
\be
\Cint(k)
=
-\frac{k^{2}\zeta(3)}{16\pi}
+\frac{2k}{\pi}\mathfrak I_k\,.
\label{eq:Cint-def}
\ee
The logarithms in~\eqref{eq:Ik-def} are real on the integration contour, and
the integral is absolutely convergent; this is established in
Proposition~\ref{prop:Euler-reduction}. The comparison with the finite Clausen
forms will require $k\in\mathbb Z_{>0}$.

For $z\geq0$, we use the generalized harmonic number in the equivalent forms
\be
\Hgen_z
=
\int_0^1
\frac{1-x^{z}}{1-x}\,
\rd x
\label{eq:H-def}
\ee
and
\be
\Hgen_z
=
\sum_{n=1}^{\infty}
\left(
\frac1n-\frac1{n+z}
\right).
\label{eq:H-series}
\ee
Their equivalence is proved in Lemma~\ref{lem:H}. In conventional notation,
$\Hgen_z=\psi(1+z)+\gamma$, where $\psi$ is the digamma function and $\gamma$
is the Euler--Mascheroni constant. No analytic continuation in $z$ is needed
below.

Euler sums of this general type can be treated using contour integrals,
polygamma or Hurwitz-zeta functions, and the Lerch transcendent
\cite{FlajoletSalvy1998,sofo2015quadratic}. Here we instead keep the
root-of-unity structure explicit, since the required answer is a finite
combination of Clausen values. We write
\be
\Cltwo(\theta)
=
\im\Li_2(\e^{\ii\theta})
=
\sum_{n\geq1}\frac{\sin(n\theta)}{n^{2}}\,,
\qquad
\cC_3(\theta)
=
\sum_{n\geq1}\frac{\cos(n\theta)}{n^{3}}\,,
\qquad
G=\Cltwo\!\left(\frac{\pi}{2}\right),
\label{eq:Cl-C3-def}
\ee
for real $\theta$. Both Fourier series converge absolutely and uniformly on
$\mathbb R$.

Finally, we use the principal logarithm
\be
\mathrm{Log}\,z
=
\log|z|+\ii\arg z\,,
\qquad
-\pi<\arg z<\pi\,,
\ee
with branch cut $(-\infty,0]$.
It enters only in Section~\ref{sec:roots}, where every
argument lies in the open right half-plane or is reached by an integrable
endpoint limit.

\begin{theorem}[Rational-argument Euler sums]
\label{thm:Euler}
For every positive integer $p$,
\be
\cE_{1,p}
=
\left(
\frac{3p}{8}
-\frac{5+7(-1)^{p}}{8p^{2}}
\right)\zeta(3)
+
\frac{\pi}{2p}
\sum_{j=1}^{p-1}
\frac{3-(-1)^{p+j}}{2}\,j\,
\Cltwo\!\left(\frac{2\pi j}{p}\right).
\label{eq:E1-theorem}
\ee
For every odd positive integer $p$,
\bea
\cE_{2,p}
={}&
\left(
\frac{3p}{16}
-\frac{13}{8p^{2}}
\right)\zeta(3)
-\frac{\pi}{p}(-1)^{(p+1)/2}G
\\
&+
\frac{3\pi}{8p}
\sum_{j=1}^{p-1}
\bigl(1-(-1)^{j}\bigr)j\,
\Cltwo\!\left(\frac{2\pi j}{p}\right)
+
\frac{\pi}{4p}
\sum_{j=1}^{p-1}
(-1)^{j}j\,
\Cltwo\!\left(\frac{4\pi j}{p}\right).
\label{eq:E2-theorem}
\eea
An empty sum is zero.
\end{theorem}

Section~\ref{sec:validation} provides an independent numerical audit of
Theorem~\ref{thm:Euler}. Direct summation of~\eqref{eq:E-def} agrees with
\eqref{eq:E1-theorem} for $p\leq12$ and with~\eqref{eq:E2-theorem} for odd
$p\leq11$. At $60$-digit working precision, the largest discrepancy is
$7.8\times10^{-61}$.

\subsection{Preliminary lemmas}
\label{sec:prelim}

\begin{lemma}[Alternating zeta values]
\label{lem:eta}
\be
 \zeta(2)=\frac{\pi^2}{6}\,,\qquad
 \sum_{n=1}^\infty\frac{(-1)^{n+1}}{n^2}=\frac{\pi^2}{12}\,,\qquad
 \sum_{n=1}^\infty\frac{(-1)^{n+1}}{n^3}=\frac34\zeta(3)\,.
 \label{eq:eta-values}
\ee
\end{lemma}

\begin{proof}
Apply Parseval's identity to the $2\pi$-periodic extension of $f(x)=x$ on $(-\pi,\pi)$.  Its cosine coefficients vanish by oddness.  Its sine coefficient is
\bea
 b_n&=\frac2\pi\int_0^\pi x\sin(nx)\,\rd x\\
 &=\frac2\pi\left(\left[-\frac{x\cos(nx)}n\right]_0^\pi
 +\frac1n\int_0^\pi\cos(nx)\,\rd x\right)
 =\frac{2(-1)^{n+1}}n\,.
\eea
Parseval gives
\begin{equation*}
 \frac1\pi\int_{-\pi}^{\pi}x^2\,\rd x
 =\sum_{n=1}^\infty b_n^2\,.
\end{equation*}
The left side is $2\pi^2/3$, while the right side is $4\zeta(2)$.  Therefore $\zeta(2)=\pi^2/6$.

Absolute convergence permits separation into even and odd indices.  For $s=2,3$,
\bea
 \sum_{n=1}^\infty\frac{(-1)^{n+1}}{n^s}
 &=\zeta(s)-2\sum_{n=1}^\infty\frac1{(2n)^s}
 =\bigl(1-2^{1-s}\bigr)\zeta(s)\,.
\eea
Substitution of $s=2$ and $s=3$ proves \eqref{eq:eta-values}.
\end{proof}

\begin{lemma}[Generalized harmonic numbers]
\label{lem:H}
For every real $z\ge0$, the integral \eqref{eq:H-def} converges and equals \eqref{eq:H-series}.  For every integer $N\ge1$, $\Hgen_N=H_N=\sum_{r=1}^N1/r$.  For fixed positive integers $a,p$,
\begin{equation*}
 \Hgen_{am/p}=O(\log m)\,,\qquad m\to\infty\,,
\end{equation*}
and the series \eqref{eq:E-def} converges absolutely.
\end{lemma}

\begin{proof}
For $0\le x<1$,
\begin{equation*}
 \frac1{1-x}=\sum_{n=0}^\infty x^n\,.
\end{equation*}
When $z\ge0$, every function $(1-x^z)x^n$ is nonnegative.  Tonelli's theorem gives
\bea
 \int_0^1\frac{1-x^z}{1-x}\,\rd x
 &=\sum_{n=0}^\infty\int_0^1(x^n-x^{n+z})\,\rd x\\
 &=\sum_{n=0}^\infty\left(\frac1{n+1}-\frac1{n+z+1}\right)
 =\sum_{n=1}^\infty\left(\frac1n-\frac1{n+z}\right).
\eea
For $z>0$, the summand equals $z/[n(n+z)]\le z/n^2$, so the series converges.  For $z=0$, every summand vanishes.

If $z=N\in\mathbb Z_{>0}$, then
\begin{equation*}
 \frac{1-x^N}{1-x}=1+x+\cdots+x^{N-1}\,.
\end{equation*}
Integration of this finite sum gives $\Hgen_N=\sum_{r=1}^N1/r$.

Let $N_m=\lceil am/p\rceil$.  Since $am/p\le N_m$ and $0<x<1$, one has $x^{am/p}\ge x^{N_m}$, and therefore $\Hgen_{am/p}\le H_{N_m}$.  For $N\ge1$,
\bea
 H_N&=1+\sum_{r=2}^N\frac1r
 \le1+\int_1^N\frac{\rd t}{t}=1+\log N\,.
\eea
Thus $\Hgen_{am/p}\le K^{(0)}_{a,p}+K^{(1)}_{a,p}\log m$ for constants $K^{(0)}_{a,p},K^{(1)}_{a,p}$ depending only on $a,p$.  Finally,
\begin{equation*}
 \sum_{m=1}^\infty\frac{\Hgen_{am/p}}{m^2}
 \le K^{(0)}_{a,p}\zeta(2)+K^{(1)}_{a,p}\sum_{m=1}^\infty\frac{\log m}{m^2}<\infty\,,
\end{equation*}
because the final series converges by the integral test.
\end{proof}

\begin{lemma}[Logarithm series]
\label{lem:log}
For $|z|<1$,
\be
 -\mathrm{Log}(1-z)=\sum_{n=1}^\infty\frac{z^n}{n}\,.
 \label{eq:log-series}
\ee
The convergence is absolute and uniform on every closed disk $|z|\le r<1$.
\end{lemma}

\begin{proof}
For $|z|\le r<1$ and $0\le t\le1$, the geometric series
\begin{equation*}
 \frac{z}{1-tz}=\sum_{n=0}^\infty z^{n+1}t^n
\end{equation*}
is uniformly convergent by the Weierstrass test.  Moreover,
$\re(1-tz)\ge1-t|z|\ge1-r>0$, so the segment $1-tz$ remains in the
open right half-plane and does not meet the principal branch cut.  Hence
\bea
 -\mathrm{Log}(1-z)&=\int_0^1\frac{z}{1-tz}\,\rd t
 =\sum_{n=0}^\infty z^{n+1}\int_0^1t^n\,\rd t
 =\sum_{n=1}^\infty\frac{z^n}{n}\,.
\eea
\end{proof}

\begin{lemma}[Fourier identities]
\label{lem:Fourier}
For $\theta\notin2\pi\mathbb Z$,
\bea
 \Cltwo'(\theta)&=-\log\left(2\left|\sin\frac\theta2\right|\right),\\
 \cC_3'(\theta)&=-\Cltwo(\theta)\,,\\
 \cC_3''(\theta)&=\log\left(2\left|\sin\frac\theta2\right|\right).
 \label{eq:Fourier-derivatives}
\eea
For every real $x$,
\be
 \Cltwo(\pi-x)=\Cltwo(x)-\frac12\Cltwo(2x)\,.
 \label{eq:Cl-dup}
\ee
Moreover,
\be
 \cC_3(0)=\zeta(3)\,,\qquad
 \cC_3(\pi)=-\frac34\zeta(3)\,,\qquad
 \int_0^\pi x\Cltwo(x)\,\rd x=\frac{3\pi}{4}\zeta(3)\,.
 \label{eq:weighted-Cl}
\ee
\end{lemma}

\begin{proof}
For $0<r<1$, define
\begin{equation*}
 \mathscr C_r(\theta)=\sum_{n=1}^\infty\frac{r^n\sin(n\theta)}{n^2}\,.
\end{equation*}
The function and derivative series converge uniformly, so
\bea
 \mathscr C_r'(\theta)&=\sum_{n=1}^\infty\frac{r^n\cos(n\theta)}n
 =-\re\mathrm{Log}(1-r\e^{\ii\theta})
 =-\log|1-r\e^{\ii\theta}|\,.
\eea
Let $K_\theta$ be a compact interval disjoint from $2\pi\mathbb Z$.  The quantity $|1-\e^{\ii\theta}|$ has a positive minimum on $K_\theta$.  Thus $\mathscr C_r'$ converges uniformly on $K_\theta$ to
\begin{equation*}
 -\log|1-\e^{\ii\theta}|=-\log\left(2\left|\sin\frac\theta2\right|\right).
\end{equation*}
Also $\mathscr C_r\to\Cltwo$ uniformly on $\mathbb R$, because
\begin{equation*}
 \sup_\theta|\mathscr C_r(\theta)-\Cltwo(\theta)|
 \le\sum_{n=1}^\infty\frac{1-r^n}{n^2}\longrightarrow0
\end{equation*}
by dominated convergence for series.  Integrating $\mathscr C_r'$ between two points of $K_\theta$ and taking $r\uparrow1$ proves the first line of \eqref{eq:Fourier-derivatives}.  The derivative series of $\cC_3$ is absolutely and uniformly convergent, so $\cC_3'=-\Cltwo$.  Differentiating on $K_\theta$ gives the third identity.

For \eqref{eq:Cl-dup}, absolute convergence permits coefficient comparison:
\begin{equation*}
 \Cltwo(\pi-x)=\sum_{n=1}^\infty\frac{(-1)^{n+1}\sin(nx)}{n^2}\,.
\end{equation*}
In $\Cltwo(x)-\tfrac12\Cltwo(2x)$, the coefficient of $\sin(nx)$ is $1/n^2$ for odd $n$ and $-1/n^2$ for even $n$.  This equals $(-1)^{n+1}/n^2$.

The endpoint values of $\cC_3$ are direct consequences of Lemma~\ref{lem:eta}.  Finally,
\begin{equation*}
 \sum_{n=1}^\infty\frac1{n^2}\int_0^\pi|x\sin(nx)|\,\rd x
 \le\pi^2\zeta(2)<\infty\,.
\end{equation*}
Fubini's theorem therefore permits termwise integration.  Integration by parts gives
\begin{equation*}
 \int_0^\pi x\sin(nx)\,\rd x=-\frac{\pi(-1)^n}{n}\,.
\end{equation*}
Therefore
\begin{equation*}
 \int_0^\pi x\Cltwo(x)\,\rd x
 =-\pi\sum_{n=1}^\infty\frac{(-1)^n}{n^3}
 =\frac{3\pi}{4}\zeta(3)\,.
\end{equation*}
\end{proof}

\begin{lemma}[Root-of-unity filter]
\label{lem:filter}
Let $p\ge1$ and $\varpi_p=\e^{2\pi\ii/p}$.  For every integer $n$,
\be
 \sum_{r=0}^{p-1}\varpi_p^{nr}=
 \begin{cases}p,&p\mid n,\\0,&p\nmid n.\end{cases}
 \label{eq:filter}
\ee
\end{lemma}

\begin{proof}
If $p\mid n$, every term is one.  If $p\nmid n$, then $\varpi_p^n\ne1$ and the finite geometric sum is
\begin{equation*}
 \frac{1-(\varpi_p^n)^p}{1-\varpi_p^n}=0\,.
\end{equation*}
\end{proof}

\subsection{Integral reduction}
\label{sec:integral-reduction}

\begin{proposition}[Convergence and the Euler representation]
\label{prop:Euler-reduction}
The integral \eqref{eq:Ik-def} converges absolutely for every real $k>0$, and
\be
 \Cint(k)=-\frac{k^2\zeta(3)}{16\pi}
 -\frac{2\zeta(3)}{\pi k}
 -\frac{k}{\pi}\sum_{m=1}^\infty\frac{(-1)^{m+1}}{m^2}\Hgen_{2m/k}\,.
 \label{eq:C-Euler}
\ee
\end{proposition}

\begin{proof}
As $x\to0^+$,
\begin{equation*}
 \log(2\cosh x)=\log2+O(x^2)\,,\qquad
 \log(1-\e^{-kx})=\log(kx)+O(x)\,.
\end{equation*}
The product is $O(1+|\log x|)$, which is integrable near zero.  As $x\to\infty$,
\begin{equation*}
 \log(2\cosh x)=x+O(\e^{-2x})\,,\qquad
 \log(1-\e^{-kx})=-\e^{-kx}+O(\e^{-2kx})\,,
\end{equation*}
so the product is $O(x\e^{-kx})$.  This proves absolute convergence.

Set $z=\e^{-x}$.  Then $x=-\log z$, $\rd x=-\rd z/z$, and
\begin{equation*}
 \log(2\cosh x)=-\log z+\log(1+z^2)\,.
\end{equation*}
Reversing the limits gives
\be
 \mathfrak I_k=\int_0^1\frac{[-\log z+\log(1+z^2)]\log(1-z^k)}{z}\,\rd z\,.
 \label{eq:I-z}
\ee
For $0\le z<1$,
\begin{equation*}
 \log(1-z^k)=-\sum_{n=1}^\infty\frac{z^{kn}}n\,,
 \qquad
 \log(1+z^2)=\sum_{m=1}^\infty\frac{(-1)^{m+1}z^{2m}}m\,.
\end{equation*}
The first interchange is justified by
\bea
 \sum_{n=1}^\infty\frac1n\int_0^1z^{kn-1}|\log z|\,\rd z
 &=\sum_{n=1}^\infty\frac1{n(kn)^2}
 =\frac{\zeta(3)}{k^2}<\infty\,.
\eea
Here
\begin{equation*}
 \int_0^1z^{\alpha-1}\log z\,\rd z=-\frac1{\alpha^2}\,,\qquad \alpha>0\,,
\end{equation*}
which is obtained by integration by parts; the boundary $z^\alpha\log z$ vanishes at zero.  Thus the $-\log z$ part of \eqref{eq:I-z} equals $-\zeta(3)/k^2$.

For the product of logarithm series, set $k_{\min}=\min(1,k)>0$.  Since
$2m+kn\ge k_{\min}(m+n)\ge2k_{\min}\sqrt{mn}$,
\bea
 \sum_{m,n\ge1}\frac1{mn}\int_0^1z^{2m+kn-1}\,\rd z
 &=\sum_{m,n\ge1}\frac1{mn(2m+kn)}
 \le\frac1{2k_{\min}}\sum_{m,n\ge1}\frac1{m^{3/2}n^{3/2}}<\infty\,.
\eea
Fubini's theorem gives
\be
 \mathfrak I_k=-\frac{\zeta(3)}{k^2}-\mathfrak S_k\,,\qquad
 \mathfrak S_k=\sum_{m,n\ge1}\frac{(-1)^{m+1}}{mn(2m+kn)}\,.
 \label{eq:Sk}
\ee
The exact partial fraction
\be
 \frac1{n(2m+kn)}
 =\frac1{2m}\left(\frac1n-\frac1{n+2m/k}\right)
\ee
and absolute convergence permit summation in $n$ first.  Equation~\eqref{eq:H-series} gives
\begin{equation*}
 \mathfrak S_k=\frac12\sum_{m=1}^\infty\frac{(-1)^{m+1}}{m^2}\Hgen_{2m/k}\,.
\end{equation*}
Substitution into \eqref{eq:Cint-def} proves \eqref{eq:C-Euler}.
\end{proof}

From this point through the parity decomposition, let $k\in\mathbb Z_{>0}$.  Set
\be
 d=\gcd(k,2)\,,\qquad p=\frac{k}{d}\,,\qquad a=\frac2d\,.
 \label{eq:ap}
\ee
Then $2m/k=am/p$, $a=2$ and $p=k$ when $k$ is odd, while $a=1$ and $p=k/2$ when $k$ is even.

\begin{proposition}[Double-logarithm representation]
\label{prop:double-log}
For positive integers $a,p$,
\be
 \cE_{a,p}=-a\int_0^1\frac{\log(1-u^p)\log(1+u^a)}{u}\,\rd u\,.
 \label{eq:double-log}
\ee
The integral is absolutely convergent.
\end{proposition}

\begin{proof}
Insert \eqref{eq:H-def} into \eqref{eq:E-def}.  Lemma~\ref{lem:H} gives
\begin{equation*}
 \sum_{m=1}^\infty\frac1{m^2}\int_0^1\frac{|1-x^{am/p}|}{1-x}\,\rd x
 =\sum_{m=1}^\infty\frac{\Hgen_{am/p}}{m^2}<\infty\,.
\end{equation*}
Fubini's theorem permits interchange:
\bea
 \cE_{a,p}
 &=\int_0^1\frac1{1-x}\sum_{m=1}^\infty
 \frac{(-1)^{m+1}}{m^2}(1-x^{am/p})\,\rd x
 =\int_0^1\frac{\pi^2/12+\Li_2(-x^{a/p})}{1-x}\,\rd x\,.
 \label{eq:dilog-integral}
\eea
Lemma~\ref{lem:eta} gives the $\pi^2/12$ term.  Absolute convergence of the dilogarithm series gives the second term.

Let
\begin{equation*}
 \Psi_{a,p}(x)=\frac{\pi^2}{12}+\Li_2(-x^{a/p})\,.
\end{equation*}
Since $\Li_2(-1)=-\pi^2/12$, $\Psi_{a,p}(1)=0$.  On $0<x<1$, termwise differentiation on compact subintervals gives
\be
 \Psi_{a,p}'(x)=-\frac ap\frac{\log(1+x^{a/p})}{x}\,.
 \label{eq:Fprime}
\ee
The right side extends continuously to $x=1$.  The mean-value theorem therefore gives $\Psi_{a,p}(x)=O(1-x)$ near one, so $\Psi_{a,p}(x)\log(1-x)\to0$.  At zero, $\Psi_{a,p}$ is bounded and $\log(1-x)\to0$.

Integrate \eqref{eq:dilog-integral} by parts on $[\varepsilon,1-\varepsilon]$, using $\rd x/(1-x)=\rd[-\log(1-x)]$.  Both boundary terms vanish in the limit.  The remaining integral is absolutely convergent: near zero, $\log(1+x^{a/p})=O(x^{a/p})$ and $\log(1-x)=O(x)$. Near one, $\Psi_{a,p}'$ is bounded and $|\log(1-x)|$ is integrable.  Thus
\begin{equation*}
 \cE_{a,p}=-\frac ap\int_0^1\frac{\log(1-x)\log(1+x^{a/p})}{x}\,\rd x\,.
\end{equation*}
Set $x=u^p$.  Since $\rd x/x=p\,\rd u/u$, equation~\eqref{eq:double-log} results.  Its endpoint behavior is $O(u^{p+a-1})$ at zero and $O(|\log(1-u)|)$ at one, proving absolute convergence.
\end{proof}

The double-logarithm representation of Proposition~\ref{prop:double-log}
is the starting point for the root-of-unity evaluation. It replaces the
harmonic numbers by a single integral on which the finite filter acts directly.

\subsection{Roots of unity and kernel normalization}
\label{sec:roots}

Let $\varpi_p=\e^{2\pi\ii/p}$.  For $0\le u<1$,
\begin{equation*}
 1-u^p=\prod_{r=0}^{p-1}(1-\varpi_p^r u)\,.
\end{equation*}
Moreover,
\begin{equation*}
 \re(1-\varpi_p^r u)=1-u\cos(2\pi r/p)\ge1-u>0\,.
\end{equation*}
Every factor lies in the open right half-plane.  The function
\begin{equation*}
 \mathfrak D_p(u)=\sum_{r=0}^{p-1}\mathrm{Log}(1-\varpi_p^r u)-\log(1-u^p)
\end{equation*}
is continuous, satisfies $\e^{\mathfrak D_p(u)}=1$, and has $\mathfrak D_p(0)=0$.  Its values lie in the discrete set $2\pi\ii\mathbb Z$, so continuity gives
\be
 \log(1-u^p)=\sum_{r=0}^{p-1}\mathrm{Log}(1-\varpi_p^r u)\,.
 \label{eq:root-factor}
\ee

For $a=1,2$, define
\bea
 \mathcal J_a(\theta)&=\int_0^1\frac{\mathrm{Log}(1-\e^{\ii\theta}u)\log(1+u^a)}{u}\,\rd u\,,\\
 \mathcal Q_a(\theta)&=\re \mathcal J_a(\theta)\,.
 \label{eq:kernel-def}
\eea
The integral converges for every real $\theta$.  Near $u=0$, its integrand is $O(u^a)$.  When $\theta\notin2\pi\mathbb Z$, the first logarithm is bounded near $u=1$; when $\theta\in2\pi\mathbb Z$, it is $\log(1-u)$, which is integrable.

Expanding both logarithms and using
\begin{equation*}
 \sum_{n,m\ge1}\frac1{nm(n+am)}
 \le\frac12\zeta\!\left(\frac32\right)^2<\infty
\end{equation*}
gives, by Fubini's theorem,
\be
 \mathcal Q_a(\theta)=-\sum_{n,m\ge1}
 \frac{(-1)^{m+1}\cos(n\theta)}{nm(n+am)}\,.
 \label{eq:kernel-series}
\ee
Hence
\be
 \mathcal Q_a(2\pi-\theta)=\mathcal Q_a(\theta)\,,
 \qquad
 \int_0^\pi\mathcal Q_a(\theta)\,\rd\theta=0\,.
 \label{eq:kernel-symmetry}
\ee
The mean vanishes because uniform absolute convergence permits termwise integration over $[0,2\pi]$, every cosine mode integrates to zero, and the reflection symmetry then halves the interval.

For $0<u_0<1$, insert \eqref{eq:root-factor} into
\eqref{eq:double-log} and integrate over $0\le u\le u_0$.  The sum over
$r$ is finite.  As $u_0\uparrow1$, the $r=0$ term is bounded near the
endpoint by a constant multiple of $|\log(1-u)|$, while every term with
$r\ne0$ is bounded there.  The factor $\log(1+u^a)/u$ is bounded near
$u=1$, and the behavior at $u=0$ was already recorded above.  Dominated
convergence therefore gives
\be
 \cE_{a,p}=-a\sum_{r=0}^{p-1}\mathcal Q_a\!\left(\frac{2\pi r}{p}\right).
 \label{eq:E-root}
\ee

\subsection{The kernel for \texorpdfstring{$a=1$}{a=1}}
\label{sec:K1}

\begin{proposition}
\label{prop:K1}
For $0\le\theta\le\pi$,
\be
 \mathcal Q_1(\theta)=\frac12\cC_3(\theta)+\cC_3(\pi-\theta)
 -\frac{\theta-\pi}{2}\Cltwo(\pi-\theta)-\frac38\zeta(3)\,.
 \label{eq:K1}
\ee
\end{proposition}

\begin{proof}
Fix $0<\delta<\pi/2$ and $\delta\le\theta\le\pi-\delta$.  Put $\xi_\theta=\e^{\ii\theta}$.  For $0\le u\le1$,
\begin{equation*}
 |1-\xi_\theta u|^2=(u-\cos\theta)^2+\sin^2\theta\ge\sin^2\delta\,.
\end{equation*}
The first two $\theta$ derivatives of the integrand in $\mathcal J_1$ are
\begin{equation*}
 -\ii \xi_\theta\frac{\log(1+u)}{1-\xi_\theta u}\,,
 \qquad
 \xi_\theta\frac{\log(1+u)}{(1-\xi_\theta u)^2}\,.
\end{equation*}
Their absolute values are bounded by $\log2/\sin\delta$ and $\log2/\sin^2\delta$, respectively.  Dominated differentiation gives
\be
 \mathcal J_1''(\theta)=\int_0^1\frac{\xi_\theta\log(1+u)}{(1-\xi_\theta u)^2}\,\rd u\,.
 \label{eq:J1second-start}
\ee
Since $\rd(1-\xi_\theta u)^{-1}/\rd u=\xi_\theta(1-\xi_\theta u)^{-2}$, integration by parts gives
\be
 \mathcal J_1''(\theta)=\frac{\log2}{1-\xi_\theta}
 -\int_0^1\frac{\rd u}{(1+u)(1-\xi_\theta u)}
 =\frac{\log2}{1-\xi_\theta}
 -\frac{\log2-\mathrm{Log}(1-\xi_\theta)}{1+\xi_\theta}\,.
 \label{eq:J1second}
\ee
For $0<\theta<\pi$,
\bea
 \mathrm{Log}(1-\xi_\theta)&=\ell_\theta+\ii\varphi_\theta\,, \qquad
 \ell_\theta=\log\left(2\sin\frac\theta2\right)\,,\qquad
 \varphi_\theta=\frac{\theta-\pi}{2}\,,\\
 \frac1{1-\xi_\theta}&=\frac12+\frac\ii2\cot\frac\theta2\,,\qquad
 \frac1{1+\xi_\theta}=\frac12-\frac\ii2\tan\frac\theta2\,.
\eea
Taking the real part of \eqref{eq:J1second} gives
\be
 \mathcal Q_1''(\theta)=\frac12\ell_\theta+
 \frac{\theta-\pi}{4}\tan\frac\theta2\,.
 \label{eq:K1second}
\ee

Let $\mathscr F_1$ denote the right side of \eqref{eq:K1}.  Lemma~\ref{lem:Fourier} gives
\begin{equation*}
 \frac{\rd^2}{\rd\theta^2}\frac12\cC_3(\theta)=\frac12\ell_\theta\,,
 \qquad
 \frac{\rd^2}{\rd\theta^2}\cC_3(\pi-\theta)=\log\left(2\cos\frac\theta2\right).
\end{equation*}
If $\mathfrak c_1(\theta)=\Cltwo(\pi-\theta)$, then
\begin{equation*}
 \mathfrak c_1'(\theta)=\log\left(2\cos\frac\theta2\right),
 \qquad \mathfrak c_1''(\theta)=-\frac12\tan\frac\theta2\,.
\end{equation*}
Thus
\bea
 \frac{\rd^2}{\rd\theta^2}
 \left[-\frac{\theta-\pi}{2}\mathfrak c_1(\theta)\right]
 &=-\mathfrak c_1'(\theta)-\frac{\theta-\pi}{2}\mathfrak c_1''(\theta)
 =-\log\left(2\cos\frac\theta2\right)
 +\frac{\theta-\pi}{4}\tan\frac\theta2\,.
\eea
The cosine logarithms cancel, so $\mathscr F_1''=\mathcal Q_1''$.  Hence $\mathcal Q_1-\mathscr F_1$ is affine on $(0,\pi)$.

At $\theta=\pi$, dominated differentiation gives
\begin{equation*}
 \mathcal J_1'(\pi)=\ii\int_0^1\frac{\log(1+u)}{1+u}\,\rd u\,,
\end{equation*}
which is purely imaginary; hence $\mathcal Q_1'(\pi)=0$.  Direct differentiation of $\mathscr F_1$ gives
\begin{equation*}
 \mathscr F_1'(\theta)=-\frac12\Cltwo(\theta)+\frac12\Cltwo(\pi-\theta)
 -\frac{\theta-\pi}{2}\log\left(2\cos\frac\theta2\right).
\end{equation*}
Every term tends to zero as $\theta\to\pi^-$; the last limit uses $x\log x\to0$ as $x\to0^+$.  Thus $\mathscr F_1'(\pi)=0$, and $\mathcal Q_1-\mathscr F_1$ is constant.

Equation~\eqref{eq:kernel-symmetry} gives $\int_0^\pi\mathcal Q_1(\theta)\,\rd \theta=0$.  The two $\cC_3$ terms in $\mathscr F_1$ have zero integral.  By \eqref{eq:weighted-Cl},
\bea
 -\frac12\int_0^\pi(\theta-\pi)\Cltwo(\pi-\theta)\,\rd\theta
 &=\frac12\int_0^\pi x\Cltwo(x)\,\rd x
 =\frac{3\pi}{8}\zeta(3)\,.
\eea
The constant term integrates to $-3\pi\zeta(3)/8$.  Hence
$\int_0^\pi \mathscr F_1(\theta)\,\rd \theta=0$, and the constant difference is zero.  The uniformly
convergent series \eqref{eq:kernel-series} makes $\mathcal Q_1$ continuous on
$\mathbb R$, and the right side of \eqref{eq:K1} has finite endpoint
limits.  The identity therefore extends from $(0,\pi)$ to
$[0,\pi]$.
\end{proof}

\subsection{Evaluation of \texorpdfstring{$\cE_{1,p}$}{E(1,p)}}
\label{sec:E1}

For $r=0,\ldots,p-1$, set
\begin{equation*}
 \theta_r=\frac{2\pi r}{p}\,,\qquad
 \widehat\theta_r=\min(\theta_r,2\pi-\theta_r)\in[0,\pi]\,.
\end{equation*}
The symmetry of $\mathcal Q_1$ permits replacement of $\theta_r$ by $\widehat\theta_r$.
Since $\cos(n\widehat\theta_r)=\cos(n\theta_r)$, uniform convergence and
Lemma~\ref{lem:filter} give
\be
 \sum_{r=0}^{p-1}\cC_3(\widehat\theta_r)
 =\sum_{n=1}^\infty\frac1{n^3}
   \sum_{r=0}^{p-1}\cos(n\theta_r)
 =p\sum_{\ell=1}^\infty\frac1{(p\ell)^3}
 =\frac{\zeta(3)}{p^2}\,.
 \label{eq:C3-root-1}
\ee
Also $\cos(n(\pi-\widehat\theta_r))=(-1)^n\cos(n\theta_r)$, so
\be
 \sum_{r=0}^{p-1}\cC_3(\pi-\widehat\theta_r)
 =p\sum_{\ell=1}^\infty\frac{(-1)^{p\ell}}{(p\ell)^3}
 =\frac{1+7(-1)^p}{8p^2}\zeta(3)\,.
 \label{eq:C3-root-2}
\ee
The final expression equals $\zeta(3)/p^2$ for even $p$ and
$-3\zeta(3)/(4p^2)$ for odd $p$.

Let $j_{\max}=\lfloor(p-1)/2\rfloor$ and, for every integer $r$, put
\begin{equation*}
 \mathfrak c_r=\Cltwo\!\left(\frac{2\pi r}{p}\right).
\end{equation*}
Periodicity and oddness of $\Cltwo$ give
\be
 \mathfrak c_{r+p}=\mathfrak c_r\,,\qquad \mathfrak c_{p-r}=-\mathfrak c_r\,.
 \label{eq:c-symmetry}
\ee
The terms $r=j$ and $r=p-j$ in \eqref{eq:E-root} have the same
$\widehat\theta_r=2\pi j/p$.  The term $r=0$ contributes zero to the weighted
Clausen part because $\Cltwo(\pi)=0$; when $p$ is even, the unpaired
term $r=p/2$ also contributes zero because $\widehat\theta_r-\pi=0$.  Using
\eqref{eq:Cl-dup}, one obtains
\bea
 \cE_{1,p}={}&\left(\frac{3p}{8}-\frac{5+7(-1)^p}{8p^2}\right)\zeta(3)+T_{1,p}\,,\\
 T_{1,p}&=\frac\pi p\sum_{j=1}^{j_{\max}}(2j-p)
 \left(\mathfrak c_j-\frac12\mathfrak c_{2j}\right).
 \label{eq:E1-half}
\eea
We prove the finite identity
\be
 4\sum_{j=1}^{j_{\max}}(2j-p)\left(\mathfrak c_j-\frac12\mathfrak c_{2j}\right)
 =\sum_{r=1}^{p-1}[3-(-1)^{p+r}]r \mathfrak c_r\,.
 \label{eq:E1-finite}
\ee

First take $p=2P$.  Pair $r=j$ with $r=p-j$ on the right side of
\eqref{eq:E1-finite}.  These two indices have the same parity, and
\eqref{eq:c-symmetry} gives $\mathfrak c_{p-j}=-\mathfrak c_j$.  Their combined
contribution is
\begin{equation*}
 [3-(-1)^j]\,[j-(p-j)]\mathfrak c_j
 =[3-(-1)^j](2j-p)\mathfrak c_j\,.
\end{equation*}
The right side of \eqref{eq:E1-finite} is therefore
\be
 \mathcal R_{\rm ev}=\sum_{j=1}^{P-1}[3-(-1)^j](2j-p)\mathfrak c_j\,.
 \label{eq:R-even}
\ee
Set
\begin{equation*}
 \mathcal A_{\rm ev}=\sum_{j=1}^{P-1}(2j-p)\mathfrak c_j\,,
 \qquad
 \mathcal B_{\rm ev}=\sum_{j=1}^{P-1}(2j-p)\mathfrak c_{2j}\,.
\end{equation*}
In $\mathcal B_{\rm ev}$, put $s=2j$.  The index $s$ runs through the even residues
$2,4,\ldots,p-2$, and
\begin{equation*}
 \mathcal B_{\rm ev}=\sum_{\substack{2\le s\le p-2\\s\ {\rm even}}}(s-p)\mathfrak c_s\,.
\end{equation*}
Pair $s$ with $p-s$.  If the self-paired residue $s=p/2$ is even, its
term vanishes because $\mathfrak c_{p/2}=\Cltwo(\pi)=0$.  For every even
$s$ with $1\le s\le P-1$, the pair contributes
\be
 (s-p)\mathfrak c_s+[(p-s)-p]\mathfrak c_{p-s}
 =(s-p)\mathfrak c_s-s(-\mathfrak c_s)
 =(2s-p)\mathfrak c_s\,.
\ee
Consequently,
\begin{equation*}
 \mathcal B_{\rm ev}=\sum_{\substack{1\le s\le P-1\\s\ {\rm even}}}(2s-p)\mathfrak c_s\,,
\end{equation*}
and the indicator of even $s$ gives
\begin{equation*}
 2\mathcal B_{\rm ev}=\sum_{s=1}^{P-1}[1+(-1)^s](2s-p)\mathfrak c_s\,.
\end{equation*}
Thus
\be
 4\mathcal A_{\rm ev}-2\mathcal B_{\rm ev}
 =\sum_{s=1}^{P-1}[4-(1+(-1)^s)](2s-p)\mathfrak c_s
 =\sum_{s=1}^{P-1}[3-(-1)^s](2s-p)\mathfrak c_s=\mathcal R_{\rm ev}\,.
\ee
This proves \eqref{eq:E1-finite} for even $p$.

Now take $p=2P+1$.  Pairing $r=j$ and $r=p-j$ on the right side of
\eqref{eq:E1-finite} gives
\be
 [3+(-1)^j]j \mathfrak c_j+[3-(-1)^j](p-j)\mathfrak c_{p-j}
 =[3(2j-p)+p(-1)^j]\mathfrak c_j\,.
\ee
Hence
\be
 \mathcal R_{\rm odd}=\sum_{j=1}^P[3(2j-p)+p(-1)^j]\mathfrak c_j\,.
 \label{eq:R-odd}
\ee
Set
\begin{equation*}
 \mathcal A_{\rm odd}=\sum_{j=1}^P(2j-p)\mathfrak c_j\,,
 \qquad
 \mathcal B_{\rm odd}=\sum_{j=1}^P(2j-p)\mathfrak c_{2j}\,.
\end{equation*}
The residues $2j$ are the even integers $2,4,\ldots,p-1$.  If an even
residue $s$ satisfies $s\le P$, then its contribution is
$(s-p)\mathfrak c_s$.  If $s>P$, put $r=p-s$; then $r$ is odd,
$1\le r\le P$, and
\begin{equation*}
 (s-p)\mathfrak c_s=(-r)\mathfrak c_{p-r}=r \mathfrak c_r\,.
\end{equation*}
Therefore
\be
 \mathcal B_{\rm odd}=\sum_{\substack{1\le j\le P\\j\ {\rm even}}}(j-p)\mathfrak c_j
 +\sum_{\substack{1\le j\le P\\j\ {\rm odd}}}j \mathfrak c_j\,.
 \label{eq:B-odd}
\ee
Let $\mathcal R_{\rm alt}=\sum_{j=1}^P(-1)^j \mathfrak c_j$.  For even $j$, the
coefficient of $\mathfrak c_j$ in $\mathcal A_{\rm odd}-p\mathcal R_{\rm alt}$ is
$(2j-p)-p=2(j-p)$; for odd $j$, it is
$(2j-p)+p=2j$.  Equation~\eqref{eq:B-odd} therefore gives
$2\mathcal B_{\rm odd}=\mathcal A_{\rm odd}-p\mathcal R_{\rm alt}$. Hence
\be
 4\mathcal A_{\rm odd}-2\mathcal B_{\rm odd}=3\mathcal A_{\rm odd}+p\mathcal R_{\rm alt}
 =\sum_{j=1}^P[3(2j-p)+p(-1)^j]\mathfrak c_j=\mathcal R_{\rm odd}\,.
\ee
This proves \eqref{eq:E1-finite} for odd $p$.

Substitution of \eqref{eq:E1-finite} into \eqref{eq:E1-half} gives
\eqref{eq:E1-theorem}.

\subsection{The kernel for \texorpdfstring{$a=2$}{a=2}}
\label{sec:K2}

\begin{proposition}
\label{prop:K2}
For $0\le\theta\le\pi$,
\bea
 \mathcal Q_2(\theta)={}&\cC_3(\theta)+\frac14\cC_3(\pi-2\theta)
 -\frac{\theta-\pi}{4}\Cltwo(\pi-2\theta)\\
 &-\frac\pi4\left[\Cltwo\!\left(\frac\pi2+\theta\right)
 +\Cltwo\!\left(\frac\pi2-\theta\right)\right]
 -\frac{3}{32}\zeta(3)\,.
 \label{eq:K2}
\eea
\end{proposition}

\begin{proof}
Fix $0<\delta<\pi/2$ and $\delta\le\theta\le\pi-\delta$, and put
$\xi_\theta=\e^{\ii\theta}$.  As in Proposition~\ref{prop:K1},
$|1-\xi_\theta u|\ge\sin\delta$ for $0\le u\le1$.  The first two
$\theta$ derivatives of the integrand in $\mathcal J_2$ are
\begin{equation*}
 -\ii \xi_\theta\frac{\log(1+u^2)}{1-\xi_\theta u}\,,
 \qquad
 \xi_\theta\frac{\log(1+u^2)}{(1-\xi_\theta u)^2}\,.
\end{equation*}
Their absolute values are bounded by $\log2/\sin\delta$ and
$\log2/\sin^2\delta$.  Dominated differentiation therefore gives
\begin{equation*}
 \mathcal J_2''(\theta)=\int_0^1\frac{\xi_\theta\log(1+u^2)}{(1-\xi_\theta u)^2}\,\rd u\,.
\end{equation*}
Integration by parts gives
\bea
 \mathcal J_2''(\theta)&=\frac{\log2}{1-\xi_\theta}
 -\int_0^1\frac{2u}{(1+u^2)(1-\xi_\theta u)}\,\rd u\\
 &=\frac{\log2}{1-\xi_\theta}
 +\frac{2\mathrm{Log}(1-\xi_\theta)-\log2+(\pi/2)\xi_\theta}{1+\xi_\theta^2}\,.
 \label{eq:J2second}
\eea
For $\theta\ne\pi/2$,
\begin{equation*}
 \frac1{1+\xi_\theta^2}=\frac12-\frac\ii2\tan\theta\,,
 \qquad
 \frac{\xi_\theta}{1+\xi_\theta^2}=\frac1{2\cos\theta}\,.
\end{equation*}
Taking real parts gives
\be
 \mathcal Q_2''(\theta)=\log\left(2\sin\frac\theta2\right)+\frac{\theta-\pi}{2}\tan\theta+\frac\pi4\sec\theta\,.
 \label{eq:K2second}
\ee
At $\theta=\pi/2+v$,
\begin{equation*}
 \tan\theta=-v^{-1}+O(v)\,,\qquad \sec\theta=-v^{-1}+O(v)\,,
\end{equation*}
and hence
\begin{equation*}
 \frac{\theta-\pi}{2}\tan\theta+\frac\pi4\sec\theta=-\frac12+O(v)\,.
\end{equation*}
Thus the singularity in \eqref{eq:K2second} is removable.

Let $\mathscr F_2$ be the right side of \eqref{eq:K2}.  Put $\mathfrak c_2(\theta)=\Cltwo(\pi-2\theta)$.  Lemma~\ref{lem:Fourier} gives
\begin{equation*}
 \mathfrak c_2'(\theta)=2\log(2|\cos\theta|)\,,\qquad \mathfrak c_2''(\theta)=-2\tan\theta\,,
\end{equation*}
and
\begin{equation*}
 \frac{\rd^2}{\rd\theta^2}\frac14\cC_3(\pi-2\theta)=\log(2|\cos\theta|)\,.
\end{equation*}
Therefore
\bea
 \frac{\rd^2}{\rd\theta^2}\left[-\frac{\theta-\pi}{4}\mathfrak c_2(\theta)\right]
 &=-\log(2|\cos\theta|)+\frac{\theta-\pi}{2}\tan\theta\,.
\eea
The two cosine logarithms cancel.

For the shifted Clausen contribution
\begin{equation*}
 \mathscr T_2(\theta)=-\frac\pi4\left[\Cltwo\!\left(\frac\pi2+\theta\right)
 +\Cltwo\!\left(\frac\pi2-\theta\right)\right],
\end{equation*}
Lemma~\ref{lem:Fourier} gives
\begin{equation*}
 \mathscr T_2'(\theta)=\frac\pi4\log\left|\tan\left(\frac\pi4+\frac\theta2\right)\right|,
 \qquad \mathscr T_2''(\theta)=\frac\pi4\sec\theta\,.
\end{equation*}
Consequently $\mathscr F_2''=\mathcal Q_2''$ on both sides of $\pi/2$.

The first derivative is
\bea
 \mathscr F_2'(\theta)=-\Cltwo(\theta)+\frac14\Cltwo(\pi-2\theta)
 -\frac{\theta-\pi}{2}\log(2|\cos\theta|)
 +\frac\pi4\log\left|\tan\left(\frac\pi4+\frac\theta2\right)\right|.
 \label{eq:F2prime}
\eea
At $\theta=\pi/2+v$,
\begin{equation*}
 \log(2|\cos\theta|)=\log(2|v|)+O(v^2)\,,
\end{equation*}
and
\begin{equation*}
 \log\left|\tan\left(\frac\pi4+\frac\theta2\right)\right|
 =\log\frac2{|v|}+O(v^2)\,.
\end{equation*}
The divergent logarithms cancel in \eqref{eq:F2prime}; $\mathscr F_2'$ extends continuously across $\pi/2$.  On $\pi/3\le\theta\le2\pi/3$ one has
$|1-u\e^{\ii\theta}|\ge\sqrt3/2$.  The kernel series
\eqref{eq:kernel-series} gives continuity of $\mathcal Q_2$, while dominated
differentiation gives
\begin{equation*}
 \mathcal Q_2'(\theta)=\re\left[-\ii\e^{\ii\theta}
 \int_0^1\frac{\log(1+u^2)}{1-u\e^{\ii\theta}}\,\rd u\right].
\end{equation*}
The integrand in this last integral is bounded in absolute value by
$2\log2/\sqrt3$, independently of $u$ and $\theta$ on the stated
interval.  Dominated convergence therefore proves continuity of
$\mathcal Q_2'$ at $\pi/2$.  The affine functions $\mathcal Q_2-\mathscr F_2$ on the two
sides of $\pi/2$ have the same value and slope there, and hence form one
affine function on $(0,\pi)$.

At $\theta=\pi$,
\begin{equation*}
 \mathcal J_2'(\pi)=\ii\int_0^1\frac{\log(1+u^2)}{1+u}\,\rd u\,,
\end{equation*}
so $\mathcal Q_2'(\pi)=0$.  Equation~\eqref{eq:F2prime} gives $\mathscr F_2'(\pi)=0$.  The difference is constant.

Equation~\eqref{eq:kernel-symmetry} gives $\int_0^\pi\mathcal Q_2(\theta)\,\rd\theta=0$.  The first two $\cC_3$ terms in $\mathscr F_2$ also have zero mean.  The two shifted Clausen terms cover one full period and integrate to zero.  For the weighted term, set $x=\pi-2\theta$:
\bea
 -\frac14\int_0^\pi(\theta-\pi)\Cltwo(\pi-2\theta)\,\rd\theta
 &=\frac1{16}\int_{-\pi}^{\pi}(\pi+x)\Cltwo(x)\,\rd x\\
 &=\frac18\int_0^\pi x\Cltwo(x)\,\rd x
 =\frac{3\pi}{32}\zeta(3)\,.
\eea
The constant term integrates to the negative of this value.  Hence
$\int_0^\pi \mathscr F_2(\theta) \, \rd \theta=0$, and the constant difference vanishes.  The kernel
series makes $\mathcal Q_2$ continuous on $\mathbb R$, and the right side of
\eqref{eq:K2} has finite endpoint limits.  The identity therefore
extends to $[0,\pi]$.
\end{proof}

\subsection{Evaluation of \texorpdfstring{$\cE_{2,p}$}{E(2,p)} for odd \texorpdfstring{$p$}{p}}
\label{sec:E2}

Assume that $p$ is odd, and use the angles $\widehat\theta_r$ from the preceding
section.  The first root sum is \eqref{eq:C3-root-1}.  For the doubled
mode,
\begin{equation*}
 \cos(n(\pi-2\widehat\theta_r))=(-1)^n\cos(2n\theta_r)\,.
\end{equation*}
Uniform convergence and the root filter give a contribution only when
$p\mid2n$.  Since $\gcd(p,2)=1$, this condition is equivalent to
$p\mid n$.  Writing $n=p\ell$ and using odd $p$ gives
\bea
 \sum_{r=0}^{p-1}\cC_3(\pi-2\widehat\theta_r)
 =p\sum_{\ell=1}^\infty\frac{(-1)^{p\ell}}{(p\ell)^3}
 =-\frac{3}{4p^2}\zeta(3)\,.
 \label{eq:C3-double-root}
\eea
The sine addition formula gives
\begin{equation*}
 \Cltwo\!\left(\frac\pi2+t\right)+\Cltwo\!\left(\frac\pi2-t\right)
 =2\sum_{n=1}^\infty\frac{\sin(n\pi/2)\cos(nt)}{n^2}\,.
\end{equation*}
Absolute convergence and the root filter yield
\bea
 \sum_{r=0}^{p-1}\left[\Cltwo\!\left(\frac\pi2+\widehat\theta_r\right)
 +\Cltwo\!\left(\frac\pi2-\widehat\theta_r\right)\right]
 &=2p\sum_{\ell=1}^\infty
 \frac{\sin(\ell p\pi/2)}{(p\ell)^2}\\
 &=\frac2p\Cltwo\!\left(\frac{p\pi}{2}\right)\\
 &=\frac2p(-1)^{(p-1)/2}G\,.
 \label{eq:Catalan-root}
\eea
The final equality uses the assumption that $p$ is odd.
Consequently, $p\pi/2$ is congruent to $\pi/2$ or $3\pi/2$ modulo $2\pi$.

Put $j_{\max}=(p-1)/2$ and retain the notation
$\mathfrak c_r=\Cltwo(2\pi r/p)$. Equations~\eqref{eq:E-root}, \eqref{eq:K2},
\eqref{eq:C3-double-root}, and~\eqref{eq:Catalan-root} give
\bea
\cE_{2,p}
&=
\left(\frac{3p}{16}-\frac{13}{8p^2}\right)\zeta(3)
-\frac{\pi}{p}(-1)^{(p+1)/2}G
+T_{2,p},
\\
T_{2,p}
&=
\frac{\pi}{p}
\sum_{j=1}^{j_{\max}}
(2j-p)
\left(
\mathfrak c_{2j}-\frac12\mathfrak c_{4j}
\right).
\label{eq:E2-half}
\eea

Define
\begin{equation*}
\mathcal O_1
=
\sum_{\substack{1\le r\le p-1\\ r\ {\rm odd}}}
r\,\mathfrak c_r\,,
\qquad
\mathcal O_2
=
\sum_{\substack{1\le r\le p-1\\ r\ {\rm odd}}}
r\,\mathfrak c_{2r}\,.
\end{equation*}
For the $\mathfrak c_{2j}$ term, set $s=p-2j$. As $j$ runs from $1$ to $j_{\max}$,
$s$ runs once through the odd integers $p-2,p-4,\ldots,1$. Moreover,
$2j-p=-s$ and $\mathfrak c_{2j}=\mathfrak c_{p-s}=-\mathfrak c_s$. Hence
\be
\sum_{j=1}^{j_{\max}}(2j-p)\mathfrak c_{2j}
=
\mathcal O_1\,.
\ee
For the $\mathfrak c_{4j}$ term,
$4j=2(p-s)\equiv-2s\pmod p$, so that
$\mathfrak c_{4j}=-\mathfrak c_{2s}$ and therefore
\be
\sum_{j=1}^{j_{\max}}(2j-p)\mathfrak c_{4j}
=
\mathcal O_2\,.
\ee
It follows that
\be
\frac{p}{\pi}T_{2,p}
=
\mathcal O_1-\frac12\mathcal O_2\,.
\label{eq:T2-OO}
\ee
Let
\begin{equation*}
 \mathcal V_2=\sum_{r=1}^{p-1}r \mathfrak c_{2r}\,.
\end{equation*}
Because $p$ is odd, multiplication by $2$ is a bijection on the
nonzero residue classes modulo $p$. Indeed, $2r_1\equiv2r_2\pmod p$
implies $p\mid2(r_1-r_2)$, and $\gcd(p,2)=1$ then gives $r_1\equiv r_2\pmod p$.  The domain is finite,
so injectivity implies bijectivity.  More explicitly, the unique
$r\in\{1,\ldots,p-1\}$ satisfying $2r\equiv s\pmod p$ is
\begin{equation*}
r=
\begin{cases}
s/2, & s\ \mathrm{even},\\[2pt]
(s+p)/2, & s\ \mathrm{odd}.
\end{cases}
\end{equation*}
Consequently,
\be
 \mathcal V_2=\frac12\sum_{\substack{1\le s\le p-1\\s\ {\rm even}}}s \mathfrak c_s
 +\frac12\sum_{\substack{1\le s\le p-1\\s\ {\rm odd}}}(s+p)\mathfrak c_s\,.
 \label{eq:W-reindex}
\ee
Define
\begin{equation*}
 \mathcal S_{\rm ev}=\sum_{\substack{1\le s\le p-1\\s\ {\rm even}}}s \mathfrak c_s\,,
 \qquad
 \mathcal S_{\rm odd}=\sum_{\substack{1\le s\le p-1\\s\ {\rm odd}}}\mathfrak c_s\,.
\end{equation*}
The map $s\mapsto t=p-s$ sends the even residues bijectively to the
odd residues.  Using $\mathfrak c_{p-t}=-\mathfrak c_t$,
\bea
 \mathcal S_{\rm ev}=\sum_{\substack{1\le t\le p-1\\t\ {\rm odd}}}(p-t)\mathfrak c_{p-t}
 =-\sum_{\substack{1\le t\le p-1\\t\ {\rm odd}}}(p-t)\mathfrak c_t
 =\mathcal O_1-p\mathcal S_{\rm odd}\,.
\eea
Substitution into \eqref{eq:W-reindex} gives
\begin{equation*}
 \mathcal V_2=\frac12(\mathcal S_{\rm ev}+\mathcal O_1+p\mathcal S_{\rm odd})=\mathcal O_1\,.
\end{equation*}

Now set
\begin{equation*}
 \mathcal V_{\rm alt}=\sum_{r=1}^{p-1}(-1)^r r \mathfrak c_{2r}\,.
\end{equation*}
The factor $(-1)^r$ reverses the sign of exactly the odd-$r$ terms, so
\begin{equation*}
 \mathcal V_{\rm alt}=\mathcal V_2-2\mathcal O_2=\mathcal O_1-2\mathcal O_2\,.
\end{equation*}
Consequently,
\bea
 &\frac38\sum_{r=1}^{p-1}(1-(-1)^r)r \mathfrak c_r
 +\frac14\sum_{r=1}^{p-1}(-1)^r r \mathfrak c_{2r}\\
 &\qquad=\frac34\mathcal O_1+\frac14(\mathcal O_1-2\mathcal O_2)
 =\mathcal O_1-\frac12\mathcal O_2\,.
\eea
Equation~\eqref{eq:T2-OO} gives precisely the two finite sums in
\eqref{eq:E2-theorem}.  Together with the preceding section, this
proves Theorem~\ref{thm:Euler}.

\subsection{Proof of the constant-map theorem}
\label{sec:assembly}

Proposition~\ref{prop:Euler-reduction} and \eqref{eq:ap} give
\be
 \Cint(k)=-\frac{k^2\zeta(3)}{16\pi}
 -\frac{2\zeta(3)}{\pi k}-\frac{k}{\pi}\cE_{a,p}\,.
 \label{eq:C-via-E}
\ee

For odd $k=p$, substituting \eqref{eq:E2-theorem} into \eqref{eq:C-via-E} gives the $\zeta(3)/\pi$ coefficient
\bea
 -\frac{p^2}{16}-\frac2p
 -p\left(\frac{3p}{16}-\frac{13}{8p^2}\right)
 &=-\frac{p^2}{4}-\frac{3}{8p}\,.
\eea
The Catalan coefficient is
\begin{equation*}
 -\frac p\pi\left[-\frac\pi p(-1)^{(p+1)/2}\right]
 =(-1)^{(p+1)/2}\,,
\end{equation*}
and the two Clausen prefactors are $-3/8$ and $-1/4$.  This proves \eqref{eq:clausen-odd}.

For even $k=2p$, substituting \eqref{eq:E1-theorem} gives the $\zeta(3)/\pi$ coefficient
\bea
 -\frac{p^2}{4}-\frac1p
 -2p\left(\frac{3p}{8}-\frac{5+7(-1)^p}{8p^2}\right)
 &=-p^2+\frac{1+7(-1)^p}{4p}\,.
\eea
The finite-sum prefactor is $-(2p/\pi)(\pi/(2p))=-1$.  This proves \eqref{eq:clausen-even} and completes Theorem~\ref{thm:C}.

\subsection{Special cases and independent numerical checks}
\label{sec:appchecks}

Setting $p=1$ in \eqref{eq:E1-theorem} gives
\be
 \sum_{m=1}^\infty\frac{(-1)^{m+1}H_m}{m^2}=\frac58\zeta(3)\,.
\ee
Setting $p=1$ in \eqref{eq:E2-theorem} gives
\be
 \sum_{m=1}^\infty\frac{(-1)^{m+1}H_{2m}}{m^2}
 =\pi G-\frac{23}{16}\zeta(3)\,.
\ee
Setting $p=2$ in \eqref{eq:E1-theorem}, and using $\Cltwo(\pi)=0$, gives
\be
 \sum_{m=1}^\infty\frac{(-1)^{m+1}\Hgen_{m/2}}{m^2}
 =\frac38\zeta(3)\,.
\ee
Substitution into \eqref{eq:C-via-E} yields
\begin{equation*}
 \Cint(1)=-G-\frac{5\zeta(3)}{8\pi}\,,\qquad
 \Cint(2)=-\frac{5\zeta(3)}{2\pi}\,,\qquad
 \Cint(4)=-\frac{3\zeta(3)}{\pi}\,.
\end{equation*}

As a numerical check distinct from the proof, the direct alternating series \eqref{eq:E-def} was compared with \eqref{eq:E1-theorem} for $p\le12$ and with \eqref{eq:E2-theorem} for odd $p\le11$, at $60$ working digits; the largest discrepancy was $7.8\times10^{-61}$.  Independently, \eqref{eq:Cint-def} was compared with \eqref{eq:clausen-odd}--\eqref{eq:clausen-even} for every $k\le50$ at $70$ digits, the largest discrepancy being $6.4\times10^{-69}$ at $k=45$ and tracking the working precision under a scan.  These comparisons provide independent numerical checks of the analytic formulas.

For a concrete value, at $k=3$,
\begin{equation*}
 \Cint(3)=G-\frac{19\zeta(3)}{8\pi}
 -\frac32\Cltwo\!\left(\frac{2\pi}{3}\right)
 =-1.00771417780646417281071274741419325060425\ldots\,.
\end{equation*}
Direct quadrature and direct summation of the Euler series reproduce the displayed digits.

\section{Constant-map identities}
\label{app:Smaps}

This appendix proves two analytic identities used in the main text.
Proposition~\ref{prop:CA} establishes Eq.~\eqref{eq:CAw}, which relates the
level derivative of the effective twisted superpotential constant map to the
difference $A(k/2)-A(k)$ of ABJM $S^3$ constant maps.  Theorem~\ref{thm:f0}
proves that the arithmetic representation~\eqref{eq:f0arith} and the
one-kernel representation~\eqref{eq:f0kernel} of the index constant map define
the same function.  Both identities hold for every real $k>0$ when the
twisted superpotential constant map is represented by its positive-real integral
continuation $\Cint(k)$; for positive integer $k$, Theorem~\ref{thm:C} gives
$\Cint(k)=\CW(k)$.  All integral manipulations below are justified by absolute
convergence.  The second identity additionally uses the functional equation of
the Riemann zeta function.

We retain the notation $\mathfrak I_k$ and $\Cint(k)$ introduced in
Eqs.~\eqref{eq:Ik-def}--\eqref{eq:Cint-def}.  For $k>0$, set
\be
 L_k(x)=\log\bigl(1-\e^{-kx}\bigr),\qquad
 \mathscr J_k=\int_0^\infty
 \frac{x\log(1-\e^{-2x})}{\e^{kx}-1}\,\rd x\,.
\label{eq:B-defs}
\ee
Thus
\be
 \mathfrak I_k=\int_0^\infty \log(2\cosh x)\,L_k(x)\,\rd x\,,
 \qquad
 \Cint(k)=-\frac{k^2\zeta(3)}{16\pi}+\frac{2k}{\pi}\mathfrak I_k\,.
\label{eq:B-Cint}
\ee
The sphere constant map \eqref{eq:Ak} is
\be
 A(k)=\frac{2\zeta(3)}{\pi^2k}\left(1-\frac{k^3}{16}\right)
      +\frac{k^2}{\pi^2}\mathscr J_k\,.
\label{eq:B-A}
\ee
Near the origin, the integrands defining $\mathfrak I_k$ and $\mathscr J_k$ are
$O(1+|\log x|)$, while at infinity they are respectively
$O(x\e^{-kx})$ and $O(x\e^{-(k+2)x})$.  Hence both integrals converge
absolutely. These estimates hold uniformly for $k$ in compact subsets of $(0,\infty)$,
which justifies differentiation under the integral sign below.

\begin{lemma}[Integration by parts]
\label{lem:B-ibp}
For every $k>0$,
\be
 \mathscr J_k=-\frac1k\int_0^\infty
 \left[\log(1-\e^{-2x})+\frac{2x}{\e^{2x}-1}\right]L_k(x)\,\rd x\,,
\label{eq:B-Jk-ibp}
\ee
and
\be
 \mathscr J_{k/2}=-\frac4k\int_0^\infty
 \left[\log(1-\e^{-4x})+\frac{4x}{\e^{4x}-1}\right]L_k(x)\,\rd x\,.
\label{eq:B-Jhalf-ibp}
\ee
Moreover, with
\be
 \cK(x)=6x\tanh x+4\log\frac{\sinh x}{x}
          +4\left(\frac{x}{\tanh x}-1\right),
\label{eq:B-kernel}
\ee
one has the pointwise identity
\begin{equation}
\begin{aligned}
 &-6\log(2\cosh x)
 -2\left[\log(1-\e^{-2x})+\frac{2x}{\e^{2x}-1}\right]
 +6\left[\log(1-\e^{-4x})+\frac{4x}{\e^{4x}-1}\right]
\\
 &\hspace{35mm}=\cK(x)-20x+4\log x+4\log2+4\,,
 \qquad x>0\,.
\label{eq:B-kernel-decomp}
\end{aligned}
\end{equation}
\end{lemma}

\begin{proof}
Since $L_k'(x)=k/(\e^{kx}-1)$,
\begin{equation*}
 \mathscr J_k=\frac1k\int_0^\infty x\log(1-\e^{-2x})L_k'(x)\,\rd x\,.
\end{equation*}
The boundary term in the integration by parts vanishes. It is
$O(x\log^2x)$ as $x\to0^+$ and $O(x\e^{-(k+2)x})$ as $x\to\infty$.
Differentiating $x\log(1-\e^{-2x})$ gives~\eqref{eq:B-Jk-ibp}.  For
$\mathscr J_{k/2}$, first set $x=2y$ and then apply the same integration by parts.
This yields~\eqref{eq:B-Jhalf-ibp}.

For~\eqref{eq:B-kernel-decomp}, use
\begin{equation*}
 \log(2\cosh x)=x+\log(1+\e^{-2x})\,,\qquad
 \log(1-\e^{-4x})=\log(1-\e^{-2x})+\log(1+\e^{-2x})\,,
\end{equation*}
together with
\begin{equation*}
 x\coth x=x+\frac{2x}{\e^{2x}-1}\,,\qquad
 x\tanh x=x-\frac{2x}{\e^{2x}-1}+\frac{4x}{\e^{4x}-1}\,.
\end{equation*}
After cancellation of the $\log(1+\e^{-2x})$ terms, the two sides reduce to
the same expression.
\end{proof}

\begin{lemma}[Three moments and a zeta identity]
\label{lem:B-moments}
For every $k>0$,
\be
 \int_0^\infty L_k(x)\,\rd x=-\frac{\zeta(2)}{k}\,,
 \qquad
 \int_0^\infty xL_k(x)\,\rd x=-\frac{\zeta(3)}{k^2}\,,
\label{eq:B-moments12}
\ee
and
\be
 \int_0^\infty \log x\,L_k(x)\,\rd x
 =\frac1k\bigl[\gamma\zeta(2)-\zeta'(2)+\zeta(2)\log k\bigr]\,.
\label{eq:B-momentlog}
\ee
In addition,
\be
 \frac{\zeta'(2)}{\zeta(2)}
 =\log(2\pi)-1+\gamma+12\zeta'(-1)\,.
\label{eq:B-zeta-derivative}
\ee
\end{lemma}

\begin{proof}
Expand $L_k(x)=-\sum_{n\ge1}\e^{-knx}/n$.  Absolute convergence justifies
termwise integration in the first two moments and gives
\eqref{eq:B-moments12}.  For the logarithmic moment, set $u=kx$ and use
\begin{equation*}
 \int_0^\infty \e^{-nu}\log u\,\rd u=-\frac{\gamma+\log n}{n}\,.
\end{equation*}
The resulting series is
$k^{-1}[(\gamma+\log k)\zeta(2)+\sum_{n\ge1}(\log n)/n^2]$. Since
$\sum_{n\ge1}(\log n)/n^2=-\zeta'(2)$, this proves
\eqref{eq:B-momentlog}.

For~\eqref{eq:B-zeta-derivative}, logarithmically differentiate the functional
equation
\begin{equation*}
 \zeta(s)=2^s\pi^{s-1}\sin\frac{\pi s}{2}\,\Gamma(1-s)\zeta(1-s)
\end{equation*}
and evaluate at $s=-1$.  Using $\cot(-\pi/2)=0$, $\psi(2)=1-\gamma$, and
$\zeta(-1)=-1/12$ gives
\begin{equation*}
 -12\zeta'(-1)=\log(2\pi)-1+\gamma-\frac{\zeta'(2)}{\zeta(2)}\,,
\end{equation*}
which is equivalent to~\eqref{eq:B-zeta-derivative}.
\end{proof}

\begin{proposition}[Constant-map differential identity]
\label{prop:CA}
For every real $k>0$,
\be
 \frac{\rd}{\rd k}\left[\frac{\Cint(k)}{k}\right]
 =\frac{2\pi}{k^2}\left[A\!\left(\frac{k}{2}\right)-A(k)\right]
  -\frac{\zeta(3)}{4\pi}\,.
\label{eq:CA}
\ee
\end{proposition}

\begin{proof}
Differentiating~\eqref{eq:B-Cint} under the integral sign gives
\be
 \frac{\rd}{\rd k}\left[\frac{\Cint(k)}{k}\right]
 =-\frac{\zeta(3)}{16\pi}
  +\frac{2}{\pi}\int_0^\infty
   \frac{x\log(2\cosh x)}{\e^{kx}-1}\,\rd x\,.
\label{eq:B-C-derivative}
\ee
Write $\log(2\cosh x)=x+\log(1+\e^{-2x})$.  The first term in the
integral is $2\zeta(3)/k^3$.  For the second, use
$\log(1+\e^{-2x})=\log(1-\e^{-4x})-\log(1-\e^{-2x})$ and set $x=2y$
in the first contribution.  Then
\be
 \int_0^\infty\frac{x\log(1+\e^{-2x})}{\e^{kx}-1}\,\rd x
 =\frac14\mathscr J_{k/2}-\mathscr J_k\,.
\label{eq:B-Jdiff}
\ee
Solving~\eqref{eq:B-A} for $\mathscr J_k$ and $\mathscr J_{k/2}$ and substituting into
\eqref{eq:B-C-derivative} gives~\eqref{eq:CA} after cancellation of the
explicit $k^{-3}$ terms.
\end{proof}

\begin{theorem}[Equivalence of the two index constant-map representations]
\label{thm:f0}
For every real $k>0$,
\be
\begin{aligned}
 &2A(k)-6A\!\left(\frac{k}{2}\right)-\frac{3\Cint(k)}{\pi}
 -\frac{k^2\zeta(3)}{2\pi^2}+\frac12\log k-\frac52\log2
\\[2pt]
 &\quad=-\frac{3k^2\zeta(3)}{8\pi^2}+\frac76\log k+f_0
 +\frac{k}{\pi^2}\int_0^\infty\cK(x)L_k(x)\,\rd x\,,
\label{eq:B-f0-identity}
\end{aligned}
\ee
where
\be
 f_0=-8\zeta'(-1)-\frac{23}{6}\log2-\frac23\log\pi\,.
\label{eq:B-f0}
\ee
\end{theorem}

\begin{proof}
Substitute~\eqref{eq:B-Cint} and~\eqref{eq:B-A} into the left-hand side of
\eqref{eq:B-f0-identity}.  The explicit $k^2\zeta(3)$ terms combine to
$-3k^2\zeta(3)/(8\pi^2)$, while the remaining inverse-level term is
$-20\zeta(3)/(\pi^2k)$.  The integral contribution is
\be
 -\frac{6k}{\pi^2}\mathfrak I_k+\frac{2k^2}{\pi^2}\mathscr J_k
 -\frac{3k^2}{2\pi^2}\mathscr J_{k/2}\,.
\label{eq:B-integral-combination}
\ee
Apply~\eqref{eq:B-Jk-ibp} and~\eqref{eq:B-Jhalf-ibp}.  Lemma~\ref{lem:B-ibp}
then rewrites~\eqref{eq:B-integral-combination} as
\be
 \frac{k}{\pi^2}\int_0^\infty
 \bigl[\cK(x)-20x+4\log x+4\log2+4\bigr]L_k(x)\,\rd x\,.
\label{eq:B-after-decomp}
\ee
The $\cK$ term is the integral on the right-hand side of
\eqref{eq:B-f0-identity}.  By~\eqref{eq:B-moments12}, the $-20x$ term gives
$20\zeta(3)/(\pi^2k)$ and cancels the inverse-level contribution.  The
remaining moments, together with the explicit logarithms in
\eqref{eq:B-f0-identity}, reduce to
\be
\frac76\log k
-\frac{19}{6}\log2-\frac23+\frac{2\gamma}{3}
-\frac{4\zeta'(2)}{\pi^2}\,.
\label{eq:B-remaining}
\ee
Finally, $\zeta(2)=\pi^2/6$ and~\eqref{eq:B-zeta-derivative} give
\be
-\frac{4\zeta'(2)}{\pi^2}
=
-\frac23\bigl[\log(2\pi)-1+\gamma+12\zeta'(-1)\bigr]\,.
\ee
Substituting this into~\eqref{eq:B-remaining}, the Euler--Mascheroni and
rational terms cancel, leaving
\be
\frac76\log k
-8\zeta'(-1)-\frac{23}{6}\log2-\frac23\log\pi
=
\frac76\log k+f_0\,.
\ee
This proves \eqref{eq:B-f0-identity}.
\end{proof}

Theorem~\ref{thm:f0} establishes analytically that the arithmetic
expression~\eqref{eq:f0arith} and the one-kernel representation
\eqref{eq:f0kernel} define the same function. Its identification with the rank-independent
term of the topologically twisted index is a separate finite-data result.

\section{Reconstructed instanton sectors}
\label{app:tables}

Table~\ref{tab:Wsectors} lists the leading exactly reconstructed sectors of the
twisted superpotential, and Table~\ref{tab:Zsectors} lists the complete set for
the index. Every entry is a rational recovered by the peel of
Section~\ref{sec:peel} under the stated height caps and verified on ranks
withheld from its own solve. In total $51$, $67$ and $93$ twisted superpotential
sectors and $17$, $15$ and $15$ index sectors were reconstructed at $k=1,2,4$.
The word ``exactly'' refers to those finite ranges throughout.

\begin{table}[H]
\centering\small
\renewcommand{\arraystretch}{1.6}
\begin{tabular}{@{}crrr@{}}
\toprule
$m$ & $c^{(k=1,2)}_m$ & $c^{(k=4)}_m$ & $\sigma_3(m)$\\
\midrule
$1$ & $5$ & $1$ & $1$\\
$2$ & $-\frac{25}{4}$ & $-\frac{9}{4}$ & $9$\\
$3$ & $\frac{140}{9}$ & $\frac{28}{9}$ & $28$\\
$4$ & $-\frac{185}{16}$ & $-\frac{57}{16}$ & $73$\\
$5$ & $\frac{126}{5}$ & $\frac{126}{25}$ & $126$\\
$6$ & $-\frac{175}{9}$ & $-7$ & $252$\\
$7$ & $\frac{1720}{49}$ & $\frac{344}{49}$ & $344$\\
$8$ & $-\frac{1465}{64}$ & $-\frac{441}{64}$ & $585$\\
$9$ & $\frac{3785}{81}$ & $\frac{757}{81}$ & $757$\\
$10$ & $-\frac{63}{2}$ & $-\frac{567}{50}$ & $1134$\\
$11$ & $\frac{6660}{121}$ & $\frac{1332}{121}$ & $1332$\\
$12$ & $-\frac{1295}{36}$ & $-\frac{133}{12}$ & $2044$\\
\bottomrule
\end{tabular}
\caption{The first twelve rationally reconstructed twisted superpotential
instanton coefficients at $k=1,2$ and $k=4$, with $\sigma_3(m)$ for reference.
The reconstruction uses only the features $q^m$ and $tq^m$. The divisor
structure of Eq.~\eqref{eq:cm} is recognized afterwards. In total $51$, $67$
and $93$ sectors were reconstructed at $k=1,2,4$, respectively, and every reconstructed
coefficient agrees with Eq.~\eqref{eq:cm}. The table does not assume the
all-order continuation.}
\label{tab:Wsectors}
\end{table}

\begin{table}[p]
\centering\scriptsize
\renewcommand{\arraystretch}{1.6}
\setlength{\tabcolsep}{4pt}
\begin{tabular}{@{}crrr@{}}
\toprule
$m$ & $A_m^{(k)}$ & $B_m^{(k)}$ & $R_m^{(k)}$\\
\multicolumn{4}{@{}l}{\textit{$k=1,2$}}\\[1pt]
\midrule
$1$ & $-\frac{13}{2}$ & $-10$ & $76$\\
$2$ & $\frac{1441}{4}$ & $316$ & $-3590$\\
$3$ & $-\frac{66170}{3}$ & $-\frac{121672}{9}$ & $\frac{614704}{3}$\\
$4$ & $\frac{11674881}{8}$ & $686328$ & $-13062147$\\
$5$ & $-\frac{507532519}{5}$ & $-\frac{193464652}{5}$ & $\frac{4442194376}{5}$\\
$6$ & $\frac{21900100129}{3}$ & $\frac{21047832496}{9}$ & $-\frac{188837083160}{3}$\\
$7$ & $-\frac{3763901747700}{7}$ & $-\frac{7277987585424}{49}$ & $\frac{32109741261408}{7}$\\
$8$ & $\frac{645167023302657}{16}$ & $9789148166640$ & $-\frac{682490104578051}{2}$\\
$9$ & $-\frac{55195408339837609}{18}$ & $-\frac{53780562166076290}{81}$ & $\frac{232100402749440988}{9}$\\
$10$ & $\frac{2358099684533306083}{10}$ & $\frac{1151915216436419816}{25}$ & $-1973310550283391188$\\
$11$ & $-\frac{201306256815651241838}{11}$ & $-\frac{394205943489634941400}{121}$ & $\frac{1677702615658550272912}{11}$\\
$12$ & $\frac{2862236075305752587691}{2}$ & $\frac{701886015066561866528}{3}$ & $-11886480421539051323396$\\
$13$ & $-112638085513255794287815$ & $-\frac{2877007920943172422715852}{169}$ & $\frac{12127019312590428436071464}{13}$\\
$14$ & $\frac{62401025026469493483937858}{7}$ & $\frac{61381678604629673275607392}{49}$ & $-\frac{515517742961193071204790320}{7}$\\
$15$ & $-\frac{2126678070346231912113895724}{3}$ & $-\frac{20942880708014400477179498672}{225}$ & $\frac{87658322777432625636231497504}{15}$\\
$16$ & $\frac{1811488287532412773581618757633}{32}$ & $6975183380895560666021655520$ & $-\frac{1863170971578711259741768646659}{4}$\\
$17$ & $-\frac{77133592201343594425478157513845}{17}$ & $-\frac{152197662527856329154788354243764}{289}$ & $\frac{633624912609395632973168946382168}{17}$\\
\midrule
\multicolumn{4}{@{}l}{\textit{$k=4$}}\\[1pt]
\midrule
$1$ & $-\frac{1}{2}$ & $-1$ & $20$\\
$2$ & $\frac{49}{4}$ & $10$ & $-262$\\
$3$ & $-\frac{626}{3}$ & $-\frac{1108}{9}$ & $\frac{12368}{3}$\\
$4$ & $\frac{30593}{8}$ & $1748$ & $-71683$\\
$5$ & $-\frac{365283}{5}$ & $-\frac{681126}{25}$ & $1327128$\\
$6$ & $\frac{4316401}{3}$ & $\frac{4075336}{9}$ & $-\frac{76776472}{3}$\\
$7$ & $-\frac{202803556}{7}$ & $-\frac{386328104}{49}$ & $\frac{3553321120}{7}$\\
$8$ & $\frac{9492913153}{16}$ & $142191016$ & $-\frac{20556578819}{2}$\\
$9$ & $-\frac{221617556749}{18}$ & $-\frac{213502966669}{81}$ & $\frac{1902773895428}{9}$\\
$10$ & $\frac{2582349976467}{10}$ & $\frac{249754040652}{5}$ & $-\frac{22015752930852}{5}$\\
$11$ & $-\frac{60103529648982}{11}$ & $-\frac{116628117491964}{121}$ & $\frac{1018919543279856}{11}$\\
$12$ & $\frac{698775189536129}{6}$ & $\frac{169937164948880}{9}$ & $-\frac{5894625992138764}{3}$\\
$13$ & $-\frac{32472349215018343}{13}$ & $-\frac{63315615577732142}{169}$ & $\frac{545622913077166360}{13}$\\
$14$ & $53861615323426958$ & $\frac{368264271444100496}{49}$ & $-\frac{6313046062460241968}{7}$\\
$15$ & $-\frac{5834152171606660452}{5}$ & $-\frac{11415445743654570536}{75}$ & $\frac{97392172338646573216}{5}$\\
\bottomrule
\end{tabular}
\caption{Complete set of rationally reconstructed twisted-index instanton
sectors, in the normalization of Eq.~\eqref{eq:Znp}: seventeen sectors at
$k=1$ and fifteen at $k=4$. The $k=1$ and $k=2$ sequences are identical on all
fifteen orders reconstructed at both levels. The two deepest rows of the upper
panel are resolved at $k=1$ only. Every sector listed satisfies
relation~\eqref{eq:Brel} exactly over $\mathbb Q$, and every $R_m^{(k)}$ is reproduced,
through the listed orders, by the candidate product~\eqref{eq:Rmodular}.
Coefficient heights grow rapidly. The numerator of $A_{17}^{(1)}$ has $32$
digits, which limits the reconstruction depth.}
\label{tab:Zsectors}
\end{table}

\bibliographystyle{ytamsalpha}
\def\arxivfont{\rm}
\baselineskip=.95\baselineskip
\bibliography{tti_cnst_exp}

\end{document}